\documentclass[paper=a4, fontsize=11pt]{scrartcl}
\usepackage[T1]{fontenc}
\usepackage{fourier}
\usepackage[english]{babel}								 
\usepackage{amssymb}
\usepackage{amsmath,amsfonts,amsthm} % Math packages
\usepackage[pdftex]{graphicx}	
\usepackage{url}
\usepackage[title,titletoc,page]{appendix} 
\usepackage{listings} 
\usepackage{color}
\usepackage{cite}
\usepackage{caption}
\usepackage{subcaption}
\usepackage[ruled,vlined]{algorithm2e}
\usepackage{float}
\usepackage{booktabs}

\numberwithin{equation}{section}		% Equationnumbering: section.eq#
\numberwithin{figure}{section}			% Figurenumbering: section.fig#
\numberwithin{table}{section}				% Tablenumbering: section.tab#

\newtheorem{defi}{Definition}[section]
\newtheorem{lem}{Lemma}[section]
\newtheorem{thm}{Theorem}[section]

\newtheorem{remark}{Remark}[section]

\newtheorem{prop}{Proposition}[section]

\newcommand{\babs}[1]{\Big|{#1}\Big|}

\newcommand{\btwonorm}[1]{\Big|\Big|{#1}\Big|\Big|_2}

\newcommand{\vect}[1]{\boldsymbol{\mathbf{#1}}}

\title{Dynamic super-resolution in particle tracking problems \thanks{\footnotesize This work was supported in part by the Swiss National Science Foundation grant number
200021--200307.}}
	\author{Ping Liu\thanks{\footnotesize Department of Mathematics, ETH Z\"urich, R\"amistrasse 101, CH-8092 Z\"urich, Switzerland (ping.liu@sam.math.ethz.ch, habib.ammari@math.ethz.ch).}  \and Habib Ammari\footnotemark[2]}

\date{}
\begin{document}

	\maketitle	
	
\begin{abstract}
Particle tracking in a live cell environment is concerned with reconstructing the trajectories, locations, or velocities of the targeting particles, which holds the promise of revealing important new biological insights. The standard approach of particle tracking consists of two steps: first reconstructing statically the source locations in each time step, and second applying tracking techniques to obtain the trajectories and velocities. In contrast to the standard approach, the dynamic reconstruction seeks to simultaneously recover the source locations and velocities from all frames, which enjoys certain advantages. In this paper, we provide a rigorous mathematical analysis for the resolution limit of reconstructing source number, locations, and velocities by general dynamical reconstruction in particle tracking problems, by which we demonstrate the possibility of achieving  super-resolution for the dynamic reconstruction. We show that when the location-velocity pairs of the particles are separated beyond certain distances (the resolution limits), the number of particles and the location-velocity pair can be stably recovered. The resolution limits are related to the cut-off frequency of the imaging system, signal-to-noise ratio, and the sparsity of the source. By these estimates we also derive a stability result for a sparsity-promoting dynamic reconstruction. In addition, we further show that the reconstruction of velocities has a better resolution limit which improves constantly as the particles moving. This result is derived by a crucial observation that the inherent cut-off frequency for the velocity recovery can be viewed as the total observation time multiplies the cut-off frequency of the imaging system, which may lead to a better resolution limit as compared to the one for each diffraction-limited frame. It is anticipated that this crucial observation can inspire many new reconstruction algorithms that significantly improve the resolution of particle tracking in practice. In addition, we propose super-resolution algorithms for recovering the number and values of the velocities in the tracking problem and demonstrate theoretically or numerically their super-resolution capability. 
\end{abstract}

\section{Introduction}
The problem of particle tracking is of particular importance in many modern imaging experiments such as visualization of cell migration and subcellular dynamics \cite{meijering2009tracking, kwik2003membrane, zenisek2000transport, ram2008high}. It is also utilized in ultrafast ultrasound localization microscopy to super-resolve structures of vascular and velocities of blood flow \cite{errico2015ultrafast, demene2021transcranial}.  The conventional tracking methods in above applications generally consist of two main image processing steps: 1) particle detection, and 2) particle linking. The former refers to determining the locations of particles in each frame and the later refers to linking particles between consecutive frames. However, this static reconstruction strategy suffers from two major issues. The first key issue is that the particles are not well-separated in some of the frames which results in unstable recovery of the source locations or abandon of a large amount of frames where particles are closely-spaced. The second issue is the ambiguities and heavy computational burden in linking particles with high densities and velocities \cite{meijering2008time, schiffmann2006open, meijering2009tracking}. To fix these issues, in \cite{alberti2019dynamic}, the authors consider a model where the particles move with constant velocities (in a short period) and propose a new dynamic reconstruction method to simultaneously recover the source locations and velocities from measurements of all the frames. However, theoretically the algorithm still requires that the point sources should be separated beyond Rayleigh limit in each frame for a stable recovery, which hampered its application in the case when sources are closely-spaced in some of the frames. In this paper, in order to understand the super-resolution capability of the tracking problem, we aim to analyze the resolution limit for super-resolving the source number, location-velocity pairs, and velocities in the tracking problem . We also provide super-resolution algorithms for reconstructing the velocities of moving particles which are able to super-resolve velocities even when static reconstruction completely fails.

\subsection{Our model and contributions}
Let us first introduce our model and main contributions. We consider the same model for reconstructing dynamic point sources (or particles) as the one in \cite{alberti2019dynamic}. We remark that we will use bold symbols for vectors, matrices and some functions, and ordinary ones for scalar values in our models and discussions throughout the paper. To be specific, we consider a cluster of point sources that are moving constantly, represented by a time-varying measures $\mu_t$, where $t\in [0,\eta]$ with $\eta>0$ being the observation window. Since $\eta$ is expected to be small, the dynamics of each point can be linearly approximated. Thus each point is modeled as a particle moving with a constant velocity:
\[
\mu_t = \sum_{j=1}^n a_j \delta_{\vect y_j + \vect v_j \tau t}, \quad t=0,\cdots, T, 
\]
where $\vect y_j\in\mathbb R^d$ are the initial source locations and $\vect v_j\in \mathbb R^d$ are the velocities. The $\tau t, t=0,\cdots,T$ is the time step in $[0, \eta]$ at which we observe the samples. We remark that in the applications of measuring the velocity of blood flow \cite{errico2015ultrafast,demene2021transcranial} and particle tracking velocimetry \cite{dabiri2020particle}, we may have a larger observation window $\eta$ and our model is much more applicable. The measurement at each time step is a noisy Fourier data in a bounded domain, 
\begin{equation}\label{equ:modelsetting1}
\vect{Y}_t(\vect \omega) = \sum_{j=1}^{n} a_j e^{i(\vect y_j^\top +t\tau \vect v_j^\top)\vect \omega} +\vect{W}_t(\vect \omega), \ \vect \omega \in \mathbb R^d,  ||\vect \omega||_2 \leq \Omega, \ t=0, \cdots,T,
\end{equation}
where $\top$ denotes the transpose, $\Omega$ is the cut-off frequency of the imaging system, $\vect W_t$ is an additive noise with $|\vect{W}_t(\vect \omega)|<\sigma$ and $\sigma$ being the noise level. We denote 
\[
m_{\min}:= \min_{j=1, \cdots, n}|a_j|.
\]
For the tracking problem in $\mathbb R^d$, we are interested in reconstructing the amplitude, location, and velocity pair set 
\[
\Big\{(a_j, \vect y_j, \tau \vect v_j)\Big \}_{j=1}^n,
\]
from the measurements $\vect Y_t$'s. 

In \cite{alberti2019dynamic}, the authors proposed a fully dynamical method to reconstruct simultaneously the source locations and velocities. They cast the inverse problem as a total variation optimization problem satisfying the measurement constraints and showed that under certain conditions, the optimization is able to reconstruct the amplitudes, locations and velocities of source from noiseless measurements with infinite precision. However, for the noisy case, theoretically, to stably recover the source, the point sources are required to be separated beyond Rayleigh limit in each frame, which is definitely not super-resolution. Thus the superiority of the dynamic reconstruction over the static recovery in the tracking problem is still unclear. Since when the spikes are separated beyond Rayleigh limit, the static reconstruction is also able to recover stably all the $\mu_t$'s.

In this paper, in contrast to \cite{alberti2019dynamic}, we consider the locations and velocities of point sources to be closely-spaced in sub-Rayleigh regimes respectively, and explore the ability of super-resolution of the dynamic reconstruction in the tracking problem. More precisely, we analyze the resolution limit for recovering the number, locations, and velocities of the source from the measurement constraint. We prove that when 
\[
\min_{j\neq p} \btwonorm{(\vect y_j^\top,\tau \vect v_j^\top)-(\vect y_p^\top, \tau \vect v_p^\top)} \geq \frac{C_1(d, n)}{\Omega}\Big(\frac{\sigma}{m_{\min}}\Big)^{\frac{1}{2n-2}}, 
\]
where $C_1(d,n)$ is an explicit constant only related to space dimensionality $d$ and source number $n$ (so do the following constants $C_2(d,n), C_3(d,n), C_4(d,n)$), then we can recover the source number $n$ in the dynamic reconstruction. When 
\[
\min_{j\neq p} \btwonorm{(\vect y_j^\top,\tau \vect v_j^\top)-(\vect y_p^\top, \tau \vect v_p^\top)} \geq \frac{C_2(d, n)}{\Omega}\Big(\frac{\sigma}{m_{\min}}\Big)^{\frac{1}{2n-1}}, 
\]
then we can stably recover the source locations and velocities. The estimate demonstrates that the dynamic reconstruction could resolve the location-velocity pairs in sub-Rayleigh regimes for sufficiently large signal-to-noise ratio. It also indicates that the dynamic reconstruction could recover source locations and velocities even when static reconstruction fails in some of the frames, which demonstrates its superiority. 

Moreover, by above results, we also show that any algorithms targeting at the sparsest solution satisfying measurement constraint could achieve the resolution $\frac{C_2(d, n)}{\Omega}\Big(\frac{\sigma}{m_{\min}}\Big)^{\frac{1}{2n-1}}$, which demonstrates the favorable performance and the superiority of sparsity-promoting dynamic reconstruction as compared to the static reconstruction methods. This is consistent with the numerical results in \cite{alberti2019dynamic} for comparison between reconstructions by TV minimization and static method.

In some applications, the recovery of the velocities of the particles is more interesting than the location recovery. For instance, in the ultrafast ultrasound localization microscopy based non-invasive super-resolution vascular imaging \cite{errico2015ultrafast, demene2021transcranial}, the velocities of the blood flow with a large dynamic range can be extracted by reconstructing the velocities of microbubbles inside. We also note that the particle tracking is a common method in velocimetry \cite{dabiri2020particle} which seeks to measure accurate velocities of fluid. We thus further consider the reconstruction of individual velocities. By some projection methods, we demonstrate super-resolution for resolving the number of point sources and value of velocities. More precisely, when the minimum separation distance of velocities satisfies
\[
\min_{j\neq p}\btwonorm{\tau \vect v_j-\tau \vect v_p}\geq \frac{C_3(d,n)}{T\Omega}\Big(\frac{\sigma}{m_{\min}}\Big)^{\frac{1}{2n-2}},
\]
we can stably recover the number of point sources. When
\[
\min_{j\neq p}\btwonorm{\tau \vect v_j-\tau \vect v_p}\geq \frac{C_4(d,n)}{T\Omega}\Big(\frac{\sigma}{m_{\min}}\Big)^{\frac{1}{2n-1}},
\]
we can stably reconstruct all the velocities. The results are derived by a crucial observation that the inherent cut-off frequency for recovering velocities $\tau \vect v_j$'s can be viewed as $T\Omega$, by which we obtain better resolution limits. We hope this observation could inspire new algorithms for super-resolving velocities in the tracking problem. For fixed dimensionality $d$ and source number $n$, all the above theoretical results for the resolution limits are optimal in the sense that at the same separation order there are counter examples in the worst-case scenario.

We also propose super-resolution algorithms for resolving the source number and velocities respectively. They are demonstrated theoretically or numerically to achieve the same order of the resolution limits that derived above. We also demonstrate numerically their superiority over the static reconstruction. 

%It is anticipated that our algorithm could be applied to analyze the velocity of blood flow \cite{errico2015ultrafast} and other problems in particle tracking velocimetry \cite{dabiri2020particle}. 

\subsection{Related works}
The main motivations of our work are ultrasound localization microscopy (ULM) \cite{couture2011microbubble, desailly2013sono,errico2015ultrafast,demene2021transcranial, couture2018ultrasound} and the dynamic reconstruction algorithm proposed in \cite{alberti2019dynamic}. The concept of ULM is drawing increasing research interest in recent years, due to its ability to solve the trade-off between spatial resolution, penetration depth, and acquisition time when incorporating with ultrasound contrast agents and ultrafast imaging \cite{errico2015ultrafast,demene2021transcranial}, by which the conventional ultrasound imaging and other non-invasive medical imaging methods for living organs are usually limited. Specifically, based on the diffraction theory, the resolution of conventional ultrasound imaging is limited by half of the wavelength of the ultrasound waves which is of the order of $300\mu m$. Thus the blood vessels separated below $300 \mu m$ is unable to be distinguished by the conventional ultrasound imaging. But the ultrasound localization microscopy focus on localizing some targeting particles (such as microbubbles) inside the blood vessels, giving rise to a particle tracking problem \cite{errico2015ultrafast} where both the locations and velocities are of interest. By this innovative method, tenfold increase in resolution is achieved as compared to conventional ultrasound imaging. Moreover, the ultrafast ultrasound imaging techniques also significantly reduce the acquisition time of ULM imaging \cite{errico2015ultrafast}. 

However, the tracking approach in \cite{errico2015ultrafast} relies on static reconstruction of point sources in each frame which suffers from some drawbacks. Among them the key issue is that a lot of data are discarded whenever static reconstruction cannot be performed as particles being too close. To solve these issues, in \cite{alberti2019dynamic} the authors proposed a new method for the tracking problem based on a fully dynamical inversion scheme, in which the locations and velocities of the point sources are reconstructed simultaneously. To be specific, assuming the targeting particles are moving at constant velocities and considering similar measurement constraint as (\ref{equ:modelsetting1}), they reconstruct source locations and velocities by a total variation optimization problem. They majorly demonstrate that, for the noiseless measurement, if in more than three time steps the point sources are well-separated to satisfy certain static dual certificates and the configuration does not admit any "ghost particles" defined there, then the TV optimization successfully recover the source locations and velocities. However, for recovering from the noisy measurement, the point sources are required to be separated by more than Rayleigh limit in each frame to ensure a stable recovery of the locations and velocities. Although the super-resolution ability of the algorithm isn't exhibited in the above theoretical result, but it is confirmed by the numerical experiments there. Motivated by this, in this paper we provide a rigorous analysis for the resolution limit of the dynamic reconstruction in the tracking problems and theoretically demonstrate the advantages of dynamic reconstruction over the static reconstruction methods.

On the other hand, the resolution analysis in this paper also follows the line of the authors' previous researches on exploring the super-resolution capability for different imaging configurations \cite{liu2021mathematicaloned, liu2021mathematicalhighd, liu2021theorylse, liu2022mathematical}. Specifically, to analyze the resolution for recovering multiple point sources from a single measurement, in \cite{liu2021mathematicaloned, liu2021mathematicalhighd, liu2021theorylse} the authors defined "computational resolution limits" which characterize the minimum required distance between point sources so that their number and locations can be stably resolved under certain noise level. Based on a new approximation theory in a so-called Vandermonde space, they derived bounds for the resolution limits of one- and multi- dimensional super-resolution problems \cite{liu2021mathematicaloned, liu2021theorylse, liu2021mathematicalhighd}. In particular, they showed that the computational resolution limit for number and location recovery should be respectively $\frac{C_{\mathrm{num}}(d,n)}{\Omega}(\frac{\sigma}{m_{\min}})^{\frac{1}{2n-2}}$ and  $\frac{C_{\mathrm{supp}}(d,n)}{\Omega}(\frac{\sigma}{m_{\min}})^{\frac{1}{2n-1}}$ where $C_{\mathrm{num}}(d,n), C_{\mathrm{supp}(d,n)}$ are certain constants depending on space dimensionality $d$ and source number $n$. In addition to the single measurement case,  the stability for sparsity-promoting super-resolution algorithms in multi-illumination imaging was also derived \cite{liu2022mathematical}. 

There were also many mathematical theories for estimating the stability of super-resolution from the perceptive of minimax error estimation. In  \cite{donoho1992superresolution}, Donoho considered a grid setting where a discrete measure is supported on a lattice (spacing by $\Delta$) and regularized by a so-called  "Rayleigh index" $b$. He demonstrated that the minimax error for the amplitude reconstruction is bounded from below and above by $SRF^{2b-1}\sigma$ and $SRF^{2b+1}\sigma$ respectively with $\sigma$ being the noise level and the super-resolution factor $SRF = 1/({\Omega \Delta})$. More recently, inspired by the huge success of new super-resolution modalities \cite{ref4,STED, ref7,PALM, ref8,ref9} and the popularity of researches for super-resolution algorithms \cite{candes2014towards, azais2015spike, duval2015exact, tang2014near, morgenshtern2016super, morgenshtern2020super, denoyelle2017support}, the study of super-resolution capability of imaging problems also becomes popular in applied mathematics. In \cite{demanet2015recoverability}, the authors considered $n$-sparse point sources supported on a grid and obtained sharper lower and upper bounds (both $SRF^{2n-1}\sigma$) for the minimax error of amplitude recovery than those in \cite{donoho1992superresolution}. The case of multi-clustered point sources was considered in \cite{li2021stable, batenkov2020conditioning} and similar minimax error estimations were derived.  In \cite{akinshin2015accuracy, batenkov2019super}, the authors considered the minimax error for recovering off-the-grid point sources. Based on an analysis of the "prony-type system", they derived bounds for both amplitude and location reconstructions of the point sources. More precisely, they showed that  for $\sigma \lessapprox (SRF)^{-2p+1}$ where $p$ is the number of point sources in a cluster, the minimax error for the amplitude and the location recoveries scale respectively as $(SRF)^{2p-1}\sigma$, $(SRF)^{2p-2} {\sigma}/{\Omega}$, while for the isolated non-clustered source, the corresponding minimax error for the amplitude and the location recoveries scale respectively as $\sigma$ and ${\sigma}/{\Omega}$. We also refer the readers to \cite{moitra2015super,chen2020algorithmic} for understanding the resolution limit from the perceptive of sample complexity.

\subsection{Organization of the paper}
The rest of the paper is organized in the following way. In Section 2, we present the main results about the stability of location-velocity pair reconstruction in the tracking problem. More precisely, in Section 2.1, we state the results for the stability of number recovery of the point sources, and in Section 2.2 those for the stability of the recovery of the location-velocity pairs. In Section 2.3, we present the stability estimate for a sparsity-promoting algorithm in the tracking problem. In Section 3, we demonstrate better resolutions for the number and value reconstruction of velocities. In Section 4 and Section 5, we propose projection-based super-resolution algorithms for the number and value recovery of the velocities respectively. Numerical experiments are also conducted to demonstrate their efficiency.  Section 6 is devoted to some concluding remarks and future works. In Section 7 and Section 8, we prove our main results in this paper.

\section{Main results about the recovery of location-velocity pairs}\label{sec:mainresultoflocationvelocity}
In this section, we present the results for the stability of the dynamic reconstruction in the tracking problem. Suppose that a series of images is generated by $n$ $d$-dimensional point sources located at $\vect y_j$'s with amplitudes $a_j$'s and velocities $\vect v_j$'s. The inverse problem we are concerned with is to reconstruct the number and location-velocity pairs of the point sources. To be more specific, we consider the set of parameters $\{(a_j, \vect y_j, \tau \vect v_j)\}_{j=1}^n$ and denote the vectors of locations and velocities as 
\[
\vect \alpha_{j} = \begin{pmatrix}\vect y_j\\ \tau \vect v_j\end{pmatrix}.
\]
We first consider the stability of reconstructing the number and value of the location-velocity pairs $\vect \alpha_j$'s. Because we are interested in the capability of super-resolution, we consider recovering point sources with close $\vect y_j$ and $\tau \vect v_j$'s. 
We define
\[
B_{\delta}^m(\vect x) := \Big\{ \mathbf y \ \Big|\ \mathbf y\in \mathbb R^m,\  ||\vect y- \vect x||_2<\delta \Big\},
\] 
and assume that $\vect \alpha_j \in B_{\delta}^m(\vect 0), j=1,\cdots,n,$ for certain $\delta >0$. On the other hand, in order to analyze the stability of the reconstruction, below we first need to define a $\sigma$-admissible parameter set of $\vect Y_t$'s. In the following sections, we shall consider recovering the number and value of location-velocity pairs from all the $\sigma$-admissible parameter sets. 

\begin{defi}{\label{def:sigmaadmissibleset}}
Given measurement $\mathbf Y_t$'s in (\ref{equ:modelsetting1}), we say that $\{(\hat a_j, \vect {\hat y}_j, \tau \vect {\hat v}_j)\}_{j=1}^k$ is a $\sigma$-admissible parameter set of $\vect Y_t$'s if
\[
\babs{\sum_{j=1}^k\hat a_j e^{i(\vect {\hat y}_j^\top+t\tau \vect {\hat v}_j^\top) \vect \omega} -\vect Y_t(\vect \omega)}< \sigma, \ ||\vect \omega||_2\leq \Omega, \ \text{for all} \ t=0,1,\cdots,T. 
\]
\end{defi}

\begin{defi}
 We say a parameter set is a $n$-sparse parameter set if it contains exactly $n$ different elements. 
\end{defi}

\subsection{Stability of number recovery in the tracking problem}\label{sec:trackstabilityofnumberrecovery}
In this subsection, we present our main results for recovering the number of point sources in the tracking problem.  Define
\begin{equation}\label{equ:defineofxi}
\xi(k)= \begin{cases}
\sum_{j=1}^{k}\frac{1}{j}, & \ k\geq 1,\\
0,& \ k=0.
\end{cases}
\end{equation}
We have the following theorem for the recovery of source number in the $d$-dimensional tracking problem. Its proof is given is Section \ref{sec:proofmainresults}. 
\begin{thm}\label{thm:trackhighdupperboundnumberlimit0}
	Let $n\geq 2$ and $T\geq \frac{n(n-1)}{2}$. Let the measurement $\mathbf Y_t\in \mathbb R^d, t=0, \cdots, T$, in (\ref{equ:modelsetting1}) be generated by a $n$-sparse parameter set $\{(a_j, \vect y_j, \tau \vect v_j)\}_{j=1}^n$ with $\vect \alpha_j:= \begin{pmatrix}
\vect y_j\\
\tau \vect v_j
\end{pmatrix} \in B_{\frac{\pi}{(n+1)\Omega}}^{2d}(\vect 0)$. Assume that the following separation condition is satisfied 
	\begin{equation}\label{equ:trackhighdupperboundnumberlimit0equ1}
	\min_{p\neq j, 1\leq p, j\leq n}\btwonorm{\vect \alpha_p- \vect \alpha_j}\geq \frac{8.8e\pi^2\sqrt{(\frac{n(n-1)}{2})^2+1}(\pi/2)^{d-1} \Big(\frac{n(n-1)}{\pi}\Big)^{\xi(d-1)}}{\Omega }\Big(\frac{\sigma}{m_{\min}}\Big)^{\frac{1}{2n-2}},
	\end{equation}
	where $\xi(d-1)$ is defined as in (\ref{equ:defineofxi}). Then there does not exist any $\sigma$-admissible parameter set of \,$\mathbf Y_t$'s with less than $n$ elements.
\end{thm}

\begin{remark}
We remark that the condition $\vect \alpha_j\in B_{\frac{\pi}{(n+1)\Omega}}^{2d}(\vect 0)$ in the above theorem and similar conditions in the rest of the paper can be straightforwardly extended to $\vect \alpha_j\in B_{\frac{\pi}{(n+1)\Omega}}^{2d}(\vect x)$ for a non-zero vector $\vect x$. Also, they can be easily extended to $\vect \alpha_j\in B_{\delta}^{2d}(\vect 0)$ for a larger $\delta$ with a slight modification of the results.   
\end{remark}

Theorem \ref{thm:trackhighdupperboundnumberlimit0} reveals that, for fixed space dimensionality $d$ and source number $n$, when \\ 
$\min_{j\neq p} \btwonorm{\begin{pmatrix} \vect y_j\\ \tau \vect v_j\end{pmatrix} - \begin{pmatrix}\vect y_p\\ \tau \vect v_p\end{pmatrix}} \geq O\Big(\frac{1}{\Omega}(\frac{\sigma}{m_{\min}})^{\frac{1}{2n-2}}\Big)$, it is possible to reconstruct the exact source number in the tracking problem. We remark that since the  space dimensionality of interest is usually small ($d=1,2,3$), the amplification factor in (\ref{equ:trackhighdupperboundnumberlimit0equ1}) due to the space dimensionality is not large.

We next show that in the worst-case scenario, the order $O(\frac{1}{\Omega}(\frac{\sigma}{m_{\min}})^{\frac{1}{2n-2}})$ is optimal without further information on the velocities. This is shown by Proposition \ref{prop:trackhighdnumberlowerbound0}. We first present a result for the $d$-dimensional static super-resolution problem which helps to derive Propositions \ref{prop:trackhighdnumberlowerbound0} and \ref{prop:trackhighdnumberlowerbound1}. 
\begin{prop}\label{prop:staticnumberlowerboundthm0}
	For given $0<\sigma<m_{\min}$ and integer $n\geq 2$, there exist $ a_j\in \mathbb C, \vect y_j \in \mathbb R^d, j=1, \cdots, n$, and $\hat a_j \in \mathbb C, \vect {\hat y}_j \in \mathbb R^d, j=1, \cdots, n-1$ such that 
	\[
	\babs{\sum_{j=1}^{n-1} \hat a_j e^{i \vect {\hat y}_j^\top \vect \omega}- \sum_{j=1}^{n} a_j e^{i \vect y_j^\top \vect \omega}}<\sigma, \quad  ||\vect \omega||_2 \leq \Omega. 
	\]
Moreover,
	\[
	\min_{1\leq j\leq n}|a_j|= m_{\min}, \quad \min_{p\neq j}\btwonorm{\vect y_p-\vect y_j}= \frac{0.81e^{-\frac{3}{2}}}{\Omega}\Big(\frac{\sigma}{m_{\min}}\Big)^{\frac{1}{2n-2}}.
	\]	
\end{prop}
\begin{proof}
See Proposition 2.4 in \cite{liu2021mathematicalhighd}. 
\end{proof}

\begin{prop}\label{prop:trackhighdnumberlowerbound0}
	For given $0<\sigma<m_{\min}$ and integer $n\geq 2$, there exist a $n$-sparse parameter set $\{(a_j, \vect y_j, \tau \vect v_j)\}_{j=1}^n$, $\vect y_j\in \mathbb R^d, \vect v_j \in \mathbb R^d$, and a $(n-1)$-sparse parameter set $\{(\hat a_j, \vect {\hat y}_j, \tau \vect {\hat v}_j)\}_{j=1}^{n-1}$, $\vect {\hat y}_j\in \mathbb R^d, \vect {\hat v}_j \in \mathbb R^d$, such that 
	\[
	\babs{\sum_{j=1}^{n-1}\hat a_j e^{i(\vect {\hat y}_j^\top + t \tau \vect {\hat v}_j^\top)\vect \omega} -\sum_{j=1}^{n}a_j e^{i(\vect y_j^\top + t \tau \vect v_j^\top)\vect \omega}} < \sigma, \quad  ||\vect \omega||_2\leq \Omega, \  t=0, \cdots,T.
	\]
	Moreover,
	\[
	\min_{1\leq j\leq n}|a_j|= m_{\min}, \quad \min_{p\neq j}\btwonorm{\vect \alpha_p- \vect \alpha_j}= \frac{0.81e^{-\frac{3}{2}}}{\Omega}\Big(\frac{\sigma}{m_{\min}}\Big)^{\frac{1}{2n-2}},
	\]
	where $\vect \alpha_j= \begin{pmatrix} \vect y_j\\ 
	\tau \vect v_j\end{pmatrix}$. 
\end{prop}
\begin{proof}
Let $\hat a_j, \vect {\hat y}_j, a_j, \vect y_j$'s be the ones in Proposition \ref{prop:staticnumberlowerboundthm0}. Then we have
\[
\babs{\sum_{j=1}^{n-1}\hat a_j e^{i\vect {\hat y}_j^\top \vect \omega} -\sum_{j=1}^{n}a_j e^{i \vect y_j^\top \vect \omega}} < \sigma, \quad ||\vect \omega||_2\leq \Omega,  \vect \omega \in \mathbb R^d.
\]  
When $\vect {\hat v}_j = \vect v_j = \vect v$, we also have 
\[
\babs{\sum_{j=1}^{n-1}\hat a_j e^{i(\vect {\hat y}_j^\top + t \tau \vect {\hat v}_j^\top)\vect \omega} -\sum_{j=1}^{n}a_j e^{i(\vect y_j^\top + t \tau \vect v_j^\top)\vect \omega}} < \sigma, \quad  ||\vect \omega||_2\leq \Omega, \ t=0, \cdots, T.
\]
The other parts of the proposition are easy to verify. 
\end{proof}

Proposition \ref{prop:trackhighdnumberlowerbound0} holds for the case when all $\tau \vect v_j$'s are equal or very close to each other. If the velocities are not close to each other, we next have Proposition \ref{prop:trackhighdnumberlowerbound1}, which shows that the order of the resolution of number recovery in Theorem \ref{thm:trackhighdupperboundnumberlimit0} is nearly optimal for the reconstruction with a short time period (i.e., when $T$ is small). For a tracking problem with long time period, we expect that the resolution for recovering the number of different velocities is of order $O(\frac{1}{T\Omega}(\frac{\sigma}{m_{\min}})^{\frac{1}{2n-2}})$. This will be demonstrated by results in Section \ref{sec:resolutionforvelocities}. 
\begin{prop}\label{prop:trackhighdnumberlowerbound1}
	For given $0<\sigma<m_{\min}$ and integer $n\geq 2$, there exist a $n$-sparse parameter set $\{(a_j, \vect y_j, \tau \vect v_j)\}_{j=1}^n$ with different $\vect v_j$'s, and a $(n-1)$-sparse parameter set $\{(\hat a_j, \vect {\hat y}_j, \tau \vect {\hat v}_j)\}_{j=1}^{n-1}$, such that 
	\[
	\babs{\sum_{j=1}^{n-1}\hat a_j e^{i(\vect {\hat y}_j^\top + t \tau \vect {\hat v}_j^\top)\vect \omega} -\sum_{j=1}^{n}a_j e^{i(\vect y_j^\top + t \tau \vect v_j^\top)\vect \omega}}< \sigma,\ ||\vect \omega||_2\leq \Omega, \ t=0, \cdots,T.
	\] 
	Moreover,
	\begin{equation}\label{equ:trackhighdnumberlowerbound1equ0}
	\min_{1\leq j\leq n}|a_j|= m_{\min}, \quad \min_{p\neq j}\btwonorm{\vect \alpha_p- \vect \alpha_j}= \frac{0.81\sqrt{2}e^{-\frac{3}{2}}}{(T+1)\Omega}\Big(\frac{\sigma}{m_{\min}}\Big)^{\frac{1}{2n-2}},
	\end{equation}
	where $\vect \alpha_j= \begin{pmatrix}
	\vect y_j \\
	\tau \vect v_j
	\end{pmatrix}$. 
\end{prop}
\begin{proof}
Let $\Delta = \frac{0.81e^{-\frac{3}{2}}}{\Omega}\Big(\frac{\sigma}{m_{\min}}\Big)^{\frac{1}{2n-2}}$. For the one-dimensional case, as in the proof in \cite{liu2021mathematicalhighd} or \cite{liu2021theorylse}, Proposition \ref{prop:staticnumberlowerboundthm0} holds when the $n-1$ and $n$ point sources are located at
\[
-(n-1)\Delta, \ -(n-2)\Delta,\  \cdots,\ 0,  \ \Delta,  \ (n-1)\Delta,
\]
with certain intensities $\hat a_j, a_j$'s. More specifically, let $\hat x_j= - j \Delta, x_j =(j-1)\Delta, j=1, \cdots, n$, we have 
\begin{equation}\label{equ:trackhighdnumberlowerbound1equ1}
\babs{\sum_{j=1}^{n-1}\hat a_j e^{i \hat x_j \omega} - \sum_{j=1}^{n-1}a_j e^{i x_j \omega}}<\sigma, \quad |\omega|\leq \Omega,
\end{equation}
with certain $\hat a_j, a_j$'s. We consider in the proposition the above $\hat a_j, a_j$'s and the following $\vect {\hat y}_j , \tau \vect {\hat v}_j, \vect y_j, \tau \vect v_j$'s,
\begin{align*}
(\vect {\hat y}_1^\top, \tau \vect {\hat v}_1^\top) = \frac{-\Delta}{T+1}\frac{1}{\sqrt{d}} (1,1, \cdots, 1),\ \cdots,\ (\vect {\hat y}_n, \tau \vect {\hat v}_n) = \frac{-(n-1)\Delta}{T+1}\frac{1}{\sqrt{d}}(1,1,\cdots,1),
\end{align*}
and 
\begin{align*}
&(\vect y_1^\top, \tau \vect v_1^\top) = (0,0, \cdots, 0),\quad (\vect y_2^\top, \tau \vect v_2^\top)=\frac{\Delta}{T+1} \frac{1}{\sqrt{d}} (1, 1, \cdots, 1),\quad \cdots,\\
&(\vect y_n^\top, \tau \vect v_n^\top) = \frac{(n-1)\Delta}{{T+1}} \frac{1}{\sqrt{d}} (1, 1, \cdots, 1).
\end{align*}
Note that (\ref{equ:trackhighdnumberlowerbound1equ0}) is satisfied. On the other hand, for any $||\vect \omega||_2\leq \Omega, \vect \omega \in \mathbb R^d, t=0, \cdots, T$, let $\vect \omega_t = \begin{pmatrix}
\vect \omega\\
t \vect \omega
\end{pmatrix}$, we have 
\begin{align*}
&\babs{\sum_{j=1}^{n-1}\hat a_j e^{i(\vect {\hat y}_j^\top + t \tau \vect {\hat v}_j^\top)\vect \omega} -\sum_{j=1}^{n}a_j e^{i(\vect y_j^\top + t \tau \vect v_j^\top)\vect \omega}}=  \babs{\sum_{j=1}^{n-1}\hat a_j e^{i(\vect {\hat y}_j^\top, \tau \vect {\hat v}_j^\top)\vect \omega_t} -\sum_{j=1}^{n}a_j e^{i(\vect y_j^\top, \tau \vect v_j^\top)\vect \omega_t}}\\
=& \babs{\sum_{j=1}^{n-1}\hat a_j e^{i(\vect {\hat y}_j^\top, \tau \vect {\hat v}^\top_j)(\vect u+\vect v)} -\sum_{j=1}^{n}a_j e^{i(\vect y_j^\top, \tau \vect v_j^\top)(\vect u+\vect v)}} \ \Big(\vect u+\vect v=\vect \omega_t, \vect u = \frac{(1+t)\omega}{2\sqrt{d}}(1,\cdots, 1)^\top, |\omega|\leq \Omega, \vect u^\top \vect v =0 \Big)\\
=& \babs{\sum_{j=1}^{n-1}\hat a_j e^{i(\vect {\hat y}_j^\top, \tau \vect {\hat v}_j^\top)\vect u} -\sum_{j=1}^{n}a_j e^{i(\vect y_j^\top, \tau \vect v_j^\top)\vect u}}\\
= &\babs{\sum_{j=1}^{n-1}\hat a_j e^{i\hat x_j \hat \omega} -\sum_{j=1}^{n}a_j e^{ix_j \hat \omega}}, 
\end{align*}
where $$\hat x_j = (\vect {\hat y}^\top_j, \tau \vect {\hat v}_j^\top)\frac{(1+T)}{2\sqrt{d}}(1,1, \cdots,1)^\top= -j \Delta,\ x_j =(\vect y_j^\top, \tau \vect v_j^\top)\frac{(1+T)}{2\sqrt{d}}(1, 1, \cdots, 1)^\top = (j-1)\Delta, $$ and $$\babs{\hat \omega} = \babs{\frac{t+1}{T+1}\omega}\leq \Omega.$$ Furthermore, by (\ref{equ:trackhighdnumberlowerbound1equ1}), we thus have 
\[
\babs{\sum_{j=1}^{n-1}\hat a_j e^{i(\vect {\hat y}_j^\top + t \tau \vect {\hat v}_j^\top)\vect \omega} -\sum_{j=1}^{n}a_j e^{i(\vect y_j^\top + t \tau \vect v_j^\top)\vect \omega}} <\sigma, \ ||\vect \omega||_2\leq \Omega, \ t=0, \cdots, T. 
\]
\end{proof}

\subsection{Stability of location and velocity recovery in the tracking problem}\label{sec:trackstabilityoflocalvelocityrecovery}
In this subsection, we present our main results for recovering the locations and velocities of point sources in the $d$-dimensional tracking problem. We have the following theorem whose proof is given in Section \ref{sec:proofmainresults}. 

\begin{thm}\label{thm:trackhighdupperboundsupportlimit0}
Let $n\geq 2$ and $T\geq \frac{n(n+1)}{2}$. Let the measurement $\mathbf Y_t, t=0, \cdots, T,$ in (\ref{equ:modelsetting1}) be generated by a $n$-sparse parameter set $\{(a_j, \vect y_j, \tau \vect v_j)\}_{j=1}^n$ with $\vect \alpha_j:= \begin{pmatrix}
\vect y_j\\
\tau \vect v_j
\end{pmatrix}
\in B_{\frac{(n-1)\pi}{n(n+2)\Omega}}^{2d}(\vect 0)$. Assume that
\begin{equation}
d_{\min}:=\min_{p\neq j}\Big|\Big|\vect \alpha_p-\vect \alpha_j\Big|\Big|_2\geq \frac{11.76e\pi^2\sqrt{(\frac{(n+1)n}{2})^2+1}4^{d-1}\Big((n+2)(n+1)/2\Big)^{\xi(d-1)}}{\Omega} \Big(\frac{\sigma}{m_{\min}}\Big)^{\frac{1}{2n-1}}, \end{equation}
where $\xi(d-1)$ is defined as in (\ref{equ:defineofxi}). If a $n$-sparse parameter set $\{(\hat a_j, \vect {\hat y}_j, \tau \vect{\hat v}_j)\}_{j=1}^n$ with $\vect {\hat \alpha}_j:= \begin{pmatrix}
\vect {\hat y}_j\\
\tau \vect {\hat v}_j
\end{pmatrix}
$ supported on $B_{\frac{(n-1)\pi}{n(n+2)\Omega}}^{2d}(\vect 0)$ is a $\sigma$-admissible parameter set of $\ \vect Y_t$'s, then after reordering the $\vect {\hat \alpha}_j$'s, we have 
\[
\btwonorm{\vect {\hat \alpha}_j - \vect {\alpha}_j} < \frac{d_{\min}}{2}.
\]
Moreover, we have 
\begin{equation}\label{equ:trackhighdsupportlimithm0equ1}
\Big|\Big|\vect {\hat \alpha}_j - \vect {\alpha}_j\Big|\Big|_2\leq \frac{C(d, n)}{\Omega}SRF^{2n-2}\frac{\sigma}{m_{\min}}, \quad 1\leq j\leq n,
\end{equation}
where $SRF$ is the super-resolution factor defined by $\frac{\pi}{d_{\min}\Omega}$ and 
\[
C(d, n)=\sqrt{6\pi}(2\pi)^{2n-2}\Big((\frac{n(n+1)}{2})^2+1\Big)^{\frac{2n-1}{2}}\big(4^{d-1}((n+2)(n+1)/2)^{\xi(d-1)}\big)^{2n-1}n2^{4n-2}e^{2n}\pi^{-\frac{1}{2}}.
\]
\end{thm}

Theorem \ref{thm:trackhighdupperboundsupportlimit0} demonstrates that for fixed dimensionality $d$ and source number $n$, when $\min_{j\neq p} \btwonorm{\begin{pmatrix} \vect y_j\\ \tau \vect v_j\end{pmatrix} - \begin{pmatrix}\vect y_p\\ \tau \vect v_p\end{pmatrix}} \geq O\Big(\frac{1}{\Omega}(\frac{\sigma}{m_{\min}})^{\frac{1}{2n-1}}\Big)$, it is possible to stably reconstruct both the source locations and velocities in the tracking problem. It also indicates that even when in some of the frames, the point sources are spaced so close that are unable to be stably resolved, but the dynamic reconstruction can still stably resolve the source locations and velocities. It demonstrates the superiority of dynamic reconstruction over the static reconstruction method.  In addition, by the following results, we shall see that the order $O(\frac{1}{\Omega}(\frac{\sigma}{m_{\min}})^{\frac{1}{2n-1}})$ is optimal in the worst-case scenario without further information on the velocities. We first recall the following result on the location recovery in the static super-resolution problem. 

\begin{prop}\label{prop:staticsupportlowerboundthm0}
	For given $0<\sigma<m_{\min}$ and integer $n\geq 2$, let 
	\begin{equation}\label{supportlowerboundequ0}
	\Delta=\frac{0.49e^{-\frac{3}{2}}}{\Omega}\ \Big(\frac{\sigma}{m_{\min}}\Big)^{\frac{1}{2n-1}}.
	\end{equation}
	Then there exist $a_j\in \mathbb C, j=1, \cdots,n$ and $\vect y_j$'s at $\{(-\Delta, 0, \cdots, 0)^\top, (-2\Delta, 0, \cdots, 0)^\top, (-n\Delta, 0, \cdots, 0)^\top\}$ and another $\hat a_j\in \mathbb C, j=1, \cdots, n$ and ${\vect{\hat y}_j}$'s at  $\{(0,0,\cdots, 0)^\top,(\Delta, 0, \cdots, 0)^\top, ((n-1)\Delta, 0, \cdots, 0)^\top\}$ such that
	\begin{equation}\label{equ:staticsupportlowerboundthm0eq0}
	\babs{\sum_{j=1}^{n} \hat a_j e^{i \vect {\hat y}_j^\top \vect \omega}- \sum_{j=1}^{n} a_j e^{i \vect y_j^\top \vect \omega}}<\sigma, \quad  ||\vect \omega||_2\leq \Omega, 
	\end{equation}
	and either $\min_{1\leq j\leq n}|a_j|= m_{\min}$ or $\min_{1\leq j\leq n}|\hat a_j|= m_{\min}$.   
\end{prop}
\begin{proof}
See Proposition 2.8 in \cite{liu2021mathematicalhighd}.
\end{proof}

As a direct consequence of Proposition \ref{prop:staticsupportlowerboundthm0}, we have the following result.
\begin{prop}\label{prop:trackhighdsupportlowerbound0}
	For given $0<\sigma<m_{\min}$ and integer $n\geq 2$, let 
	\begin{equation}\label{equ:trackhighdsupportlowerbound0equ0}
	\Delta=\frac{0.49e^{-\frac{3}{2}}}{\Omega}\ \Big(\frac{\sigma}{m_{\min}}\Big)^{\frac{1}{2n-1}}.
	\end{equation}
	Then there exist $n$-sparse parameter set $\{(a_j, \vect y_j, \tau \vect v_j)\}_{j=1}^n$'s with $\vect v_j = \vect v,j=1,\cdots,n$ and $$\vect y_1 = (-\Delta, 0, \cdots, 0)^\top, \vect y_2 = (-2\Delta, 0, \cdots, 0)^\top, \cdots , \vect y_n = (-n\Delta, \cdots, 0)^\top,$$ and $n$-sparse parameter set $\{(\hat a_j, \vect {\hat y}_j, \tau \vect {\hat v}_j)\}_{j=1}^n$'s with $\vect {\hat v}_j = \vect v,j=1,\cdots,n$ and $$\vect {\hat y}_1 = (0, \cdots, 0)^\top, \vect {\hat y}_2=(\Delta, 0, \cdots, 0)^\top, \cdots, \vect {\hat y}_n=((n-1)\Delta, \cdots, 0)^\top$$ such that
	\[
	\babs{\sum_{j=1}^{n}\hat a_j e^{i(\vect{\hat y}_j^\top + t \tau \vect {\hat v}_j^\top)\vect \omega} -\sum_{j=1}^{n}a_j e^{i(\vect y_j^\top + t \tau \vect v_j^\top)\vect \omega}}< \sigma, \ ||\vect \omega||_2\leq \Omega, \ t=0, \cdots,T,
	\]
	and either $\min_{1\leq j\leq n}|a_j|= m_{\min}$ or $\min_{1\leq j\leq n}|\hat a_j|= m_{\min}$.   
\end{prop}

This result holds for the case when all $\tau \vect v_j$'s are equal or very close to each other with respect to $T$. Thus, theoretically the resolution of $\vect y_j$'s in the worst-case is expected to be of order $O\Big(\frac{1}{\Omega}(\frac{\sigma}{m_{\min}})^{\frac{1}{2n-1}}\Big)$.  

If the velocities are not close to each other, we next have Proposition \ref{prop:trackhighdsupportlowerbound1}, which shows that the order of the resolution in Theorem \ref{thm:trackhighdupperboundsupportlimit0} is nearly optimal for the reconstruction problem with short time period (i.e., when $T$ is small). For a tracking problem with long time period, we expect the resolution for the velocities to be of order $O(\frac{1}{T\Omega}(\frac{\sigma}{m_{\min}})^{\frac{1}{2n-1}})$. This will also be demonstrated by results in Section \ref{sec:resolutionforvelocities}. 
\begin{prop}\label{prop:trackhighdsupportlowerbound1}
	For given $0<\sigma<m_{\min}$ and integer $n\geq 2$, let 
	\begin{equation}\label{equ:trackhighdsupportlowerbound1equ0}
	\Delta=\frac{0.49e^{-\frac{3}{2}}}{(T+1)\Omega}\ \Big(\frac{\sigma}{m_{\min}}\Big)^{\frac{1}{2n-1}}.
	\end{equation}
	Then there exist $n$-sparse parameter set $\{(\hat a_j, \vect {\hat y}_j, \tau \vect {\hat v}_j)\}_{j=1}^n$'s with $(\vect {\hat y}_1^\top, \vect {\hat v}_1^\top) = - \frac{\Delta}{\sqrt{d}}(1,\cdots, 1),\ \cdots,$ 
	$(\vect {\hat y}_n^\top, \vect {\hat v}_n^\top) = -\frac{n\Delta}{\sqrt{d}}(1,\cdots, 1)$, and $n$-sparse parameter set $\{(a_j, \vect y_j, \tau \vect v_j)\}_{j=1}^n$'s with $(\vect y_1^\top, \vect v_1^\top) = (0,0, \cdots, 0),$\\ 
	$(\vect y_2^\top, \vect v_2^\top) = \frac{\Delta}{\sqrt{d}}(1, \cdots, 1), \cdots, (\vect y_n^\top, \vect v_n^\top) = \frac{(n-1)\Delta}{\sqrt{d}}(1, \cdots, 1)$ such that
	\[
	\babs{\sum_{j=1}^{n}\hat a_j e^{i(\vect {\hat y}_j^\top + t \tau \vect {\hat v}_j^\top)\vect \omega} -\sum_{j=1}^{n}a_j e^{i(\vect y_j^\top + t \tau \vect v_j^\top)\vect \omega}}< \sigma, \ ||\vect \omega||_2\leq \Omega, \ t=0, \cdots,T,
	\]
	and either $\min_{1\leq j\leq n}|a_j|= m_{\min}$ or $\min_{1\leq j\leq n}|\hat a_j|= m_{\min}$.   
\end{prop}
\begin{proof}
Let $ \hat \Delta = \frac{0.49e^{-\frac{3}{2}}}{\Omega}\Big(\frac{\sigma}{m_{\min}}\Big)^{\frac{1}{2n-1}}$. Let $\hat x_j= - j \hat \Delta, x_j =(j-1) \hat \Delta, j=1, \cdots, n$, by Proposition \ref{prop:staticsupportlowerboundthm0} for the one-dimensional case we have 
\begin{equation}\label{equ:trackhighdsupportlowerbound1equ1}
\babs{\sum_{j=1}^{n}\hat a_j e^{i \hat x_j \omega} - \sum_{j=1}^{n-1}a_j e^{i x_j \omega}}<\sigma, \quad |\omega|\leq \Omega,
\end{equation}
for certain $\hat a_j, a_j$'s. We consider the above $\hat a_j, a_j$'s and the following $\vect {\hat y}_j , \tau \vect {\hat v}_j, \vect y_j, \tau \vect v_j$'s,
\[
(\vect {\hat y}_1^\top, \tau \vect {\hat v}_1^\top)= \frac{-\Delta}{\sqrt{d}} (1,\cdots, 1),\quad \cdots, \quad (\vect{\hat  y}_n^\top, \tau \vect {\hat v}_n^\top) = \frac{-n\Delta}{\sqrt{d}}(1, \cdots, 1),
\]  
and 
\[
(\vect y_1^\top, \tau \vect v_1^\top) = (0,\cdots, 0),\ (\vect y_2^\top, \tau \vect v_2^\top)=\frac{\Delta}{\sqrt{d}}(1, \cdots, 1),\ \cdots,\ (\vect y_n^\top, \tau \vect v_n^\top) = \frac{(n-1)\Delta}{\sqrt{d}}(1, \cdots, 1).
\]
For any $||\vect \omega||_2\leq \Omega, t=0, \cdots, T$, we obtain that
\begin{align*}
&\babs{\sum_{j=1}^{n}\hat a_j e^{i(\vect {\hat y}_j^\top + t \tau \vect {\hat v}_j^\top)\vect \omega} -\sum_{j=1}^{n}a_j e^{i(\vect y_j^\top + t \tau \vect v_j^\top)\vect \omega}}=  \babs{\sum_{j=1}^{n}\hat a_j e^{i(\vect {\hat y}_j^\top, \tau \vect {\hat v}_j^\top)\vect \omega_t} -\sum_{j=1}^{n}a_j e^{i(\vect y_j^\top, \tau \vect v_j^\top)\vect \omega_t}}\\
=& \babs{\sum_{j=1}^{n}\hat a_j e^{i(\vect {\hat y}_j^\top, \tau \vect {\hat v}_j^\top)(\vect u+\vect v)} -\sum_{j=1}^{n}a_j e^{i(\vect y_j^\top, \tau \vect v_j^\top)(\vect u+\vect v)}} \ \Big(\vect u+\vect v=\vect \omega_t, \vect u = \frac{(1+t)\omega}{2\sqrt{d}}(1,\cdots, 1)^\top, |\omega|\leq \Omega, \vect u^\top \cdot \vect v =0 \Big)\\
=& \babs{\sum_{j=1}^{n}\hat a_j e^{i(\vect {\hat y}_j^\top, \tau \vect {\hat v}_j^\top)\vect u} -\sum_{j=1}^{n}a_j e^{i(\vect y_j^\top, \tau \vect v_j^\top)\vect u}}\\
= &\babs{\sum_{j=1}^{n}\hat a_j e^{i\hat x_j \hat \omega} -\sum_{j=1}^{n}a_j e^{ix_j \hat \omega}}, 
\end{align*}
where $\hat x_j = (\vect {\hat y}_j^\top, \tau \vect {\hat v}_j^\top)\frac{(1+T)}{2\sqrt{d}}(1, \cdots, 1)^\top= -j \hat \Delta, x_j =(\vect {\hat y}_j^\top, \tau \vect {\hat v}_j^\top)\frac{(1+T)}{2\sqrt{d}}(1, \cdots, 1)^\top = (j-1)\hat \Delta$ and $\babs{\hat \omega} = \babs{\frac{t+1}{T+1}\omega}\leq \Omega$. Furthermore,  by (\ref{equ:trackhighdsupportlowerbound1equ1}), we thus have 
\[
\babs{\sum_{j=1}^{n}\hat a_j e^{i(\vect {\hat y}_j^\top + t \tau \vect {\hat v}_j^\top)\vect \omega} -\sum_{j=1}^{n}a_j e^{i(\vect y_j^\top + t \tau \vect v_j^\top)\vect \omega}} <\sigma, \quad ||\vect \omega||_2\leq \Omega, \ t=0, \cdots, T. 
\]
\end{proof}

\subsection{Stability for sparsity-promoting dynamic reconstruction}
Sparsity-based modeling and optimization is a common way in super-resolution that seeks to accelerate the resolving process or reduce the number of measurements. In \cite{alberti2019dynamic}, a sparsity-promoting algorithm is proposed to super-resolve source locations and velocities in the tracking problem, where the authors' aim is to find an admissible $2d$-dimensional source with minimum total variation norm. There, the theoretical stability result for the optimization requires the point sources to be separated by more than a Rayleigh length in each time step $t$, which is inadequate for demonstrating its super-resolution capability.

As a corollary of the results in Subsections \ref{sec:trackstabilityofnumberrecovery} and \ref{sec:trackstabilityoflocalvelocityrecovery}, we can derive a stability result for a sparsity-promoting dynamic reconstruction in the tracking problem, which demonstrates the possibility of achieving super-resolution for the sparsity-promoting dynamic reconstructions. We consider a sparsity-promoting dynamic tracking problem which seeks to find the sparest solution satisfying the measurement constraints. The optimization problem is  
\begin{equation}\label{equ:sparsityoptimization}
\min_{\{(\hat a_j, \vect {\hat y}_j, \tau \vect {\hat v}_j)\}, (\vect {\hat y}_j^\top,\tau \vect {\hat v}_j^\top)^\top \in \mathcal O} \# \{(\hat a_j, \vect {\hat y}_j, \tau \vect {\hat v}_j)\} \quad \text{subject to} \  |\mathcal F[\rho_t](\vect \omega) - \vect Y_k(\vect \omega)|< \sigma, \ t=0,\cdots,T, \ ||\vect \omega||_2\leq \Omega,
\end{equation}
where $\rho_t = \sum_{j=1}^n \hat a_j \delta_{\vect {\hat y}_j+t\tau \vect {\hat v}_j}$, and $\mathcal F[\rho_t](\vect \omega) = \sum_{j=1}^n \hat a_j e^{i(\vect {\hat y}_j^\top+t\tau \vect {\hat v}_j^\top)\vect \omega}$. We have the following theorem for the stability of the minimization problem (\ref{equ:sparsityoptimization}).

\begin{thm}\label{thm:stabilityofsparsitybaserecovery}
Let $n\geq 2$, $T\geq \frac{n(n+1)}{2}$, and $\sigma\leq m_{\min}$. Let the measurement $\mathbf Y_t, t=0, \cdots, T$, in (\ref{equ:modelsetting1}) be generated by a $n$-sparse parameter set $\{(a_j, \vect y_j, \tau \vect v_j)\}_{j=1}^n$ with $\vect \alpha_j:= \begin{pmatrix} \vect y_j\\ \tau \vect v_j \end{pmatrix} \in B_{\frac{(n-1)\pi}{n(n+2)\Omega}}^{2d}(\vect 0)$. Assume that
\begin{equation}\label{equ:stabilityofsparsitybaserecoveryequ0}
d_{\min}:=\min_{p\neq j}\Big|\Big|\vect \alpha_p-\vect \alpha_j\Big|\Big|_2\geq \frac{11.76e\pi^2\sqrt{(\frac{(n+1)n}{2})^2+1}4^{k-1}\Big((n+2)(n+1)/2\Big)^{\xi(k-1)}}{\Omega} \Big(\frac{\sigma}{m_{\min}}\Big)^{\frac{1}{2n-1}}. 
\end{equation}
Let $\mathcal O$ in (\ref{equ:sparsityoptimization}) be $B_{\frac{(n-1)\pi}{n(n+2)\Omega}}^{2d}(\vect 0)$, then the solution to (\ref{equ:sparsityoptimization}) contains exactly $n$ elements. If a $n$-sparse parameter set $\{(\hat a_j, \vect {\hat y}_j, \tau \vect {\hat v}_j)\}_{j=1}^n$ is the solution, after reordering the $\vect {\hat \alpha}_j$'s, we have 
\[
\btwonorm{\vect {\hat \alpha}_j - \vect {\alpha}_j}< \frac{d_{\min}}{2}.
\]
Moreover, we have 
\begin{equation}\label{equ:stabilityofsparsitybaserecoveryequ1}
\Big|\Big|\vect {\hat \alpha}_j - \vect {\alpha}_j\Big|\Big|_2\leq \frac{C(d, n)}{\Omega}SRF^{2n-2}\frac{\sigma}{m_{\min}}, \quad 1\leq j\leq n,
\end{equation}
where $SRF:=\frac{\pi}{d_{\min}\Omega}$ and 
\[
C(d, n)=\sqrt{6\pi}(2\pi)^{2n-2}\Big((\frac{n(n+1)}{2})^2+1\Big)^{\frac{2n-1}{2}}\big(4^{d-1}((n+2)(n+1)/2)^{\xi(d-1)}\big)^{(2n-1)}n2^{4n-2}e^{2n}\pi^{-\frac{1}{2}}.
\]
\end{thm}  
\begin{proof}
Note that by $\sigma \leq m_{\min}$, (\ref{equ:stabilityofsparsitybaserecoveryequ0}) is greater than (\ref{equ:trackhighdupperboundnumberlimit0equ1}). Thus by Theorem \ref{thm:trackhighdupperboundnumberlimit0}, the solution to (\ref{equ:sparsityoptimization}) contains $n$ elements. Then by Theorem \ref{thm:trackhighdupperboundsupportlimit0}, we prove the desired result.
\end{proof}

Theorem \ref{thm:stabilityofsparsitybaserecovery} reveals that sparsity promoting over admissible solutions satisfying the measurements of all the frames could resolve the location-velocity pairs to the resolution limit order. It provides an insight that theoretically sparsity-promoting dynamic reconstruction algorithms would have favorable performances on the tracking problem. Especially, under the separation condition \ref{equ:stabilityofsparsitybaserecoveryequ0}, any tractable sparsity-promoting algorithms (such as TV minimization) rendering the sparsest solution could stably reconstruct all the location-velocity pairs. Theorem \ref{thm:stabilityofsparsitybaserecovery} also indicates that even when the static reconstruction fails in some of the frames, the sparsity-promoting dynamic reconstruction could also resolve the location-velocity pairs, which is consistent with the numerical results in \cite{alberti2019dynamic}.

\section{Main results for velocity recovery}\label{sec:resolutionforvelocities}
In this section, we present the results for the stability of velocity recovery in the tracking problems. We shall first introduce better resolution estimates for recovering the number and values of velocities and then present a stability result for a sparsity-promoting velocity reconstruction algorithm. 

\subsection{Better resolution estimates for recovering the number and values of velocities}
In this subsection, we shall derive better resolution estimates for recovering the number of sources and the values of velocities in the tracking problem. We first have the following theorem for the reconstruction of the number of sources. 

\begin{thm}\label{thm:trackvelocitiesnumberlimit0}
	Let $n\geq 2$ and $T\geq 2n-2$. Let the measurement $\mathbf Y_t, t=0, \cdots, T,$ in (\ref{equ:modelsetting1}) be generated by a $n$-sparse parameter set $\{(a_j, \vect y_j, \tau \vect v_j)\}_{j=1}^n$ with $\tau \vect v_j \in B_{\frac{(n-1)\pi}{T\Omega}}^{d}(0)$. Assume that
	\begin{equation}\label{equ:trackhighdvelocitiesnumberlimit0equ0}
		\min_{p\neq j}\btwonorm{\tau \vect v_p - \tau \vect v_j}\geq \frac{8.8\pi e(\pi/2)^{d-1} \Big(n(n-1)/\pi\Big)^{\xi(d-1)}}{T\Omega }\Big(\frac{\sigma}{m_{\min}}\Big)^{\frac{1}{2n-2}}.
	\end{equation}
	Then there does not exist any $\sigma$-admissible parameter set of \,$\mathbf Y_t$'s with less than $n$ elements.
\end{thm}

Theorem \ref{thm:trackvelocitiesnumberlimit0} reveals that when the minimum difference between velocities is greater than $O(\frac{1}{T\Omega}(\frac{\sigma}{m_{\min}})^{\frac{1}{2n-2}})$, then the number of point sources can be exactly reconstructed in the tracking problem. By Proposition \ref{prop:trackhighdnumberlowerbound1}, we note that this order is optimal. In the following theorem, we further develop a better estimate for the resolution of the velocity reconstruction. 

\medskip
\begin{thm}\label{thm:trackvelocitiessupplimit0}
	Let $n\geq 2$ and $T\geq 2n-1$. Let the measurement $\mathbf Y_t, t=0, \cdots, T,$ in (\ref{equ:modelsetting1}) be generated by a $n$-sparse parameter set $\{(a_j, \vect y_j, \tau \vect v_j)\}_{j=1}^n$ with $\tau \vect v_j \in B_{\frac{(n-1)\pi}{T\Omega}}^{d}(\vect 0)$. Assume that
	\begin{equation}\label{equ:trackhighdvelocitiessupplimit0equ0}
		d_{\min}:=\min_{p\neq j}\btwonorm{\tau \vect v_p-\tau \vect v_j}\geq \frac{11.76e \pi 4^{d-1}\Big((n+2)(n+1)/2\Big)^{\xi(d-1)}}{T\Omega} (\frac{\sigma}{m_{\min}})^{\frac{1}{2n-1}}. \end{equation}
	If a $n$-sparse parameter set $\{(\hat a_j, \vect {\hat y}_j, \tau \vect {\hat v}_j)\}_{j=1}^n$ with $\tau \vect {\hat v}_j$'s supported on $B_{\frac{(n-1)\pi}{T\Omega}}^{d}(\vect 0)$ is a $\sigma$-admissible parameter set of $\vect Y_t$'s, then after reordering the $\vect {\hat v}_j$'s, we have 
	\[
	\btwonorm{\tau \vect {\hat v}_j - \tau \vect v_j} < \frac{d_{\min}}{2}.
	\]
	Moreover, we have 
	\begin{equation}\label{equ:trackhighdvelocitiessupplimit0equ1}
		\btwonorm{\tau \vect {\hat v}_j - \tau \vect {v}_j}\leq \frac{C(d, n)}{T\Omega}SRF^{2n-2}\frac{\sigma}{m_{\min}}, \quad 1\leq j\leq n,
	\end{equation}
	where $SRF$ is the super-resolution factor defined by $\frac{\pi}{d_{\min}T\Omega}$ and 
	\[
	C(d, n) = \big(4^{d-1}((n+2)(n+1)/2)^{\xi(d-1)}\big)^{2n-1}n2^{6n-3}e^{2n}\pi^{-\frac{1}{2}}.
	\]
\end{thm}

Theorem \ref{thm:trackvelocitiessupplimit0} reveals that when the minimum difference between the velocities is greater than $O(\frac{1}{T\Omega}(\frac{\sigma}{m_{\min}})^{\frac{1}{2n-2}})$, then the velocities can be stably reconstructed in the tracking problem. By Proposition \ref{prop:trackhighdsupportlowerbound1}, we note that this order is optimal. 

In Sections \ref{section:velocitynumberalgorithm} and \ref{section:velocitysupportalgorithm}, we will propose super-resolution algorithms for reconstructing respectively the number of sources and the values of the velocities, which are demonstrated theoretically or numerically to lead to the optimal resolution orders that are shown above. 

We remark that both Theorems \ref{thm:trackvelocitiesnumberlimit0} and \ref{thm:trackvelocitiessupplimit0} are derived by a crucial observation that the measurements at point $\vect \omega = \Omega \vect v$ are   
\begin{equation}\label{equ:onedmeasurementforvelocityrecover}
\vect Y_t(\vect v\Omega) = \sum_{j=1}^n a_je^{i(\vect y_j^\top+t\tau \vect v_j^\top)\vect v \Omega} + \vect W_t(\vect v \Omega), \quad t=0, \cdots, T. 
\end{equation}
Let $b_j = a_j e^{i\vect y_j^\top \vect v \Omega}$ and $y_j =\tau  \vect v_j^\top \vect v$ and $\vect W(t) = \vect W_t(\vect v \Omega)$, the measurement can be written as 
\[
\vect Y(t) = \sum_{j=1}^n b_j e^{iy_j \Omega t} + \vect W(t) \quad  t=0, \cdots, T.
\]	
Thus the inherent cut-off frequency for recovering $\tau \vect v_j$'s can be viewed as $T\Omega$, which results in a much better resolution limit $\frac{\pi}{T\Omega}$ compared to the one $\frac{\pi}{\Omega}$ for each frame. Because recovering objects separated larger than Rayleigh limit is stable even for imaging system with low signal-to-noise ratio, our observation indicates that, by some dynamic reconstruction algorithms, we are very likely able to achieve better resolution for velocities than that of the static reconstruction. This will be confirmed by numerical algorithms in Section \ref{section:velocitysupportalgorithm}. For other examples, by applying the algorithms in \cite{candes2014towards, tang2014near} which are provable to resolve point sources with a separation distance of Rayleigh limit order to several measurements like (\ref{equ:onedmeasurementforvelocityrecover}),  we are able to recover the velocities with resolution of order $O(\frac{\pi}{T\Omega})$. We anticipate the above observation can inspire many new reconstruction algorithms that enhance significantly the resolution of the tracking problem in practice.

\subsection{Stability for a sparsity-promoting velocity reconstruction}
As a corollary of the results in the above subsection, we can derive a stability result for velocity recovery by the sparsity-promoting dynamic reconstruction algorithm (\ref{equ:sparsityoptimization}), by which we demonstrate that the sparsity-promoting dynamic reconstruction attains the optimal super-resolution capability for the velocity recovery. Specifically, we have the following theorem. 

\begin{thm}\label{thm:stabilityofsparsitybasevelocityrecovery}
	Let $n\geq 2$, $T\geq 2n-1$, and $\sigma\leq m_{\min}$. Let the measurement $\mathbf Y_t, t=0, \cdots, T$, in (\ref{equ:modelsetting1}) be generated by a $n$-sparse parameter set $\{(a_j, \vect y_j, \tau \vect v_j)\}_{j=1}^n$ with $\tau \vect v_j \in B_{\frac{(n-1)\pi}{T\Omega}}^{d}(\vect 0)$. Assume that
	\begin{equation}\label{equ:stabilityofsparsitybasevelocityrecoveryequ0}
		d_{\min}:=\min_{p\neq j}\btwonorm{\tau \vect v_p-\tau \vect v_j}\geq \frac{11.76e \pi 4^{d-1}\Big((n+2)(n+1)/2\Big)^{\xi(d-1)}}{T\Omega} (\frac{\sigma}{m_{\min}})^{\frac{1}{2n-1}}. 
	\end{equation}
	Let $\mathcal O$ in (\ref{equ:sparsityoptimization}) be $B_{\frac{(n-1)\pi}{T\Omega}}^{d}(\vect 0)$, then the solution to (\ref{equ:sparsityoptimization}) contains exactly $n$ elements. If a $n$-sparse parameter set $\{(\hat a_j, \vect {\hat y}_j, \tau \vect {\hat v}_j)\}_{j=1}^n$ is the solution, after reordering the $\vect {\hat \alpha}_j$'s, we have 
	\[
	\btwonorm{\tau \vect {\hat v}_j - \tau \vect v_j} < \frac{d_{\min}}{2}.
	\]
	Moreover, we have 
	\begin{equation}\label{equ:stabilityofsparsitybasevelocityrecoveryequ1}
		\btwonorm{\tau \vect {\hat v}_j - \tau \vect {v}_j}\leq \frac{C(d, n)}{T\Omega}SRF^{2n-2}\frac{\sigma}{m_{\min}}, \quad 1\leq j\leq n,
	\end{equation}
	where $SRF:=\frac{\pi}{d_{\min}T\Omega}$ and 
	\[
	C(d, n) = \big(4^{d-1}((n+2)(n+1)/2)^{\xi(d-1)}\big)^{2n-1}n2^{6n-3}e^{2n}\pi^{-\frac{1}{2}}.
	\]
\end{thm}

\section{Projection-based velocity number detection algorithm}\label{section:velocitynumberalgorithm}
By the projection methods we use in Section \ref{sec:proofresultsofvelocity}, we propose in this section and next section new algorithms for super-resolving the velocities. The strength of our algorithms is that they are able to resolve the number of sources and reconstruct the values of the velocities at the resolution limit orders, i.e., $O(\frac{1}{T\Omega} (\frac{\sigma}{m_{\min}})^{\frac{1}{2n-2}})$ and $O(\frac{1}{T\Omega} (\frac{\sigma}{m_{\min}})^{\frac{1}{2n-1}})$, respectively. Therefore, as the number of time steps increases, the resolution improves as well.

In this section, we propose a projection-based sweeping singular-value-thresholding number detection algorithm for the tracking problem in two dimensions. The extension to higher dimensions is straightforward. 

\subsection{One-dimensional sweeping singular-value-thresholding number detection algorithm}
In this subsection, we review the sweeping singular-value-thresholding number detection algorithm in  dimension one \cite{liu2021theorylse} for the static super-resolution problem. We refer to  \cite{akaike1998information, akaike1974new, wax1985detection, schwarz1978estimating, rissanen1978modeling, wax1989detection, lawley1956tests, chen1991detection, he2010detecting, han2013improved} and the references therein for other interesting algorithms in one dimension. 

We consider a model that is tuned for our velocity recovery problem. Especially, we consider the measurement $\vect Y$ generated by $\mu =\sum_{j=1}^n a_j \delta_{y_j}, y_j \in\mathbb R$ as
\begin{equation}\label{equ:modelsetting2}
\vect Y(t) = \sum_{j=1}^n a_j e^{i y_j \Omega t} + \vect W(t), \quad t=0, \cdots, T,
\end{equation}
with $|\vect W(t)|<\sigma$ and $\min_{j}|a_j|=m_{\min}$. We choose measurements at the sample points $qr, q=0,\cdots,2s$ with $s\geq n$ and $r,g$ being integer satisfying $T+1=2sr+g, 0\leq g<2s$:
\[
\mathbf Y(qr)= \sum_{j=1}^{n}a_j e^{i y_j \Omega qr} +\mathbf W(qr), \quad 0\leq q \leq 2s.
\]
We then form the following Hankel matrix:
\begin{equation}\label{equ:hankelmatrix1}
\mathbf H(s)=\left(\begin{array}{cccc}
\mathbf Y(0)&\mathbf Y(r)&\cdots& \mathbf Y(sr)\\
\mathbf Y(r)&\mathbf Y(2r)&\cdots&\mathbf Y((s+1)r)\\
\cdots&\cdots&\ddots&\cdots\\
\mathbf Y(sr)&\mathbf Y((s+1)r)&\cdots&\mathbf Y(2sr)
\end{array}
\right),
\end{equation}
and consider the singular value decomposition of $\mathbf H(s)$ 
\[\mathbf H(s)=\hat U\hat \Sigma \hat U^*,\]
where $\hat\Sigma =\text{diag}(\hat \sigma_1,\cdots, \hat \sigma_n, \hat \sigma_{n+1},\cdots,\hat\sigma_{s+1})$ with the singular values $\hat \sigma_j$, $1\leq j \leq s+1$, ordered in a decreasing manner. We then determine the source number by thresholding on these singular values with a properly chosen threshold based on Theorem \ref{thm:onedMUSICthm1} below. To derive Theorem \ref{thm:onedMUSICthm1}, we first introduce the notation
\begin{equation}\label{equ:defiofphi}
    \phi_s(x) = (1, x, \cdots, x^s)^\top,
\end{equation}
\begin{equation}\label{equ:defiofzeta}
    \zeta(k)= \left\{
\begin{array}{cc}
(\frac{k-1}{2}!)^2,& \text{$k$ is odd,}\\
(\frac{k}{2})!(\frac{k-2}{2})!,& \text{$k$ is even,}
\end{array} 
\right. 
\end{equation}
and the following auxiliary lemma. 
\begin{lem}\label{norminversevandermonde2}
	Let $k \geq 2$ and $-\frac{\pi}{2}\leq \theta_1<\theta_2<\cdots<\theta_{k}  \leq \frac{\pi}{2}$. Let
	$\theta_{\min}=\min_{p\neq j}|\theta_p-\theta_j|$ and $V_{k}(k-1)=\big(\phi_{k-1}(e^{i \theta_1}),\cdots,\phi_{k-1}(e^{i\theta_k})\big)$ with $\phi_{k-1}(x)$ defined as in (\ref{equ:defiofphi}). Then
	\[
	||V_k(k-1)^{-1}||_{\infty}\leq \frac{\pi^{k-1}}{\zeta(k) \theta_{\min}^{k-1}},
	\] 
	where $\zeta(k)$ is defined in (\ref{equ:defiofzeta}).
\end{lem}
\begin{proof}
See Lemma 3 in \cite{liu2021theorylse}.
\end{proof}

% \begin{lem}\label{singularvaluevandermonde2}
% Let $V_{k}(k-1)=\big(\phi_{k-1}(e^{i \theta_1}),\cdots,\phi_{k-1}(e^{i\theta_k})\big)$ and $V_{k}(s)=\big(\phi_{s}(e^{i \theta_1}),\cdots,\phi_s(e^{i\theta_k})\big)$
% with $s>k-1$, then the following estimate on their singular values hold:
% \begin{align*}
% &\frac{1}{\sqrt{k}}\min_{1\leq j\leq k}\Pi_{1\leq p\leq k,p\neq j}\frac{|e^{i\theta_j}-e^{i\theta_p}|}{2} \leq\frac{1}{||V_k(k-1)^{-1}||_{2}}\\
% \leq &\sigma_{\min}\big(V_k(k-1)\big)\leq \sigma_{\min}\big(V_k(s)\big).
% \end{align*}
% \end{lem}
% \begin{proof}
% See Lemma 4 in \cite{liu2021theorylse}.
% \end{proof}

\begin{thm}\label{thm:onedMUSICthm1}
	Let $s\geq n$ and $\mu=\sum_{j=1}^{n}a_j \delta_{y_j}$ with $y_j\in [\frac{-(n-1)\pi}{T\Omega}, \frac{(n-1)\pi}{T\Omega}], 1\leq j\leq n$. We have 
	\begin{equation}\label{MUSICthm1equ-1}
	\hat \sigma_j\leq  (s+1)\sigma,\quad j=n+1,\cdots,s+1.
	\end{equation}
	Moreover, if the following separation condition is satisfied
	\begin{equation}\label{MUSICthm1equ0}
	\min_{p\neq j}\babs{y_p-y_j}>\frac{2\pi (s+1)}{T\Omega}\Big(\frac{2n(s+1)}{\zeta(n)^2}\frac{\sigma}{m_{\min}}\Big)^{\frac{1}{2n-2}},
	\end{equation}
	where $\zeta(n)$ is defined as in (\ref{equ:defiofzeta}), then
	\begin{equation}\label{MUSICthm1equ2}
	\hat\sigma_{n}>(s+1)\sigma.
	\end{equation}
\end{thm}
\begin{proof} \textbf{Step 1.} Note that $\vect H(s)$ has a decomposition that $\vect H(s) = DAD^\top +\Delta$ where $A=\text{diag}(a_1, \cdots, a_n)$, $D=\big(\phi_{s}(e^{i y_1 \Omega r}), \cdots, \phi_{s}(e^{i y_n \Omega r})\big)$ with $\phi_{s}(\omega)$ being defined in (\ref{equ:defiofphi}) and $\Delta$ is the matrix from the noise $\vect W$. We first consider the case when $\vect W(t)=0$, in which $\vect H(s)=DAD^\top$. Denoting $\sigma_n$ as the $n$-th singular value of $\vect H(s)$. Note that $\sigma_n$ is the minimum nonzero singular value of $DAD^\top$. Let $\ker(D^\top)$ be the kernel space of $D^\top$ and $\ker^{\perp}(D^\top)$ be its orthogonal complement, we have 
\begin{align*}
\sigma_n=\min_{||x||_2=1,x\in \ker^{\perp}(D^\top)}||DAD^\top||_2\geq \sigma_{\min}(DA)\sigma_n(D^\top)\geq \sigma_{\min}(D)\sigma_{\min}(A)\sigma_{\min}(D).
\end{align*}
Let $\theta_{\min}=\min_{p\neq j}\Big|y_pr\Omega-y_jr\Omega\Big|$. Since $s\geq n$, $y_jr\Omega\in \big[\frac{-\pi}{2}, \frac{\pi}{2}\big]$ for $y_j\in \big[-\frac{(n-1)\pi}{T\Omega}, \frac{(n-1)\pi}{T\Omega}\big]$. By Lemma \ref{norminversevandermonde2}, we have
\begin{align*}
\sigma_{\min}(D)\geq \frac{1}{||V_n(n-1)^{-1}||_{2}} \geq  \frac{1}{\sqrt{n}||V_n(n-1)^{-1}||_{\infty}}\geq  \frac{1}{\sqrt{n}}\frac{\zeta(n)\theta_{\min}^{n-1}}{\pi^{n-1}},
\end{align*}
where $V_n(n-1) = \big(\phi_{n-1}(e^{i y_1 \Omega r}), \cdots, \phi_{n-1}(e^{i y_n \Omega r})\big)$. It follows that
\begin{align}\label{equ:MUSICthm1equ1}
\sigma_n\geq \sigma_{\min}(A)\Big(\frac{1}{\sqrt{n}}\frac{\zeta(n)\theta_{\min}^{n-1}}{\pi^{n-1}}\Big)^2\geq \frac{m_{\min}\zeta(n)^2\theta_{\min}^{2n-2}}{n\pi^{2n-2}}.
\end{align}

\textbf{Step 2.} We now prove the theorem. Since $||\mathbf {W}||_{\infty}\leq \sigma$, for $\Delta$ in Step 1 we have $||\Delta||_2\leq ||\Delta||_F\leq (s+1)\sigma$. By Weyl's theorem, we have $|\hat \sigma_j-\sigma_j|\leq ||\Delta||_2, j=1,\cdots,s+1,$ where $\sigma_j$ is the $j$-th singular value of $DAD^\top$. Note that $\sigma_j=0, n+1\leq j \leq s+1$, we get $|\hat \sigma_j|\leq ||\Delta||_2\leq(s+1)\sigma, n+1\leq j \leq s+1$. This proves (\ref{MUSICthm1equ-1}). 

 Since $y_j\in [\frac{-(n-1)\pi}{T\Omega}, \frac{(n-1)\pi}{T\Omega}]$, by the relation between $r,s, T$, we have $\theta_{\min}\in [-\frac{\pi}{2}, \frac{\pi}{2}]$ and
\[
\theta_{\min} = r\Omega \min_{p\neq j}\Big|y_p-y_j\Big| >\frac{2\pi (s+1)r}{T}\Big(\frac{2n(s+1)}{\zeta(n)^2}\frac{\sigma}{m_{\min}}\Big)^{\frac{1}{2n-2}} \geq \pi \Big(\frac{2n(s+1)}{\zeta(n)^2}\frac{\sigma}{m_{\min}}\Big)^{\frac{1}{2n-2}},
\]
where we have used the separation condition (\ref{MUSICthm1equ0}). By (\ref{equ:MUSICthm1equ1}), we have 
\begin{align}\label{MUSICthm1equ7}
\sigma_n \geq\frac{m_{\min}\zeta(n)^2\theta_{\min}^{2n-2}}{n\pi^{2n-2}}>2(s+1)\sigma.
\end{align}
Similarly, by Weyl's theorem, $|\hat \sigma_n-\sigma_n|\leq ||\Delta||_2$. Thus, $\hat \sigma_n> 2(s+1)\sigma-||\Delta||_2\geq (s+1)\sigma$. Conclusion (\ref{MUSICthm1equ2}) then follows. \end{proof}

The reconstruction procedure is summarized in \textbf{Algorithm \ref{algo:onedsimplenumber}} below. Note that in \textbf{Algorithm \ref{algo:onedsimplenumber}}, it is required that the input integer $s$ is greater than the source number $n$. However, a suitable $s$ is not easy to estimate and large $s$ may yield a deterioration of resolution. To remediate this issue, we propose a sweeping singular-value-thresholding number detection algorithm which allows us to find the minimum (or sparsest) source number from admissible measurements; see \textbf{Algorithm \ref{algo:onedsweepingnumber}} below. We remark that since when $s=n$ the separation distance in (\ref{MUSICthm1equ0}) is near $\frac{c}{T\Omega}(\frac{\sigma}{m_{\min}})^{\frac{1}{2n-2}}$ for a small constant $c$, the resolution of \textbf{Algorithm \ref{algo:onedsweepingnumber}} attains the optimal resolution order derived in \cite{liu2021theorylse}.

\begin{algorithm}[H]\label{algo:onedsimplenumber}
	\caption{\textbf{Singular-value-thresholding number detection algorithm}}
	\textbf{Input:} Number $s$, noise level $\sigma$;\\
	\textbf{Input:} Measurement: $\mathbf{Y}=(\mathbf Y(0),\cdots, \mathbf Y(T))^\top$;\\	
	1: $r=T+1\mod 2s$,  $\mathbf{Y}_{new}=(\mathbf Y(0), \mathbf Y(r), \cdots, \mathbf Y(\omega_{2sr}))^\top$\;
	2: Formulate the $(s+1)\times(s+1)$ Hankel matrix $\mathbf H(s)$ from $\mathbf{Y}_{new}$, and
	compute the singular values of $\mathbf H(s)$ as $\hat \sigma_{1}, \cdots,\hat \sigma_{s+1}$ distributed in a decreasing manner\;
	4: Determine $n$ by $\hat \sigma_n>(s+1)\sigma$ and $\hat \sigma_{j}\leq (s+1)\sigma, j=n+1,\cdots, s+1$\;
	\textbf{Return:} $n$ 
\end{algorithm}

\begin{algorithm}[H]\label{algo:onedsweepingnumber}
	\caption{\textbf{Sweeping singular-value-thresholding number detection algorithm}}	
	\textbf{Input:} Noise level $\sigma$, measurement: $\mathbf{Y}=(\mathbf Y(0),\cdots, \mathbf Y(T))^\top$;\\
	\textbf{Input:} $n_{max}=0$;\\
	\For{$s=1: \lfloor \frac{T-1}{2}\rfloor$}{
		Input $s,\sigma, \mathbf{Y}$ to \textbf{Algorithm 1}, save the output of \textbf{Algorithm 1} as $n_{recover}$\; 
		\If{$n_{recover}>n_{max}$}{$n_{max}=n_{recover}$}
	}
	{
		\textbf{Return} $n_{max}$
	}
\end{algorithm}

\subsection{Two-dimensional sweeping singular-value-thresholding number detection algorithms}\label{section:highdnumberalgorithm}
We now derive the two-dimensional sweeping singular-value-thresholding number detection algorithm for recovering the number of sources. The strategy is considering the measurement at some proper sample points. To be specific, let the parameter set of the source be $\{(a_j, \vect y_j, \vect v_j)\}_{j=1}^n, a_j \in \mathbb C, \vect y_j\in \mathbb R^2, \vect v_j\in \mathbb R^2$ and $\vect Y_t(\vect \omega)$ in (\ref{equ:modelsetting1}) be the associated measurement. We first choose the following $\frac{n(n+1)}{2}$ vectors
\begin{equation}\label{equ:numberalgovectorlist}
\vect v(\theta_q)= (\cos{\theta_q}, \sin\theta_q)^\top\in \mathbb R^2, \ q=1,\cdots,\frac{n(n+1)}{2},
\end{equation}
where $\theta_q= \frac{q2\pi}{n(n+1)}$. For each $q$, the measurement at $\Omega \vect v(\theta_q)$ is 
\[
\vect Y_t(\Omega \vect v (\theta_q)) = \sum_{j=1}^n a_j e^{i \vect y_j^\top\vect v(\theta_q)\Omega }e^{i \tau \vect  v_j^\top\vect v(\theta_q)\Omega t} + \vect W_t(\vect v(\theta_q)\Omega), \quad t=0,\cdots, T. 
\]
This can be viewed as $\vect Y(t) = \sum_{j=1}^n b_j e^{i\tau \vect v_j^\top\vect v(\theta_q)\Omega t}+\vect W(t), \ t=0, \cdots, T$ with $b_j = e^{i \vect y_j^\top\vect v(\theta_q)\Omega}$ and $\vect W(t)=\vect W_t(\vect v(\theta_q)\Omega)$. Thus we can form the Hankel matrix $\vect H_{q}(s)$ in the same way as in (\ref{equ:hankelmatrix1}) from the above measurements. Denote $\hat \sigma_{q,j}$ the $j$-th singular value of $\vect H_{q}(s)$, we can detect the exact source number by thresholding on $\hat \sigma_{q,j}$'s under a suitable separation condition, as  shown in Theorem \ref{thm:highdMUSICthm1} below. We first present a result which is needed in the proof of Theorem \ref{thm:highdMUSICthm1}.

\begin{lem}\label{lem:highdmusicdirection}
For $\tau \vect{v}_1, \cdots, \tau \vect{v}_n  \in \mathbb R^2, n\geq 2$, let $d_{\min}:=\min_{p\neq j}||\tau \vect{v}_p-\tau \vect{v}_j||_2$. Let  \, $\vect v(\theta_q), q=1,\cdots,\frac{n(n+1)}{2}$ be the ones in (\ref{equ:numberalgovectorlist}) and $S_q$'s be the one-dimensional spaces spanned by $\vect v(\theta_q)$'s. There exists $q^*$ so that 
\[
\min_{p\neq j, 1\leq p, j \leq n}\btwonorm{\mathcal P_{S_{q^*}}(\tau \vect{v}_p)-\mathcal P_{S_{q^*}}(\tau \vect{v}_j)}\geq \frac{2 d_{\min}}{n(n+1)}.
\]
\end{lem}
\begin{proof} It is clear there are at most $\frac{n(n-1)}{2}$ different $\vect u_{pj}= \tau \vect v_p - \tau \vect v_j, 1\leq p<j\leq n$. Denote $\vect v(\theta) = (\cos \theta, \sin \theta)^\top$ and $\Delta = \frac{\pi}{n(n+1)}$. We observe that if $|\vect v(\theta)^\top\vect{u}| < ||\vect{u}||_2\sin 
\Delta, \theta\in [2\Delta,\pi]$, then $|\vect v(\theta^*)^\top\vect{u}| \geq ||\vect{u}||_2\sin 
\Delta$, for $|\theta^*-\theta|\geq 2\Delta, \theta^*\in [2\Delta,\pi]$. Define $N(\vect u, \Delta) = \{\vect v|\vect v \in \mathbb R^2, ||\vect v||_2=1, |\vect v^\top\vect{u}| < ||\vect{u}||_2\sin 
\Delta \}$. If $\vect v(\theta_{q_0})\in N(\vect u_{p_0 j_0}, \Delta)$ for some $1\leq p_0, j_0\leq n$, then $\vect v(\theta_{q})\not \in N(\vect u_{p_0 j_0}, \Delta), \forall q\neq q_0, q=1,\cdots, \frac{n(n+1)}{2}$. Since we have $\frac{n(n+1)}{2}$ different $q$'s and only $\frac{n(n-1)}{2}$ $\vect u_{pj}$'s, there must be some $\vect v(\theta_{q^*}) \not \in \cup_{p<j, 1\leq j, p\leq n}N(\vect u_{pj}, \Delta)$. Hence, 
\[
\min_{p\neq j, 1\leq p, j \leq n}||\mathcal P_{S_{q^*}}(\vect{v}_p)-\mathcal P_{S_{q^*}}(\vect{v}_j)||_2\geq d_{\min} \sin\Delta \geq d_{\min} \frac{2\Delta}{\pi},
\]
whence the lemma follows. \end{proof}

\begin{thm}\label{thm:highdMUSICthm1}
Let $n\geq 2, s\geq n$ and consider the parameter set $\{(a_j, \vect y_j, \tau \vect v_j)\}_{j=1}^n$ with $\tau \vect v_j\in B_{\frac{(n-1)\pi}{T\Omega}}^2(\vect 0), 1\leq j\leq n$. For the singular values of $\ \vect H_{q}(s)$, we have 
\begin{equation}\label{equ:highdMUSICthm1equ-1}
\hat \sigma_{q,j}\leq  (s+1)\sigma, \quad j=n+1,\cdots,s+1, \quad q=1,\cdots, \frac{n(n+1)}{2}.
\end{equation}
Moreover, if the following separation condition is satisfied
\begin{equation}\label{equ:highdMUSICthm1equ0}
d_{\min}:=\min_{p\neq j}\btwonorm{\tau \vect v_p-\tau \vect v_j}>\frac{\pi (s+1)n(n+1)}{2T\Omega}\Big(\frac{n(s+1)}{\zeta(n)^2}\frac{\sigma}{m_{\min}}\Big)^{\frac{1}{2n-2}},
\end{equation}
then there exists $q^*$ so that 
\begin{equation}\label{equ:highdMUSICthm1equ2}
\hat\sigma_{q^*,n}>(s+1)\sigma.
\end{equation}
\end{thm}
\begin{proof} Note that for each $\vect v(\theta_q)$ in (\ref{equ:numberalgovectorlist}), the $\tau \vect v_j$'s satisfy that $|\tau \vect v_j^\top\vect v(\theta_q)|\leq \frac{(n-1)\pi}{T\Omega}, j=1,\cdots,n$. By applying Theorem \ref{thm:onedMUSICthm1} on the Hankel matrix $\vect H_{q}(s)$ formulated from the measurements $\vect Y_t(\Omega \vect v (\theta_q)),$ for $t=0,\cdots, T$,  we immediately get (\ref{equ:highdMUSICthm1equ-1}). Moreover, when the separation condition (\ref{equ:highdMUSICthm1equ0}) holds, by Lemma \ref{lem:highdmusicdirection}, there is some $q^*$ so that 
\[
\min_{p\neq j, 1\leq p, j \leq n}\babs{\mathcal P_{S_{q^*}}(\vect{v}_p)-\mathcal P_{S_{q^*}}(\vect{v}_j)}\geq \frac{2 d_{\min}}{n(n+1)}> \frac{\pi (s+1)}{T\Omega}\Big(\frac{2n(s+1)}{\zeta(n)^2}\frac{\sigma}{m_{\min}}\Big)^{\frac{1}{2n-2}}.
\]
Applying Theorem \ref{thm:onedMUSICthm1} again, we get $\hat\sigma_{q^*,n}>(s+1)\sigma$.\end{proof}

\medskip
The above theorem shows that for point sources that are well-separated, we can determine the source number $n$ by 
thresholding on the singular values of the Hankel matrices $\vect H_q(s)$'s. We note that the number of required unit vectors $\vect v(\theta_q)$ is not available since $n$ is unknown. In practice, we can choose a large enough $N$, say $N\geq \frac{n(n+1)}{2}$.
We summarize our algorithm as below.

\begin{algorithm}[H]\label{algo:twodsweepingnumber}
	\caption{\textbf{Two-dimensional sweeping singular-value-thresholding number detection algorithm}}	
	\textbf{Input:} Noise level $\sigma$, measurement: $\mathbf{Y}(\vect \omega), \vect \omega \in \mathbb R^2, ||\vect \omega||_2\leq \Omega$, and $n_{\max} = 0$\\
	\textbf{Input:} A large enough $N$, and corresponding $N$ unit vectors $\vect v(\theta_q)=(\cos{\theta_q}, \sin\theta_q)^\top, \theta_q = \frac{\pi}{N}, \frac{2\pi}{N}, \cdots, \pi$ \\
	\For{$\theta_q=\frac{\pi}{N}, \frac{2\pi}{N}, \cdots, \pi$}{
		Input $\sigma$ and $Y_t(\Omega \vect v (\theta_q)), t=0,\cdots, T$ to \textbf{Algorithm 2}, save the output of \textbf{Algorithm 2} as $n_{recover}$\; 
		\If{$n_{recover}>n_{max}$}{$n_{max}=n_{recover}$}
	}
	{
	\textbf{Return} $n_{max}$
	}
\end{algorithm}

\subsection{Numerical experiments}
In this subsection, we conduct some numerical experiments to demonstrate the super-resolution ability of our number detection algorithm and the superiority of our algorithm over the static reconstruction method. 

We present an example which shows that our algorithm can recover the source number even when it is impossible to recover the source number by the static reconstruction in any frame. For simplicity, we consider $\Omega =1$, $\tau =0.2$, $\sigma =10^{-2}$ and the measurements at $5$ time steps. We construct an example where two point sources with $O(1)$ intensities are located at $\vect y_1= (0, 0.27), \vect y_2 = (0.20, 0.17)$ moving respectively at the velocities $\vect v_1= (0.14, 0.51), \vect v_2 = (0.45, 0.17)$. At each of the $5$ time steps, the locations of the two point sources are, 
\begin{align*}
	&(0, 0.27),\quad  (0.20, 0.17),\quad \text{at time step $t=0$,}\\
	&(0.028, 0.372),\quad  (0.29,	0.204),\quad \text{at time step $t=1$,}\\
	&(0.056,	0.474),\quad  (0.38,	0.238),\quad \text{at time step $t=2$,}\\
	&(0.084,	0.576),\quad  (0.47,	0.272),\quad \text{at time step $t=3$,}\\
	&(0.112,	0.678),\quad  (0.56, 0.306),\quad \text{at time step $t=4$}.
\end{align*}
Note that at each time step, the point sources are separated by a distance which is much lower than the Rayleigh limit. Thus the conventional static reconstruction cannot recover the source number and locations in each frame. However, by \textbf{Algorithm \ref{algo:twodsweepingnumber}} we recover the number of point sources. This demonstrates the superiority of our algorithm over the static reconstructions for recovering correctly the source number. 

In addition, in the above example the resolution is near $\tau ||\vect v_1-\vect v_2||_2\approx 0.1$, which is much better than the Rayleigh limit ($\pi$) of the imaging system.  The reason is that we take $T=5$ and the resolution limit in Theorem \ref{thm:trackvelocitiessupplimit0} which of order $O(\frac{1}{T\Omega}(\frac{\sigma}{m_{\min}})^{\frac{1}{2n-2}})$ indicates that a very good resolution for number detection can be achieved when we have multiple observations.

\section{Projection-based velocity recovery algorithms}\label{section:velocitysupportalgorithm}
In this section, we propose a projection-based velocity recovery algorithm in two dimensions. The algorithm can be easily extended to higher dimensions. 

\subsection{Review of one-dimensional MUSIC algorithm}\label{section:onedsupportalgorithm}

In this subsection, we review the one-dimensional MUSIC algorithm \cite{schmidt1986multiple, stoica1989music, liao2016music, liu2022measurement}. In a standard MUSIC algorithm for solving the inverse problem (\ref{equ:modelsetting2}), one first assembles the Hankel matrix $\vect H(s)$ as in (\ref{equ:hankelmatrix1}), where 

\begin{equation}\label{equ:musicequ1}
	s= \begin{cases}\frac{T}{2}, &\text{even}\  T, \\
\frac{T-1}{2}, & \text{odd} \ T,
\end{cases}
\end{equation}
then performs the singular value decomposition of $\vect H(s)$,
\[
\vect H(s) = \hat U\hat \Sigma \hat U^*=[\hat U_1\quad \hat U_2]\text{diag}(\hat \sigma_1, \hat \sigma_2,\cdots,\hat \sigma_n,\hat \sigma_{n+1},\cdots,\hat \sigma_{\hat N+1})[\hat U_1\quad \hat U_2]^*,
\]
where $\hat U_1=(\hat U(1),\cdots,\hat U(n)), \hat U_2=(\hat U(n+1),\cdots,\hat U(\hat N+1))$ with $n$ being the estimated source number (model order). The source number $n$ can be detected by many algorithms such as those in \cite{akaike1998information, wax1985detection, schwarz1978estimating, wax1989detection, chen1991detection, he2010detecting, han2013improved, liu2021theorylse}. Denote the orthogonal projection to the space $\hat U_2$ by $\hat P_2x=\hat U_2(\hat U_2^*x)$. For a test vector $\Phi(\omega)=(1, e^{i\Omega \omega},\cdots,e^{i s \Omega \omega})^\top$, one defines the MUSIC imaging functional 
\begin{align*}
%\hat R(\omega)=\frac{||\hat P_2\Phi(\omega)||_2}{||\Phi(\omega)||_2}=\frac{||\hat U_2^*\Phi(\omega)||_2}{||\Phi(\omega)||_2}, \quad
\hat J(\omega)=\frac{||\Phi(\omega)||_2}{||\hat P_2\Phi(\omega)||_2}=\frac{||\Phi(\omega)||_2}{||\hat U_2^*\Phi(\omega)||_2}.
\end{align*}
The local  maximizer of $\hat J(\omega)$ indicates the location of the point sources. In practice, one can test evenly spaced points in a specified region and plot the discrete imaging functional and then determine the source locations by detecting the peaks. A peak selection algorithm \textbf{Algorithm \ref{algo:peakselection}} is given below. Finally, we summarize the standard MUSIC algorithm in \textbf{Algorithm \ref{algo:standardmusic}} below.

\begin{algorithm}[H]
	\caption{\textbf{Standard MUSIC algorithm}}
	\textbf{Input:} Noise level $\sigma$, source number $n$\;
	\textbf{Input:} Measurements: $\mathbf{Y}=(\vect Y(0), \vect Y(1), \cdots, \vect Y(T))^\top$\;
	\textbf{Input:} Region of test points $[TS, TE]$ and spacing of test points $TPS$\;
	1: Determine $s$ by (\ref{equ:musicequ1}), formulate the $(s+1)\times (s+1)$ Hankel matrix $\vect H(s)$ as (\ref{equ:hankelmatrix1}) from $\mathbf{Y}$\;
	2: Compute the singular vectors of $\vect H(s)$ as $\hat U(1), \hat U(2),\cdots,\hat U(s+1)$ and formulate the noise space $\hat U_{2}=(\hat U(n+1),\cdots,\hat U(s +1))$\;
	3: For test points $x$ in $[TS, TE]$ evenly spaced by $TPS$, construct the test vector $\Phi(\omega)=(1,e^{i\Omega \omega}, \cdots, e^{is \Omega \omega})^\top$\; 
	4: Plot the MUSIC imaging functional $\hat J(\omega)=\frac{||\Phi(\omega)||_2}{||\hat U_2^*\Phi(\omega)||_2}$\;
	5: Select the peak locations $\hat y_j$'s in the $\hat J(x)$ by \textbf{Algorithm \ref{algo:peakselection}}\;
	\textbf{Return} $\hat y_j$'s.
	\label{algo:standardmusic}
\end{algorithm}

\begin{algorithm}[H]
    \caption{\textbf{Peak selection algorithm}}
    \textbf{Input:} Image $IMG = (f(\omega_1), \cdots, f(\omega_M))$\;
    \textbf{Input:} Peak compare range $PCR$, differential compare range $DCR$,  differential compare threshold $DCT$\;
    1: Initialize the local maximum points $LMP = [\ ]$, peak points $PP = [\ ]$\;
    2: Differentiate the image $IMG$ to get the $DIMG = (f'(\omega_1), \cdots, f'(\omega_M))$\;
    3: \For{$j=1:M$}{
    \If{$f(\omega_j) = \max(f(\omega_{j-PCR}), f(\omega_{j-PCR+1}), \cdots, f(\omega_{j+PCR}))$}{
    $LMP$ appends $\omega_j$\;
    }}
    4: \For{$\omega_j$ in $LMP$}{
    \If{
    $\max(|f'(\omega_{j-DCR})|, |f'(\omega_{j-DCR+1})|, \cdots, |f'(\omega_{j+DCR})|)\geq DCT$
    }{$PP$ appends $\omega_j$\;}
    \textbf{Return:} $PP$.
    }
    \label{algo:peakselection}
\end{algorithm}

\subsection{Projection-based MUSIC algorithm for super-resolving velocities}\label{section:highdsupportalgorithm}
In this subsection we propose a projection-based MUSIC algorithm for super-resolving the velocities in two dimensions. As indicated by the proof of Theorem \ref{thm:trackvelocitiessupplimit0}, when the velocities are well-separated in $\mathbb R^{2}$, there exist two unit vectors so that the projection of the velocities in two one-dimensional subspaces spanned by these unit vectors can be stably recovered. We can then find the original two-dimensional velocities from their projections. More precisely, let $N=\frac{(n+2)(n+1)}{2}$ and
\begin{equation}\label{equ:support2dvectors1}
\vect v(\phi)=(\cos\phi, \ \sin \phi)^\top, \quad \phi\in \Big\{\frac{\pi}{N}, \frac{2\pi}{N},\cdots, \pi \Big\}.
\end{equation}
For each $\vect v(\phi)$, we denote the space spanned by $\vect v(\phi)$ as $S(\phi)$ and call the $(\vect v_j^\top \vect v(\phi)) \vect v(\phi)$'s the projected velocities in $S(\phi)$. We consider the measurement at $\Omega \vect v(\phi)$ that 
\begin{equation}\label{equ:velocitytwodalgomeasure1}
\vect Y_t(\Omega \vect v (\phi)) = \sum_{j=1}^n a_j e^{i \vect y_j^\top\vect v(\phi)\Omega }e^{i \tau \vect  v_j^\top\vect v(\phi)\Omega t}+ \vect W_t(\vect v(\theta_q)\Omega), \quad t=0,\cdots, T. 
\end{equation}
We then recover the projected velocities in each of those one-dimensional subspaces using the one-dimensional MUSIC algorithm fed with an estimated source number. We then only consider those subspaces $S(\phi)$'s where $n$ peaks appear, i.e., $n$ projected velocities are reconstructed. We further choose two vectors from the $\vect v(\phi)$'s, denoted by $\vect v(\phi_1), \vect v(\phi_2)$, so that the recovered velocities in $S(\phi_1), S(\phi_2)$ have respectively the largest and second largest minimum separation distance. We remark that, when $N$ is large, one can require additionally that $\vect v(\phi_1)$ is not too much correlated to $\vect v(\phi_2)$, say $|\vect v(\phi_1)^\top \vect v(\phi_2) | \leq c$ for some constant $0<c<1$, to ensure that the reconstruction of two-dimensional velocities from their projections on $\vect v(\phi_1), \vect v(\phi_2)$ is stable.

We next construct the original velocities from their projections on $S(\phi_1)$ and $S(\phi_2)$. This is usually called the pair matching in the problem of direction of arrival, where ad hoc schemes \cite{zoltowski1989sensor, johnson1991operational, chen1992direction, yilmazer2006matrix} were derived to associate the estimated azimuth and elevation angles. In our paper, this can be done in the following manner. From the recovered projected velocities in $S(\phi_1)$ and $S(\phi_2)$, we first form a grid of $n^2$ points, say $\vect z_{1,1}, \vect z_{1,2}, \cdots, \vect z_{n,n}$. It can be shown that the original velocities are close to these grid points. These grid points reduce the off-the-grid recovery problem to an on-the-grid one. We then employ an enumeration method to recover the source velocities from these grid points. To be more specific, we define \begin{equation}\label{equ:velocityalgoGformula1}
G_{\phi}(\vect z_{1,j_1},\cdots, \vect z_{n,j_n}) = \big(\phi_T(e^{i \vect z_{1,j_1}^\top \vect v(\phi) \Omega}), \ \phi_T(e^{i \vect z_{2,j_2}^\top \vect v(\phi)\Omega}), \ \cdots, \ \phi_T(e^{i \vect z_{n,j_n}^\top \vect v(\phi) \Omega})\big),
\end{equation}
where $\phi_T(x)$ is defined as in (\ref{equ:defiofphi}), and solve the following optimization problem by enumeration, 
 \begin{equation}\label{equ:velocityrecoveralgominiequ1}
 \min_{\pi \in \gamma(n)}\Big(\sum_{\phi}\min_{\vect {\hat a}} \btwonorm{G_{\phi}(\vect z_{1,\pi_1},\cdots, \vect z_{n,\pi_n})\mathbf{\hat a} - \vect Y(\vect v(\phi))}\Big),
 \end{equation} 
 where $\vect Y(\vect v(\phi)) = (\vect Y_0(\vect v(\phi) \Omega), \cdots, \vect Y_T(\vect v(\phi) \Omega))^\top$, and $\gamma(n)$ is the set of all permutations of $\{1, \cdots, n\}$. We note that the computational complexity of the enumeration is low when $n$ is not large. We summarize the algorithm in \textbf{Algorithm \ref{algo:twodmusicvelocity}} below.

\begin{remark}
 We remark that since in each one-dimensional space, the MUSIC algorithm attains nearly the optimal order of the resolution, say $O(\frac{1}{T\Omega}(\frac{\sigma}{m_{\min}})^{\frac{1}{2n-1}})$, as demonstrated in \cite{liao2016music, li2021stable} numerically, our \textbf{Algorithm \ref{algo:twodmusicvelocity}} achieves the optimal order of the resolution that is derived in Theorem \ref{thm:trackvelocitiessupplimit0}. 
\end{remark}
 
\begin{remark}
We remark that the range of the recovered velocities is confined by the sampling rate in (\ref{equ:velocitytwodalgomeasure1}). Suppose we have a large range of velocities, one way to recover them is to consider all the possible projected velocities in a large range in one-dimensional space and construct a large grid. Then we recover the correct velocities in a similar way as (\ref{equ:velocityrecoveralgominiequ1}) with more measurements in $\vect Y_t$'s.  
\end{remark}

%As is discussed in the survey \cite{couture2018ultrasound}, although the MUSIC algorithm could already claim the name of super-resolution \cite{prada2003experimental}, but it only applies for a very limited number of point sources which does not correspond to the conﬁguration of many biomedical imaging tasks. 

\begin{algorithm}[H]\label{algo:twodmusicvelocity}
\caption{\textbf{Two-dimensional projection-based velocity recovery algorithm}}	
\textbf{Input:} Noise level $\sigma$, source number $n$;\\
\textbf{Input:} Measurement: $\mathbf{Y}_t(\vect \omega), \vect \omega \in \mathbb R^2, ||\vect \omega||_2\leq \Omega, t=0, \cdots, T$;\\
\textbf{Input:} Region of test points $[TS, TE]$ and spacing of test points $TPS$\;
\textbf{Input:} A large enough $N$, and corresponding $N$ unit vectors $\vect v(\phi)$ in (\ref{equ:support2dvectors1}), $ \phi \in \Big\{ \frac{\pi}{N}, \cdots, \pi\Big\}$ \\
1:\For{$\phi \in \Big\{ \frac{\pi}{N}, \frac{2\pi}{N}, \cdots, \pi\Big\}$}{
	Input $\sigma$ and $\mathbf{Y}_t(\vect v(\phi)\Omega), t=0, \cdots, T$ to \textbf{Algorithm \ref{algo:onedsweepingnumber}} to recover the projected source number $\hat n$\;
	Input $\sigma, \mathbf{Y}_t(\vect v(\phi)\Omega), t=0,\cdots, T$, $\hat n$, $[TS, TE]$ and $TPS$ to \textbf{Algorithm \ref{algo:standardmusic}}, save the output as $b_1, \cdots, b_q$\;
	The recovered projected velocities are $\vect {\hat p}_{1} = b_1 \vect v(\phi), \cdots, \mathbf {\hat p}_{q} = b_q \vect v(\phi)$\;
}
2: Choose two vectors $\vect v(\phi_1), \vect v(\phi_2)$, from those $\vect v(\phi)$'s so that $n$ projected velocities were recovered in each of the spaces $S(\phi_1), S(\phi_2)$ and the recovered projected velocities $\vect {\hat p}_{j}(\vect v(\phi_1))$'s, $\mathbf {\hat p}_{j}(\vect v_2(\phi_2))$'s have respectively the largest and second largest minimum separation distance\;
3:Construct the $n^2$ grid points $\vect z_{1,1}, \vect z_{1,2}, \cdots, \vect z_{n,n}$ by considering the intersection points of lines $\mathbf {\hat p}_{q}(\vect v(\phi_j))+ \lambda \vect g(\phi_j), \ \lambda \in \mathbb R, \ q =1,\cdots,n, \ j =1,2$, where $\vect g(\phi_j)$ is the unit vector that is perpendicular to $\vect v(\phi_j)$\;
4: Solve the following optimization problem by enumeration, 
\[\min_{\pi \in \gamma(n)}\Big(\sum_{\phi}\min_{\vect {\hat a}} \btwonorm{G_{\phi}(\vect z_{1,\pi_1},\cdots, \vect z_{n,\pi_n})\mathbf{\hat a} - \vect Y(\vect v(\phi))}\Big),
\]
where $G_{\phi}$ is defined by (\ref{equ:velocityalgoGformula1}), $\vect Y(\vect v(\phi)) = (\vect Y_0(\vect v(\phi) \Omega), \cdots, \vect Y_T(\vect v(\phi) \Omega))^\top$, and $\gamma(n)$ is the set of all permutations of $\{1,\cdots, n\}$\;
5: The minimizer $\vect z_{1,\pi_1}, \cdots, \vect z_{n, \pi_n}$'s are the recovered velocities $\tau\mathbf{\hat v}_1, \cdots, \tau \mathbf{\hat v}_n$\;
\textbf{Return} $\tau \mathbf{\hat v}_1, \cdots, \tau \mathbf{\hat v}_n$.
\end{algorithm}

\subsection{Numerical experiments}
In this subsection, we conduct some numerical experiments to demonstrate the super-resolution ability of our velocity recovery algorithm and the superiority of our algorithm over the static reconstruction method. 

We present an example which shows that our algorithm can super-resolve the velocities even when it is impossible to recover stably any locations and velocities by the static reconstruction. For simplicity, we consider $\Omega =1$, $\tau =1$, $\sigma =10^{-2}$ and the measurements at $5$ time steps. We construct an example where two point sources with $O(1)$ intensities are located at $\vect y_1=(0.22, 0.08), \vect y_2= (0.05, 0.08)$ moving respectively at the velocities $\vect v_1=(0.47, 0.11), \vect v_2=(0.58, 0.56)$. At each of the $5$ time steps, the locations of the two point sources are, 
\begin{align*}
&(0.22, 0.08),\quad  (0.05, 0.08), \quad \text{at time step $t=0$,} \\
&(0.69, 0.19),\quad  (0.63, 0.64), \quad \text{at time step $t=1$,}\\
&(1.16, 0.30),\quad  (1.21, 1.20),\quad \text{at time step $t=2$,}\\
&(1.63, 0.41),\quad  (1.79, 1.76), \quad \text{at time step $t=3$,}\\
&(2.10, 0.52),\quad  (2.37, 2.32), \quad \text{at time step $t=4$}.
\end{align*}
Note that at each time step, the point sources are separated by a distance below the Rayleigh limit. Thus the conventional static reconstruction cannot stably recover any of the locations or velocities. Since the point sources are so close in each frame, the super-resolution algorithms such as MUSIC algorithm even cannot resolve the locations. However, by \textbf{Algorithm \ref{algo:twodmusicvelocity}} we recover stably the velocities: $\vect {\hat v}_1 = (0.442, 0.120), \vect {\hat v}_2 = (0.558, 0.570)$. This demonstrates the superiority of our algorithm over the static reconstructions for recovering the velocities. 

In addition, in the above example the resolution for the velocities is near $0.4$, which is much better than the Rayleigh limit ($\pi$) of the imaging system.  The reason is that we take $5$ times observations and the resolution limit in Theorem \ref{thm:trackvelocitiessupplimit0} which of order $O(\frac{1}{T\Omega}(\frac{\sigma}{m_{\min}})^{\frac{1}{2n-1}})$ indicates a very good resolution for velocity recovery can be achieved when we have multiple observations. On the other hand, as is indicated by measurement (\ref{equ:velocitytwodalgomeasure1}), the cut-off frequency for the velocity recovery can be viewed as $T\Omega$ rather than the $\Omega$, the inherent Rayleigh limit for the velocity recovery should be $\frac{\pi}{T\Omega}$. The resolution in the example $0.4$ is near the inherent Rayleigh limit $\frac{\pi}{4}$. It is also indicated that even if the signal-to-noise ratio becomes worse, we anticipate that our algorithm or any other super-resolution algorithm for velocity recovery can resolve the velocities stably when they are separated by a distance beyond $\frac{\pi}{4}$. For our algorithm, this is demonstrated by the following one-dimensional example. 

We consider $\Omega =1$, $\tau =1$, $\sigma =0.3$, and the measurements are taken at $5$ time steps. We construct an example where two point sources with $O(1)$ intensities are located at $\vect y_1=0.296, \vect y_2= 0.038$ moving respectively at the velocities $\vect v_1=0.2, \vect v_2=1.1$. Note that the signal-to-noise ratio is small. However, we can still stably recover the velocities by the one-dimensional velocity recovery algorithm (analogously to \textbf{Algorithm \ref{algo:twodmusicvelocity}}) that yields $\vect {\hat v}_1 = 0.14, \vect {\hat v}_2 = 1.11$.

 Note that we use a one-dimensional example because the resolution of our two-dimensional algorithm is also related to the projections.

%We consider in some cases when the static reconstruction for each frame fails to recover stably the source locations, whence the velocity is impossible to be stably reconstructed. However, our algorithm can stably reconstruct all the velocities. 

% \subsection{Velocity combined trajectory reconstruction}
% In this subsection, we introduce a new way to reconstruct the trajectories of the point sources. Since in each of the static reconstruction, the point sources may not be stably recovered, the recovered locations are insufficient to reconstruct the locations and velocities of the point sources. But when we have recovered the source velocities, b

\section{Conclusions and future works}
In this paper, we have explored the resolution limits for recovering the number and value of location-velocity pairs in the dynamic reconstruction of the tracking problem. We have also derived sharp and better resolution limits for reconstructing the number and values of velocities in the dynamic reconstruction. Also, two projection-based algorithms have been introduced to super-resolve respectively the number and values of velocities. By these results, we have demonstrated certain advantages of the dynamic reconstruction in the tracking problem over the conventional static reconstructions.

Besides our research findings, our work is also a start of many new topics. Especially, from the crucial observation of measurement (\ref{equ:velocitytwodalgomeasure1}), many new algorithms can be inspired to obtain much better resolution for the velocity recovery over the static reconstruction method. In practice, the point spread function may be approximated by other different functions, such as Gaussian functions, and a large amount of point sources may be clustered together in a single image \cite{couture2018ultrasound}, which hampered the application of subspace methods. In this case, the convex optimization based algorithms are great surrogates for resolving the point sources. Applying convex optimization to (\ref{equ:velocitytwodalgomeasure1}) or its variant may enhance significantly the resolution in the practical tracking problem. In addition, developing efficient algorithms for exploring the amplitudes and the locations of point sources when the velocities are known is also an interesting topic. Note that when the velocities are stably recovered, we can reconstruct the $\hat a_j, \vect {\hat y}_j$'s by solving the minimization problem
\[
 \min_{\hat a_j, \vect{\hat y}_j}\btwonorm{\sum_{j=1}^n \hat a_j e^{ i \vect {\hat y}_j^\top\vect \omega} e^{i t\tau \vect {\hat v}_j^\top\vect \omega} -\vect Y_t(\vect \omega)}, \quad ||\vect \omega||_2\leq \Omega, t =0, \cdots, T. 
\]
The aim is to develop a tractable algorithm in order to recover stably the $\hat a_j, \vect {\hat y}_j$'s for the above minimization problem or its variants.

\section{Proofs of Theorems \ref{thm:trackhighdupperboundnumberlimit0} and \ref{thm:trackhighdupperboundsupportlimit0}}\label{sec:proofmainresults}

\subsection{The geometry of the problem}
For each time step, the noiseless measurement $\sum_{j=1}^n a_j e^{i(\vect y_j+t \tau \vect v_j)^\top\vect \omega}, ||\vect \omega||_2\leq \Omega,$ can be written as
\begin{equation}\label{equ:measurementtwodview}
\sum_{j=1}^n a_j e^{i\vect \alpha_j^\top\vect \omega_t}, 
\end{equation}
where $\vect \alpha_j = \begin{pmatrix}
\vect y_j\\
\tau \vect v_j
\end{pmatrix}, \vect \omega_t =\begin{pmatrix}
\vect \omega\\
t\vect \omega
\end{pmatrix}, \ t=0,\cdots, T, \ ||\vect \omega||_2\leq \Omega$. Define the spaces 
\begin{equation}\label{equ:spacestd}
S_t^d:=\Big\{\begin{pmatrix}
\vect v\\
t\vect v
\end{pmatrix}\Big| \vect v \in \mathbb R^d \Big\}, \quad t=0,1, \cdots.
\end{equation}
Then the measurement $\vect Y_t(\vect \omega)$ can also be written as 
\begin{equation}\label{equ:measurementsprojection}
	\vect Y(\vect \omega_t) = \sum_{j=1}^n a_j e^{i\mathcal P_{S_t^d}(\vect \alpha_j)^\top\vect \omega_t}+\vect W(\vect \omega_t), \quad  ||\vect \omega_t||_2\leq \sqrt{1+t^2}\Omega,  
\end{equation}
with $|\vect W(\vect \omega_t)|<\sigma$. Note that $S_t^d$'s are $d$-dimensional spaces. For each $t$, the above (\ref{equ:measurementsprojection}) is equivalent to a $d$-dimensional super-resolution problem in \cite{liu2021mathematicalhighd}; see Section \ref{section:resultsforstatichighd} as well. Hence, now the tacking problem can be viewed as reconstructing the $2d$-dimensional vectors $\vect \alpha_j$'s from the measurements (or super-resolution problems) in several $d$-dimensional subspaces ($S_t^d$'s) of $\mathbb R^{2d}$. By the same projection idea as the one in \cite{liu2021mathematicalhighd}, we can estimate the stability of the recovery. We prove our main results by analyzing some geometrical properties of these subspaces $S_t^d$'s. 

We first consider the case when $d=1$ and estimate the angle between adjacent lines in $\{(\omega, t\omega)^\top\}_{t=0}^T, \omega\in \mathbb R$. Note that
\[
\arctan(u)\pm \arctan(v) = \arctan\Big(\frac{u\pm v}{1\mp uv}\Big)  \quad  
(\mod \pi),\quad  uv\neq 1,  
\]
and
\[
\arctan(x) \geq  x- \frac{x^3}{3}.
\]
Thus for $t\geq 1$ we have 
\begin{equation}\label{equ:spaceangleestimate1}
\arctan(t)-\arctan(t-1) =  \arctan\Big(\frac{1}{1+t(t-1)}\Big) \geq  \frac{1}{t^2-t+1}- \frac{1}{3(t^2-t+1)^3} > \frac{1}{t^2+1}.
\end{equation}
We denote the unit vector in $S_t^1$'s as $\vect q_t = \frac{1}{\sqrt{1+t^2}}(1, t)^\top$. Define $\angle(\vect q_t, \vect q_j)$ the angle between vectors $\vect q_t, \vect q_j$ that 
\[
\angle(\vect q_t, \vect q_j) = \arccos\Big(\frac{\vect q_t^\top \vect q_j}{||\vect q_t||_2 ||\vect q_j||_2}\Big).
\]
By the above observation, we have
\begin{equation}\label{equ:spaceangle0}
\angle(\vect q_t, \vect q_j) = \arctan(t) - \arctan(j) > \sum_{k=j+1}^{t}\frac{1}{1+k^2}, \quad j<t. 
\end{equation}
We first estimate the projection to these one-dimensional subspaces $S_t^1$'s.  
\begin{lem}\label{lem:spaceangleestimate2}
For $t=0,1,\cdots$ and a $\vect u \in \mathbb R^2$, if 
\[
\btwonorm{\mathcal P_{S_t^1}(\vect u) } < \frac{1}{2\pi(1+t^2)}||\vect u||_2,
\]
then for $0\leq j<  t$, we have 
\[
\btwonorm{\mathcal P_{S_j^1}(\vect u)} \geq \frac{1}{2\pi(1+j^2)} ||\vect u||_2.
\]
\end{lem}
\begin{proof}
We first prove the lemma for $j=t-1$. By (\ref{equ:spaceangle0}), we have 
\begin{equation}\label{equ:spaceangle1}
\frac{1}{t^2+1} < \angle(\vect q_t, \vect q_{t-1})\leq \frac{\pi}{4}. 
\end{equation}
Since $||\mathcal P_{S_t^1}(\vect u) ||_2 < \frac{1}{2\pi}\frac{1}{1+t^2} ||\vect u||_2$ and $|\vect q_t^\top\vect u| = ||\vect u||_2 |\cos(\angle(\vect q_t, \vect u))|$, we have $|\cos(\angle(\vect q_t, \vect u))|< \frac{1}{2\pi(1+t^2)}$. Thus $\frac{2}{\pi}|\frac{\pi}{2}- \angle(\vect q_t, \vect u)|\leq |\sin (\frac{\pi}{2}- \angle(\vect q_t, \vect u))| = |\cos(\angle(\vect q_t, \vect u))|< \frac{1}{2\pi(1+t^2)}$, and consequently,
\[
\frac{\pi}{2}- \frac{1}{4(1+t^2)}<\angle(\vect q_t, \vect u)< \frac{\pi}{2}+\frac{1}{4(1+t^2)}, or \ -\frac{\pi}{2}- \frac{1}{4(1+t^2)}<\angle(\vect q_t, \vect u)< -\frac{\pi}{2}+\frac{1}{4(1+t^2)}.
\]
Without loss of generality, we only consider the first case. Together with (\ref{equ:spaceangle1}), we have 
\[
0<\angle(\vect q_{t-1}, \vect u) < \frac{\pi}{2}- \frac{3}{4(1+t^2)}, or \  \frac{\pi}{2}+\frac{3}{4(1+t^2)}< \angle(\vect q_{t-1}, \vect u) < \pi.
\]
Thus, 
\[ 
||\mathcal P_{S_{t-1}}(\vect u) ||_2 \geq \sin\Big(\frac{3}{4(1+t^2)}\Big)||\vect u||_2. 
\]
By $\sin(x) \geq \frac{2x}{\pi}$ and $\frac{3}{1+t^2}\geq \frac{1}{1+(t-1)^2}$, we then show that
\[
||\mathcal P_{S_{t-1}}(\vect u) ||_2 \geq \frac{1}{2\pi}\frac{1}{1+(t-1)^2} ||\vect u||_2.
\]
For $j<t-1$, since $\angle(\vect q_j, \vect q_t)> \sum_{k=j+1}^t \frac{1}{1+k^2}$, we have 
\[
||\mathcal P_{S_{j}}(\vect u) ||_2 \geq \sin\Big(\sum_{k=j+1}^{t-1} \frac{1}{1+k^2}+\frac{3}{4}\frac{1}{(1+t^2)}\Big)||\vect u||_2. 
\]
Thus
\[
||\mathcal P_{S_{j}}(\vect u) ||_2 \geq \frac{2}{\pi}\Big(\sum_{k=j+1}^{t-1} \frac{1}{1+k^2}+\frac{3}{4}\frac{1}{(1+t^2)}\Big)||\vect u||_2.
\]
It is clear that $\frac{2}{\pi}\Big(\sum_{k=j+1}^{t-1} \frac{1}{1+k^2}+\frac{3}{4}\frac{1}{(1+t^2)}\Big)\geq \frac{1}{2\pi}\frac{1}{1+j^2}$, whence the lemma is proved. 
\end{proof}

We now consider the general $S_t^d$'s. 
\begin{lem}\label{lem:spaceangleestimate3}
For any $\vect q_t \in S_t^d$ and $\vect q_j \in S_j^d, j<t$ satisfying $\vect q_t^\top\vect q_j\geq 0$, we have 
\[
\sum_{k=j+1}^t \frac{1}{1+k^2}< \angle(\vect q_t, \vect q_j)\leq \frac{\pi}{2}.
\]
\end{lem}
\begin{proof}
Let $\vect {\hat q}_t$ and $\vect {\hat q}_j$ be two unit vectors in $\mathbb R^d$ with $\vect {\hat q}_t^\top \vect {\hat q}_j>0$. Then we have
\[
\cos\Big(\angle \Big(\begin{pmatrix}
\vect {\hat q}_t\\
t\vect {\hat q}_t
\end{pmatrix}, \begin{pmatrix}
\vect {\hat q}_j\\
j\vect {\hat q}_j
\end{pmatrix} \Big)\Big) = \begin{pmatrix}
\vect {\hat q}_t\\
t\vect {\hat q}_t
\end{pmatrix}^\top \begin{pmatrix}
\vect {\hat q}_j\\
j\vect {\hat q}_j
\end{pmatrix} / \Big(\btwonorm{\begin{pmatrix}
\vect {\hat q}_t\\
t\vect {\hat q}_t
\end{pmatrix}} \btwonorm{\begin{pmatrix}
\vect {\hat q}_j\\
j\vect {\hat q}_j
\end{pmatrix}} \Big) = \frac{(1+tj)\vect {\hat q}_t^\top \vect {\hat q}_j}{\sqrt{(1+t^2)(1+j^2)}}. 
\]
Thus, introducing $\vect q_t = \begin{pmatrix}
\vect {\hat q}_t\\
t\vect {\hat q}_t
\end{pmatrix}$ and $\vect q_j = \begin{pmatrix}
\vect {\hat q}_j\\
j\vect {\hat q}_j
\end{pmatrix}$ yields
\[
0\leq  \cos(\angle(\vect q_t, \vect q_j)) \leq  \frac{1+tj}{\sqrt{(1+t^2)(1+j^2)}}.
\]
Let $\vect q = \vect q_t - \vect q_j$. Then, it follows that 
\[
||\vect q||_2 = \sqrt{||\vect q_t||_2^2+ ||\vect q_j||_2^2 - 2||\vect q_t||_2||\vect q_j||_2\cos(\angle(\vect q_t, \vect q_j))} \geq t-j. 
\]
Hence, by considering the two-dimensional space spanned by $\vect q_t,\vect q_j$ and using the same idea as the one in proving (\ref{equ:spaceangleestimate1}) and (\ref{equ:spaceangle0}), we show that for any $\vect q_t \in S_t^d$ and $\vect q_j \in S_j^d, j<t$, we have 
\[
\sum_{k=j+1}^t \frac{1}{1+k^2} < \angle(\vect q_t, \vect q_j)\leq \frac{\pi}{2}.
\]
\end{proof}

We now extend Lemma \ref{lem:spaceangleestimate2} to the multi-dimensional case.
\begin{lem}\label{lem:spaceangleestimate4}
For $t=0,1,\cdots$ and $\vect u \in \mathbb R^{2d}$, if 
\[
||\mathcal P_{S_t^d}(\vect u) ||_2 < \frac{1}{2\pi(1+t^2)}||\vect u||_2,
\]
then for $0\leq j< t$, we have 
\[
||\mathcal P_{S_j^d}(\vect u) ||_2 \geq \frac{1}{2\pi(1+j^2)} ||\vect u||_2.
\]
\end{lem}
\begin{proof}
For a fixed $\vect u \in \mathbb R^{2d}$, let 
\[
\vect q_t = \mathcal P_{S_t^d}(\vect u)/ ||\mathcal P_{S_t^d}(\vect u)||_2,
\]
and 
\[
\vect q_j = \mathcal P_{s_j^d}(\vect u)/ ||\mathcal P_{s_j^d}(\vect u)||_2,
\] 
if $\vect q_t^\top \vect q_j \geq 0$. Otherwise, set $\vect q_j =- \mathcal P_{s_j^d}(\vect u)/ ||\mathcal P_{s_j^d}(\vect u)||_2$. Under the condition stated in the lemma,  we have $|\vect u^\top \vect q_t|< \frac{1}{2\pi(1+t^2)}||\vect u||_2$. Also, by Lemma \ref{lem:spaceangleestimate3} we get $$\frac{\pi}{2}\geq \angle(\vect q_t, \vect q_j)> \sum_{k=j+1}^t \frac{1}{1+k^2}.$$ Considering the two-dimensional space spanned by $\vect q_t, \vect q_j$, similarly to proof of Lemma \ref{lem:spaceangleestimate2},  we obtain
\[
\vect u^\top \vect q_j\geq \frac{1}{2\pi(1+j^2)}||\vect u||_2,
\]
which implies that $||\mathcal P_{S_j^d}(\vect u) ||_2 \geq \frac{1}{2\pi(1+j^2)} ||\vect u||_2$.
\end{proof}

We next present two auxiliary lemmas that are used in the proof of main results. 
\begin{lem}\label{lem:projectlengthestimate1}
	For a vector $\vect{u}\in \mathbb R^2$, and two unit vectors $\vect{q}_1, \vect{q}_2\in \mathbb R^2$ satisfying $0 \leq |\vect q_1^\top\vect q_2| \leq \cos\theta$, we have 
	\begin{equation}\label{equ:projectlengthestimate1}
	|\vect q_1^\top\vect u|^2+|\vect q_2^\top\vect u|^2  \geq (1-\cos(\theta))||\vect u||_2^2.
	\end{equation}
\end{lem}
\begin{proof}
We have 
\begin{equation}\label{equ:projectionlower3}
\btwonorm{(\vect{q}_1\cdot \vect u,\  \vect{q}_2\cdot \vect u)^\top}^2
= \btwonorm{\begin{pmatrix}
	\vect{q}_1^\top\\
	\vect{q}_2^\top
	\end{pmatrix}\cdot \vect u}^2 \geq \sigma_{\min}^2(\begin{pmatrix}
\vect{q}_1^\top\\
\vect{q}_2^\top
\end{pmatrix})||\vect u||_2^2 \geq (1-\cos \theta)||\vect u||_2^2,
\end{equation}
where the last inequality follows from calculating $\sigma_{\min}(\begin{pmatrix}
\vect{q}_1^\top\\
\vect{q}_2^\top
\end{pmatrix})$.
\end{proof}

\begin{lem}\label{lem:projectlengthestimate2}
For a vector $\vect u \in \mathbb R^d$, and for spaces $S_t^d, S_j^d, j<t$, we have 
\[
\btwonorm{\mathcal P_{S_t^d}(\vect u)}^2 + \btwonorm{\mathcal P_{S_j^d}(\vect u)}^2 \geq \Big(1-\cos(\theta)\Big) ||\vect u||_2^2,
\]
where $\theta = \sum_{k=j+1}^{t}\frac{1}{1+k^2}$. 
\end{lem}
\begin{proof} We first construct the basis of $S_t^d$ and $S_j^d$. Let 
\[
\vect e_1 := \begin{pmatrix}
1\\
0\\
\vdots\\
0
\end{pmatrix},\quad 
\vect e_2 := \begin{pmatrix}
0\\
1\\
\vdots\\
0
\end{pmatrix},\quad \cdots, \quad
\vect e_d := \begin{pmatrix}
0\\
0\\
\vdots\\
1
\end{pmatrix}. 
\]
It is easy to verify that the vectors
\[
\vect e_{1,t} : = \frac{1}{\sqrt{1+t^2}}\begin{pmatrix}
\vect e_1\\
t\vect e_1
\end{pmatrix}, \
\vect e_{2,t} : = \frac{1}{\sqrt{1+t^2}}\begin{pmatrix}
\vect e_2\\
t\vect e_2
\end{pmatrix},\ \cdots,\ 
\vect e_{d,t} : = \frac{1}{\sqrt{1+t^2}}\begin{pmatrix}
\vect e_d\\
t\vect e_d
\end{pmatrix},
\]
form an orthonormal basis of $S_t^d$. In the same manner, the vectors
\[
\vect e_{1,j} : = \frac{1}{\sqrt{1+j^2}}\begin{pmatrix}
\vect e_1\\
j\vect e_1
\end{pmatrix}, \
\vect e_{2,j} : = \frac{1}{\sqrt{1+j^2}}\begin{pmatrix}
\vect e_2\\
j\vect e_2
\end{pmatrix},\ \cdots,\ 
\vect e_{d,j} : = \frac{1}{\sqrt{1+j^2}}\begin{pmatrix}
\vect e_d\\
j\vect e_d
\end{pmatrix}, 
\]
form an orthonormal basis of $S_j^d$. Also, we have 
\begin{equation}\label{equ:spaceangleestimate2}
|\angle(\vect e_{p,t}, \vect e_{q,j})| = \frac{\pi}{2}, \ p\neq q, \ \text{and} \  \sum_{k=j+1}^t\frac{1}{1+k^2}<\angle(\vect e_{q,t}, \vect e_{q,j})< \frac{\pi}{2}, 
\end{equation}
where the second inequality is from Lemma \ref{lem:spaceangleestimate3}. Thus, if we denote by $V_{q}$ the two-dimensional space spanned by $\{\vect e_{q,t}, \vect e_{q,j}\}$, then the $V_q$'s are orthogonal to each other. Moreover, we have 
\begin{align*}
&\btwonorm{\mathcal P_{S_t^d}(\vect u)}^2 + \btwonorm{\mathcal P_{S_j^d}(\vect u)}^2 = \Big|\vect u^\top \vect e_{1,t}\Big|^2+\cdots+\Big|\vect u^\top \vect e_{d,t}\Big|^2 + \Big|\vect u^\top \vect e_{1,j}\Big|^2+\cdots+\Big|\vect u^\top \vect e_{d,j}\Big|^2 \\
=& \Big(\Big|\vect u^\top \vect e_{1,t}\Big|^2+ \Big|\vect u^\top \vect e_{1,j}\Big|^2\Big)+ \Big(\Big|\vect u^\top \vect e_{2,t}\Big|^2+ \Big|\vect u^\top \vect e_{2,j}\Big|^2\Big)+\cdots + \Big(\Big|\vect u^\top \vect e_{d,t}\Big|^2+ \Big|\vect u^\top \vect e_{d,j}\Big|^2\Big)\\
\geq & \Big(1-\cos\Big(\sum_{k=j+1}^{t}\frac{1}{1+k^2}\Big)\Big) \sum_{q=1}^d\btwonorm{\mathcal P_{V_q}(\vect u)}^2 \quad \Big( \text{by Lemma \ref{lem:projectlengthestimate1} and (\ref{equ:spaceangleestimate2})}\Big)\\
= & \Big(1-\cos\Big(\sum_{k=j+1}^{t}\frac{1}{1+k^2}\Big)\Big) ||\vect u||_2^2  \quad \Big( \text{since $V_q$'s are orthogonal to each other}\Big).\\
\end{align*}
\end{proof}

\subsection{Results for the static super-resolution problem}\label{section:resultsforstatichighd}
In this subsection, we review and restate the stability results in \cite{liu2021mathematicalhighd} for the static super-resolution problem in multi-dimensional spaces. These results are useful to prove stability results for the dynamic super-resolution problem.   

In the super-resolution problem of single snapshot case, the source is the discrete measure $\mu = \sum_{j=1}^n a_j \delta_{\vect y_j}, \vect y_j\in \mathbb R^k$ with $\min_{1\leq j\leq n}|a_j|=m_{\min}>0$. The measurement $\vect Y$ is the noisy Fourier data of $\mu$ in a bounded domain:
\begin{equation}\label{equ:staticmeasurement1}
\vect Y(\vect \omega) = \sum_{j=1}^na_je^{i \vect y_j^\top \vect  \omega} +\vect W( \vect  \omega), \quad \ ||\vect  \omega||_2 \leq \Omega, \vect \omega \in \mathbb R^k,
\end{equation}
with $|\vect W(\vect \omega)|<\sigma$. The inverse problem is to recover the discrete measure from the set of $\sigma$-admissible measures of $\vect Y$ defined below.

\begin{defi}{\label{def:sigmaadmissiblemeasure}}
	Given the measurement $\mathbf Y$, we say that $\mu = \sum_{j=1}^k \hat a_j \delta_{\vect {\hat y}_j}, \vect {\hat y}_j\in \mathbb R^k$ is a $\sigma$-admissible measure of $\vect Y$ if
	\[
	\babs{\sum_{j=1}^k\hat a_j e^{i \vect {\hat y}_j^\top \vect \omega} -\vect Y(\vect \omega)}< \sigma, \ ||\vect \omega||_2 \leq \Omega.
	\]
\end{defi}

From \cite{liu2021mathematicalhighd},
we have the following stability results for the recovery of source number and locations.
\begin{thm}\label{thm:statichighdupperboundnumberlimit0}
	Let the measurement $\mathbf Y$ in (\ref{equ:staticmeasurement1}) be generated by a $n$-sparse measure $\mu =\sum_{j=1}^{n}a_j\delta_{\mathbf y_j}, \vect y_j \in B_{\frac{(n-1)\pi}{2\Omega}}^{k}(\vect 0)$. Let $n\geq 2$ and assume that the following separation condition is satisfied 
	\begin{equation}\label{equ:highdupperboundnumberlimit1}
	\min_{p\neq j, 1\leq p, j\leq n}\btwonorm{\mathbf y_p- \mathbf y_j}\geq \frac{4.4\pi e \ (\pi/2)^{k-1} (n(n-1)/\pi)^{\xi(k-1)}}{\Omega }\Big(\frac{\sigma}{m_{\min}}\Big)^{\frac{1}{2n-2}},
	\end{equation}
	where $\xi(k-1)$ is defined by (\ref{equ:defineofxi}). Then there does not exist any $\sigma$-admissible measures of \,$\mathbf Y$ with less than $n$ supports.
\end{thm}

\begin{thm}\label{thm:statichighdupperboundsupportlimit0}
	Let $n\geq 2$. Let the measurement $\vect Y$ in (\ref{equ:staticmeasurement1}) be generated by a $n$-sparse measure $\mu=\sum_{j=1}^{n}a_j \delta_{\vect y_j}, \vect y_j \in B_{\frac{(n-1)\pi}{2\Omega}}^{k}(\vect 0)$ in the $k$-dimensional space. Assume that
	\begin{equation}\label{equ:highdsupportlimithm0equ0}
	d_{\min}:=\min_{p\neq j}\Big|\Big|\mathbf y_p-\mathbf y_j\Big|\Big|_2\geq \frac{5.88\pi e 4^{k-1}\Big((n+2)(n+1)/2\Big)^{\xi(k-1)}}{\Omega }\Big(\frac{\sigma}{m_{\min}}\Big)^{\frac{1}{2n-1}}, \end{equation}
	where $\xi(k-1)$ is defined as in (\ref{equ:defineofxi}). If $\hat \mu=\sum_{j=1}^{n}\hat a_{j}\delta_{\mathbf{\hat y}_j}$ supported on $B_{\frac{(n-1)\pi}{2\Omega}}^{k}(\vect 0)$ is a $\sigma$-admissible measure of $\vect Y$, then after reordering the $\hat y_j$'s, we have 
	\[
	\btwonorm{\vect {\hat y}_j - \vect y_j}< \frac{d_{\min}}{2}, \quad j=1, \cdots,n.
	\]
	Moreover, we have 
	\begin{equation}\label{equ:highdsupportlimithm0equ1}
	\Big|\Big|\mathbf {\hat y}_j-\mathbf y_j\Big|\Big|_2\leq \frac{C(k, n)}{\Omega}SRF^{2n-2}\frac{\sigma}{m_{\min}}, \quad 1\leq j\leq n,
	\end{equation}
	where $SRF:= \frac{\pi}{d_{\min}\Omega}$ is the super-resolution factor and 
	\[
	C(k, n)=\big(4^{k-1}((n+2)(n+1)/2)^{\xi(k-1)}\big)^{2n-1}n2^{4n-2}e^{2n}\pi^{-\frac{1}{2}}.
	\]
\end{thm}

\begin{remark}
Note that here we have $(n+2)(n+1)/2$ as a factor in the separation condition (\ref{equ:highdsupportlimithm0equ0}), which is different from the one in \cite{liu2021mathematicalhighd}. This is due to a minor error in  \cite{liu2021mathematicalhighd} where $(n+2)(n-1)/2$ is used instead of $(n+2)(n+1)/2$.
\end{remark}

\subsection{Proof of Theorem \ref{thm:trackhighdupperboundnumberlimit0}}
\begin{proof}
Note that there are at most $\frac{n(n-1)}{2}$ different vectors of the form $\vect u_{pq}= \vect \alpha_p -\vect \alpha_q, p<q$ with $\vect \alpha_{p} = \begin{pmatrix}
\vect y_p\\
\tau \vect v_p
\end{pmatrix}$. Because we have 
\begin{equation}\label{equ:trackproofnumberthm1equ0}
\begin{aligned}
d_{\min}:=& \min_{p\neq q} \btwonorm{\vect u_{pq}}=\min_{p\neq q} \btwonorm{\vect \alpha_p -\vect \alpha_q}\\
\geq & \frac{8.8 e \pi^2  \sqrt{(\frac{(n-1)n}{2})^2+1} (\pi/2)^{d-1} (n(n-1)/\pi)^{\xi(d-1)}}{\Omega }\Big(\frac{\sigma}{m_{\min}}\Big)^{\frac{1}{2n-2}}, 
\end{aligned}
\end{equation}
we also have, 
\begin{equation}\label{equ:trackproofnumberthm1equ1}
\min_{p\neq q} \btwonorm{\vect u_{p,q}} \geq \frac{8.8 e \pi^2 \sqrt{t^2+1}(\pi/2)^{d-1}(n(n-1)/\pi)^{\xi(d-1)}}{\Omega }\Big(\frac{\sigma}{m_{\min}}\Big)^{\frac{1}{2n-2}}, \ \forall t\leq \frac{(n-1)n}{2}. 
\end{equation}
Let $\Delta = \frac{4.4e\pi(\pi/2)^{d-1}(n(n-1)/\pi)^{\xi(d-1)}}{\Omega}\Big(\frac{\sigma}{m_{\min}}\Big)^{\frac{1}{2n-2}}$. We define that $N(\vect u_{pq}, \Delta) : = \Big\{S_{t}^d\Big| S_{t}^d \in (\ref{equ:spacestd}), t=0,\cdots, \frac{(n-1)n}{2}, ||\mathcal P_{S_t^d}(\vect u_{pq})||_2 < \frac{1}{\sqrt{t^2+1}} \Delta\Big\}$. If $S_t^d\in N(\vect u_{pq}, \Delta)$, by (\ref{equ:trackproofnumberthm1equ1}) we have $||\mathcal P_{S_t^d}(\vect u_{pq})||_2< \frac{1}{2\pi(1+t^2)}||\vect u_{pq}||_2$. By Lemma \ref{lem:spaceangleestimate4} and (\ref{equ:trackproofnumberthm1equ1}), we have, for $j<t$, 
\[
\btwonorm{\mathcal P_{S_j^d}(\vect u_{pq})}\geq \frac{1}{2\pi(1+j^2)}||\vect u_{pq}||_2 \geq \frac{1}{\sqrt{j^2+1}} \Delta.
\]
Thus if $S_t^d \in N(\vect u_{pq}, \Delta)$, for $j<t$, $S_j^d\not\in N(\vect u_{pq}, \Delta)$. Now we start by considering $t=\frac{(n-1)n}{2}$. If for some $\vect u_{p_0q_0}$, $S_{\frac{(n-1)n}{2}}^d$ is in $N(\vect u_{p_0q_0}, \Delta)$, then other $S_j^d, j<\frac{(n-1)n}{2}$ are all not in $N(\vect u_{p_0q_0}, \Delta)$. If there is no such $\vect u_{p_0q_0}$ so that $S_{\frac{(n-1)n}{2}}^d \in N(\vect u_{p_0q_0}, \Delta)$, then $\btwonorm{\mathcal P_{S_{\frac{(n-1)n}{2}}^d}(\vect u_{pq})} \geq  \frac{1}{\sqrt{(\frac{n(n-1)}{2})^2+1}} \Delta$ for all $\vect u_{pq}$'s. We can continue the process for $t= \frac{(n-1)n}{2}-1, \cdots, t=0$. Since we have $1+\frac{(n-1)n}{2}$ different $S_t^d$'s and at most $\frac{n(n-1)}{2}$ different $\vect u_{pq}$'s, and 
\[
1+\frac{n(n-1)}{2}-\frac{n(n-1)}{2} = 1,
\]
by the above process, we can find at least one $S_t^d$ so that 
\begin{equation}\label{equ:trackproofnumberthm1equ2}
\btwonorm{\mathcal P_{S_t^d}(\vect u_{pq})} \geq \frac{1}{\sqrt{t^2+1}} \Delta
\end{equation}
for all $\vect u_{pq}$'s. We consider this specific $t$ in the following argument. In the space $S_t^d$, we have the noisy measurement that 
\[
\sum_{j=1}^{n} a_je^{i \mathcal P_{S_t^d}(\vect \alpha_j)^\top \vect  \omega_t} + \vect W(\vect \omega_t), 
\]
where $||\vect \omega_t||_2\leq \sqrt{1+t^2}\Omega$. On the other hand, since $\vect \alpha_j\in B_{\frac{\pi}{(n+1)\Omega}}^{2d}(\vect 0)$, we have 
\[
||\mathcal P_{S_t^d}(\vect \alpha_j)||_2 \leq ||\vect \alpha_j||_2 \leq \frac{\pi}{(n+1)\Omega} = \frac{(n-1)\pi}{(n-1)(n+1)\Omega}\leq \frac{(n-1)\pi}{2\sqrt{1+t^2}\Omega},
\] 
where the last inequality is because $t\leq \frac{n(n-1)}{2}$ and $\frac{1}{2\sqrt{1+t^2}}\geq \frac{1}{2\sqrt{1+(\frac{n(n-1)}{2})^2}}\geq \frac{1}{(n-1)(n+1)}$. 
Also, by (\ref{equ:trackproofnumberthm1equ2}) we have 
\[
\min_{p\neq q} \btwonorm{\mathcal P_{S_t^d}(\vect \alpha_p) - \mathcal P_{S_t^d}(\vect \alpha_q) }\geq \frac{4.4e\pi(\pi/2)^{d-1} (n(n-1)/\pi)^{\xi(d-1)}}{\sqrt{t^2+1}\Omega} (\frac{\sigma}{m_{\min}})^{\frac{1}{2n-2}}. 
\]
Thus now we can apply Theorem \ref{thm:statichighdupperboundnumberlimit0}. 
By Theorem \ref{thm:statichighdupperboundnumberlimit0}, there is no $(n-1)$ vectors $\vect {\tilde \alpha}_j\in S_t^d$ and $(n-1)$ $\tilde a_j$'s so that 
\[
\babs{\sum_{j=1}^{n-1}\tilde a_j e^{i\vect {\tilde \alpha}_j^\top \vect \omega_t}-\sum_{j=1}^{n} a_j e^{i\mathcal P_{S_t^d}(\vect \alpha_j)^\top \vect \omega_t} + \vect W(\vect \omega_t)}<\sigma,\quad ||\vect \omega_t ||_2\leq \sqrt{t^2+1}\Omega.  
\]
Based on (\ref{equ:measurementsprojection}) for measurements $\vect Y_t$'s, the above argument proves that there is no such $(n-1)$-sparse $\sigma$-admissible parameter set $\{(\hat a_j, \vect {\hat y}_j, \vect {\hat v}_j)\}_{j=1}^{n-1}$ of $\vect Y_t$'s.	
\end{proof}

\subsection{Proof of Theorem \ref{thm:trackhighdupperboundsupportlimit0}}
\begin{proof}
Note that there are at most $\frac{n(n-1)}{2}$ different vectors of the form $\vect u_{pq}= \vect \alpha_p -\vect \alpha_q, p<q$ with $\vect \alpha_{p} = \begin{pmatrix}\vect y_p\\
\tau \vect v_p\end{pmatrix}$. Since
\begin{equation}\label{equ:trackprojectlem1equ0}
\begin{aligned}
d_{\min}:=& \min_{p\neq q} \btwonorm{\vect u_{pq}}=\min_{p\neq q} \btwonorm{\vect \alpha_p -\vect \alpha_q}\\ 
\geq & \frac{11.76e\pi^2\sqrt{(\frac{(n+1)n}{2})^2+1}4^{d-1}((n+2)(n+1)/2)^{\xi(d-1)}}{\Omega} (\frac{\sigma}{m_{\min}})^{\frac{1}{2n-1}}, 
\end{aligned}
\end{equation}
we also obtain that for all $t\leq \frac{(n+1)n}{2}$, 
\begin{equation}\label{equ:trackprojectlem1equ1}
\min_{p\neq q} \btwonorm{\vect u_{pq}}=\min_{p\neq q} \btwonorm{\vect \alpha_p -\vect \alpha_j} \geq \frac{11.76e\pi^2\sqrt{t^2+1}4^{d-1}((n+2)(n+1)/2)^{\xi(d-1)}}{\Omega} (\frac{\sigma}{m_{\min}})^{\frac{1}{2n-1}}. 
\end{equation}
Let $\Delta = \frac{5.88e\pi4^{d-1}((n+2)(n+1)/2)^{\xi(d-1)}}{\Omega}\Big(\frac{\sigma}{m_{\min}}\Big)^{\frac{1}{2n-1}}$. We define $N(\vect u_{pq}, \Delta)$ as $\Big\{S_{t}^d| S_t^d\in (\ref{equ:spacestd}),  t=0,\cdots, \frac{(n+1)n}{2},$\\
 $||\mathcal P_{S_t^d}(\vect u_{pq})||_2 < \frac{1}{\sqrt{t^2+1}} \Delta\Big\}$. If $S_t^d\in N(\vect u_{pq}, \Delta)$, by (\ref{equ:trackprojectlem1equ1}) we have $||\mathcal P_{S_t^d}(\vect u_{pq})||_2< \frac{1}{2\pi(1+t^2)}||\vect u_{pq}||_2$. By Lemma \ref{lem:spaceangleestimate4} and (\ref{equ:trackprojectlem1equ1}), we have, for $j<t$, 
\[
\btwonorm{\mathcal P_{S_j^d}(\vect u_{pq})}\geq \frac{1}{2\pi(1+j^2)}||\vect u_{pq}||_2 \geq \frac{1}{\sqrt{j^2+1}} \Delta.
\]
Thus if $S_t^d \in N(\vect u_{pq}, \Delta)$, for $j<t$, $S_j^d\not\in N(\vect u_{pq}, \Delta)$. Again, we start by considering $t=\frac{(n+1)n}{2}$. If for some $\vect u_{pq}$, $S_{\frac{(n+1)n}{2}}^d$ is in $N(\vect u_{pq}, \Delta)$, then other $S_j^d, j<\frac{n(n+1)}{2}$ are all not in $N(\vect u_{pq}, \Delta)$. If there is no such $\vect u_{pq}$ so that $S_{\frac{(n+1)n}{2}}^d \in N(\vect u_{pq}, \Delta)$, then $\btwonorm{\mathcal P_{S_{\frac{(n+1)n}{2}}^d}(\vect u_{pq})} \geq  \frac{1}{\sqrt{(\frac{n(n+1)}{2})^2+1}} \Delta$ for all $\vect u_{pq}$'s. We can continue the discussion for $t= \frac{n(n+1)}{2}-1, \cdots, t=0$. Since we have $1+\frac{n(n+1)}{2}$ different $S_t^d$'s and at most $\frac{n(n-1)}{2}$ different $\vect u_{pq}$'s, and 
\[
1+\frac{n(n+1)}{2}-\frac{n(n-1)}{2} = n+1,
\]
by the above process, we can find at least $n+1$ different $S_t^d$'s so that 
\begin{equation}\label{equ:trackonedthm1equ1}
\btwonorm{\mathcal P_{S_t^d}(\vect u_{pq})} \geq \frac{1}{\sqrt{t^2+1}} \Delta
\end{equation}
for all $\vect u_{pq}$'s. We only consider these $t$'s in the sequel. 

In the space $S_t^d$, we have the noisy measurement given by
\[
\sum_{j=1}^{n} e^{i\mathcal P_{S_t^d}(\vect \alpha_j)^\top \vect \omega_t} + \vect W(\vect \omega_t), 
\]
where $\vect \omega_t = \begin{pmatrix}\vect \omega\\
t\vect \omega
\end{pmatrix}$,$||\vect \omega_t||_2\leq \sqrt{1+t^2}\Omega$. On the other hand, since $\vect \alpha_j\in B_{\frac{(n-1)\pi}{n(n+2)\Omega}}^{2d}(\vect 0)$, we have 
\[
||\mathcal P_{S_t^d}(\vect \alpha_j)||_2 \leq ||\vect \alpha_j||_2\leq  \frac{(n-1)\pi}{n(n+2)\Omega}\leq \frac{(n-1)\pi}{2\sqrt{1+t^2}\Omega},
\] 
where the last inequality is because $t\leq \frac{n(n+1)}{2}$ and $\frac{1}{2\sqrt{1+t^2}}\geq \frac{1}{2\sqrt{1+(\frac{n(n+1)}{2})^2}}\geq \frac{1}{n(n+2)}$.  Thus now we can apply Theorem \ref{thm:statichighdupperboundsupportlimit0}. By Theorem \ref{thm:statichighdupperboundsupportlimit0}, when (\ref{equ:trackonedthm1equ1}) holds, i.e., 
\[
\min_{p\neq q} \btwonorm{\mathcal P_{S_t^d}(\vect \alpha_p)-\mathcal P_{S_t^d} (\vect \alpha_q)}\geq \frac{5.88e\pi4^{d-1}((n+2)(n+1)/2)^{\xi(d-1)}}{\sqrt{t^2+1}\Omega} \Big(\frac{\sigma}{m_{\min}}\Big)^{\frac{1}{2n-1}}=: d_{\min, t}, 
\]
we can conclude that for each $t$, we have a permutation $\tau_t$ of $\{1,\cdots, n\}$ so that
\begin{equation}\label{equ:proofhighdsupport1equ2}
\begin{aligned}
&\btwonorm{\mathcal P_{S_t^d}(\vect {\hat{\alpha}}_{\tau_t(j)})- \mathcal P_{S_t^d}(\vect{\alpha}_j)}\leq \frac{C(d,n)}{\sqrt{t^2+1}\Omega}\Big(\frac{\pi}{d_{\min,t}\sqrt{t^2+1}\Omega}\Big)^{2n-2}\frac{\sigma}{m_{\min}}\\
= & \frac{C(d,n)}{\sqrt{t^2+1}\Omega}\Big(\frac{1}{5.88e4^{d-1}((n+2)(n+1)/2)^{\xi(d-1)}}\Big)^{2n-2}\Big(\frac{\sigma}{m_{\min}}\Big)^{\frac{1}{2n-1}}, \quad 1\leq j\leq n,
\end{aligned} 
\end{equation}
where $C(d, n) = \big(4^{d-1}((n+2)(n+1)/2)^{\xi(d-1)}\big)^{(2n-1)}n2^{4n-2}e^{2n}\pi^{-\frac{1}{2}}$.
Note that, for fixed $j$ in (\ref{equ:proofhighdsupport1equ2}), we have $n+1$ $\tau_t(j)$'s (since (\ref{equ:trackonedthm1equ1}) holds for at least $n+1$ $S_t^d$'s), while $\vect{\hat \alpha}_p$'s take at most $n$ values. Therefore, by the pigeonhole principle, for each fixed $\vect \alpha_j$, we can find two different $t$'s, say, $t_1$ and $t_2$, such that $\vect{\hat \alpha}_{\tau_{t_1}(j)} =\vect{\hat \alpha}_{\tau_{t_2}(j)} = \vect{\hat \alpha}_{p_j}$ for some $p_j$. Suppose $t_1>t_2$, by (\ref{equ:proofhighdsupport1equ2}) for $t_1,t_2$ and Lemma \ref{lem:projectlengthestimate2} we have 
\[
\btwonorm{\vect {\hat \alpha}_{p_j} - \vect \alpha_j}\leq \frac{\sqrt{\frac{1}{t_1^2+1}+\frac{1}{t_2^2+1}}}{\sqrt{1-\cos\Big(\sum_{k=t_2+1}^{t_1}\frac{1}{k^2+1}\Big)}} \frac{C(d,n)}{\Omega}\Big(\frac{1}{5.88e4^{d-1}((n+2)(n+1)/2)^{\xi(d-1)}}\Big)^{2n-2}\Big(\frac{\sigma}{m_{\min}}\Big)^{\frac{1}{2n-1}}.
\]
Using the inequality 
$1-\cos x\geq \frac{1}{\pi}x^2$, we further obtain
\[
\btwonorm{\vect {\hat \alpha}_{p_j} - \vect \alpha_j}\leq \frac{\sqrt{\pi}\sqrt{\frac{1}{t_1^2+1}+\frac{1}{t_2^2+1}}}{\sum_{k=t_2+1}^{t_1}\frac{1}{k^2+1}} \frac{C(d,n)}{\Omega}\Big(\frac{1}{5.88e4^{d-1}((n+2)(n+1)/2)^{\xi(d-1)}}\Big)^{2n-2}\Big(\frac{\sigma}{m_{\min}}\Big)^{\frac{1}{2n-1}}.
\]
Since 
\begin{align*}
	\frac{\sqrt{\frac{1}{t_1^2+1}+\frac{1}{t_2^2+1}}}{\sum_{k=t_2+1}^{t_1}\frac{1}{k^2+1}}< \frac{\sqrt{t_1^2+1}\sqrt{1+\frac{t_1^2+1}{t_2^2+1}}}{(t_1-t_2)}\leq \sqrt{6(t_1^2+1)},
\end{align*}
and $t_1\leq \frac{n(n+1)}{2}$,  we have 
\begin{equation}\label{equ:prooftracklocavelocityequ3}
\btwonorm{\vect {\hat \alpha}_{p_j} - \vect \alpha_j}\leq  \frac{\sqrt{6\pi}\sqrt{(\frac{n(n+1)}{2})^2+1}C(d,n)}{\Omega}\Big(\frac{1}{5.88e4^{d-1}((n+2)(n+1)/2)^{\xi(d-1)}}\Big)^{2n-2}\Big(\frac{\sigma}{m_{\min}}\Big)^{\frac{1}{2n-1}}
\end{equation}
for $1\leq j\leq n$. We next claim that
\begin{equation}\label{equ:trackonedthm1equ2}
\btwonorm{\vect {\hat \alpha}_{p_j} - \vect \alpha_j}<\frac{d_{\min}}{2},
\end{equation}
where $d_{\min}$ is defined in (\ref{equ:trackprojectlem1equ0}). Indeed, by $C(d, n) = \big(4^{d-1}((n+2)(n+1)/2)^{\xi(d-1)}\big)^{(2n-1)} n2^{4n-2}e^{2n}\pi^{-\frac{1}{2}}$ and a direct calculation, we can verify that
\begin{align*}
&\sqrt{6\pi}\sqrt{(\frac{n(n+1)}{2})^2+1}C(d,n)\Big(\frac{1}{5.88e4^{d-1}((n+2)(n+1)/2)^{\xi(d-1)}}\Big)^{2n-2} \\
< &\frac{1}{2}11.76e\pi^2\sqrt{(\frac{n(n+1)}{2})^2+1}4^{d-1}\Big((n+2)(n+1)/2\Big)^{\xi(d-1)},
\end{align*}
whence (\ref{equ:trackonedthm1equ2}) follows.

So far, we have proved that for each $\vect \alpha_j$, there exists a point 
$\vect {\hat \alpha}_{p_j}$ satisfying $||\vect {\hat \alpha}_{p_j} - \vect \alpha_j||_2<\frac{d_{\min}}{2}$. Thus for each $\vect \alpha_j$ there exists only one such $\vect {\hat \alpha}_{p_j} \in \{\vect {\hat \alpha}_{1},\cdots, \vect {\hat \alpha}_{n}\}$ . We can reorder the index so that 
\[
\btwonorm{\vect {\hat \alpha}_{j} - \vect \alpha_j}< \frac{d_{\min}}{2},
\]
and
\[
\btwonorm{\vect {\hat \alpha}_{j} - \vect \alpha_j}\leq  \frac{\sqrt{6\pi}(2\pi)^{2n-2}((\frac{n(n+1)}{2})^2+1)^{\frac{2n-1}{2}}C(d, n)}{\Omega}\Big(\frac{\pi}{d_{\min}\Omega}\Big)^{2n-2}\frac{\sigma}{m_{\min}}, \quad 1\leq j\leq n,
\]
which follows from (\ref{equ:prooftracklocavelocityequ3}) and (\ref{equ:trackprojectlem1equ0}). 
\end{proof}

\section{Proofs of Theorems \ref{thm:trackvelocitiesnumberlimit0} and \ref{thm:trackvelocitiessupplimit0}}\label{sec:proofresultsofvelocity}
\subsection{Lemmas for projection}
In this subsection, we present some lemmas for projection in multi-dimensional spaces, which will help us to prove Theorems \ref{thm:trackvelocitiesnumberlimit0} and \ref{thm:trackvelocitiessupplimit0}. 

For a $\vect v\in \mathbb R^k$, we denote by $\vect v^\perp$ the $(k-1)$-dimensional orthogonal complement space of the one-dimensional space spanned by $\vect v$. We consider the unit sphere in $\mathbb R^k$ and the following spherical coordinates: 
\begin{equation}\label{equ:sphericalcoordinates1}
	\begin{aligned}
		x_{1}(\Phi)&=\cos(\phi_{1}),\\
		x_{2}(\Phi)&=\sin(\phi_{1})\cos(\phi_{2}),\\
		%x_{3}(\Phi)&=\sin(\phi_{1})\sin(\phi_{2})\cos(\phi _{3})\\
		&\vdots \\
		%x_{k-2}(\Phi)&= \sin(\phi_{1})\cdots \sin(\phi_{k-3})\cos(\phi_{k-2})\\
		x_{k-1}(\Phi)&=\sin(\phi_{1})\cdots \sin(\phi_{k-2})\cos(\phi_{k-1}),\\
		x_{k}(\Phi)&=\sin(\phi_{1})\cdots \sin(\phi_{k-2})\sin(\phi_{k-1}),
	\end{aligned}
\end{equation}
where $\Phi=(\phi_1, \cdots , \phi_{k-1}) \in [0,\pi]^{k-2}\times [0, 2\pi)$. For $0<\theta<\frac{\pi}{2}$ and $N=\lfloor \frac{\pi}{2\theta} \rfloor$, we let 
\begin{align}\label{equ:vectorlist1}
	\vect v_{\tau_1\cdots \tau_{k-1}}= \big(x_1(\Phi_{\tau_1\cdots\tau_{k-1}}), \cdots , x_k(\Phi_{\tau_1\cdots\tau_{k-1}})\big)^\top, \quad 1\leq \tau_j \leq N,
\end{align}
where $\Phi_{\tau_1\cdots\tau_{k-1}}= (\tau_1\theta, \cdots, \tau_{k-1}\theta)$. It is obvious that $\Phi_{\tau_1\cdots\tau_{k-1}} \in [0, \frac{\pi}{2}]^{k-1}$ and $\vect v_{\tau_1\cdots \tau_{k-1}} \neq \vect v_{p_1\cdots p_{k-1}}$ if $(\tau_1, \cdots, \tau_{k-1})\neq (p_1, \cdots, p_{k-1})$. There are $N^{k-1}$ different unit vectors of the form (\ref{equ:vectorlist1}).

\begin{lem}\label{lem:projectvectorangle1}
For two different vectors $\vect v_{\tau_1\cdots \tau_{k-1}} \neq \vect v_{p_1\cdots p_{k-1}}$ in (\ref{equ:vectorlist1}), we have
\begin{equation}\label{equ:projectlemangle1}
0 \leq \vect v_{\tau_1\cdots \tau_{k-1}} \cdot \vect v_{p_1\cdots p_{k-1}} \leq \cos\theta.
\end{equation}
\end{lem}
\begin{proof}
See Lemma 6.3 in \cite{liu2021mathematicalhighd}.
\end{proof}

\begin{lem}\label{lem:projectvectorangle2}
For a vector $\vect u\in \mathbb R^k$,  suppose $||\mathcal P_{\vect v_{\tau_1\cdots \tau_{k-1}}^\perp}(\vect u)||_2 < \sin (\frac{\theta}{2}) ||\vect u||_2$ with $\vect v_{\tau_1\cdots \tau_{k-1}}$ defined in (\ref{equ:vectorlist1}), we have $||\mathcal P_{\vect v_{p_1\cdots p_{k-1}}^\perp}(\vect u)||_2 \geq \sin (\frac{\theta}{2})||\vect u||_2$ for $\vect v_{p_1\cdots p_{k-1}} \neq \vect v_{\tau_1\cdots \tau_{k-1}}$.
\end{lem}
\begin{proof}
See Lemma 6.4 in \cite{liu2021mathematicalhighd}.
\end{proof}

\begin{lem}\label{lem:highdsupportproject2}
Let $k\geq 2$. For a vector $\vect{u}\in \mathbb R^k$, and two unit vectors $\vect{v}_1, \vect{v}_2\in \mathbb R^k$ satisfying $0 \leq \vect v_1 \cdot \vect v_2 \leq \cos(\theta)$, we have 
\begin{equation}
||\mathcal P_{\vect v_1^{\perp}}(\vect u)||_2^2+||\mathcal P_{\vect v_2^{\perp}}(\vect u)||_2^2  \geq (1-\cos(\theta))||\vect u||_2^2.
\end{equation}
\end{lem}
\begin{proof}
See Lemma 3.3 in \cite{liu2021mathematicalhighd}.
\end{proof}

\begin{lem}\label{lem:spaceangleestimate5}
Suppose we have $\frac{n(n-1)}{2}$ different $\vect u_{pq}\in \mathbb R^k$, then there exists a one-dimensional space $S$ so that for all $\vect u_{pq}$'s we have 
\[
\btwonorm{\mathcal P_S(\vect u_{pq})} \geq \frac{||\vect{u}_{pq}||_2}{(\pi/2)^{k-1}}\Big(\frac{\pi}{n(n-1)}\Big)^{\xi(k-1)},
\]
where $\xi(k-1)$ is defined as in (\ref{equ:defineofxi}).
\end{lem}
\begin{proof}
Let $S_{k-1}$ be  the unit sphere in $\mathbb R^k$. For a subspace $S$ of $\mathbb R^k$ and $\vect v\in S$, we denote by $\vect v^{\perp}(S)$ the orthogonal complement space of $\vect v$ in $S$. For each $\vect u \in \mathbb R^k$, we let
\begin{align}\label{equ:defineofN}
N(\vect{u}, \Delta) &= \Big\{ \vect v \Big|\vect v \in S_{k-1}, ||\mathcal P_{\vect v^\perp(\mathbb R^k)}(\vect{u})||_2 <||\vect{u}||_2 \sin\Delta \Big\}. 
\end{align}
Define $\text{area}(A)$ as the area of set $A$ and let $\Delta_k = \Big(\frac{\pi}{n(n-1)}\Big)^{\frac{1}{k-1}}$, we have
\[
\text{area} \left(\cup_{pq} N(\vect{u}_{pq}, \Delta_k)\right) \leq \frac{n(n-1)}{2} \frac{2\text{area}(S_{k-2})}{k-1}\Delta_k^{k-1}= \frac{\pi\text{area}(S_{k-2})}{k-1}\leq \text{area}(S_{k-1}),
\]
where the first inequality is from Lemma 6.1 in \cite{liu2021mathematicalhighd} and the last one is from Lemma 6.2 in \cite{liu2021mathematicalhighd}. On the other hand, $\cup_{pq} N(\vect{u}_{pq}, \Delta_k)$ is an open set in $S_{k-1}$. Thus $S_{k-1}\setminus \cup_{pq} N(\vect{u}_{pq}, \Delta_k)$ is not empty. By the definition of $N(\vect{u}_{pq}, \Delta_k)$, there exists a unit vector $\vect{v}_k \in \mathbb R^{k}$ such that 
\[
\btwonorm{\mathcal P_{\vect v_k^\perp(\mathbb R^k)}(\vect u_{pq})}\geq ||\vect{u}_{pq}||_2 \sin\Delta_k \geq \frac{||\vect{u}_{pq}||_2}{\pi/2} \Big(\frac{\pi}{n(n-1)}\Big)^{\frac{1}{k-1}} 
\]
for all $\vect u_{pq}$'s. Note that $\mathcal P_{\vect v_k^\perp(\mathbb R^k)}(\vect u_{pq})$ are at most $n(n-1)$ different vectors in a $(k-1)$-dimensional space $\vect v_k^\perp(\mathbb R^k)$. Applying similar arguments as above to $\vect v_k^\perp(\mathbb R^k)$ and $\mathcal P_{\vect v_k^\perp(\mathbb R^k)}(\vect u_{pq})$'s, we can show that there exists a unit vector $\vect{v}_{k-1} \in \vect v_k^\perp(\mathbb R^k)$ such that for all $\vect u_{pq}$'s, 
\begin{align*}
\btwonorm{\mathcal P_{\vect v_{k-1}^\perp(\vect v_k^\perp(\mathbb R^k))}\Big(\mathcal P_{\vect v_k^\perp(\mathbb R^k)}(\vect u_{pq})\Big)}  \geq \btwonorm{\mathcal P_{\vect v_k^\perp(\mathbb R^k)}(\vect u_{pq})}\Big(\frac{2}{\pi}\Big) \Big(\frac{\pi}{n(n-1)}\Big)^{\frac{1}{k-2}} 
\geq  \frac{||\vect{u}_{pq}||_2}{(\pi/2)^2} \Big(\frac{\pi}{n(n-1)}\Big)^{\frac{1}{k-1}+\frac{1}{k-2}}. 
\end{align*}
On the other hand, since $\vect v_{k-1}^\perp(\vect v_k^\perp(\mathbb R^k))$ is a subspace of $\vect v_k^\perp(\mathbb R^k)$, we further obtain
\[
\btwonorm{\mathcal P_{\vect v_{k-1}^\perp(\vect v_k^\perp(\mathbb R^k))}(\vect u_{pq})} = \btwonorm{\mathcal P_{\vect v_{k-1}^\perp(\vect v_k^\perp(\mathbb R^k))}\Big(\mathcal P_{\vect v_k^\perp(\mathbb R^k)}(\vect u_{pq})\Big)}  \geq  \frac{||\vect{u}_{pq}||_2}{(\pi/2)^2} \Big(\frac{\pi}{n(n-1)}\Big)^{\frac{1}{k-1}+\frac{1}{k-2}}. 
\]
Continuing the above arguments from dimension $(k-2)$ to dimension $2$, we can find a unit vector $\vect v_1$ so that 
\[
\btwonorm{\mathcal P_{\vect v_1^{\perp}(\cdots\vect v_{k-1}^\perp(\vect v_k^\perp(\mathbb R^k)))}(\vect u_{pq})}\geq \frac{||\vect{u}_{pq}||_2}{(\pi/2)^{k-1}}\Big(\frac{\pi}{n(n-1)}\Big)^{\xi(k-1)},
\]
for all $\vect u_{pq}$'s. If we define $S$ in the lemma as the one-dimensional space $\vect v_1^{\perp}(\cdots\vect v_{k-1}^\perp(\vect v_k^\perp(\mathbb R^k)))$, then the proof is completed. 
\end{proof}

\begin{lem}\label{lem:spaceangleestimate7}
Suppose we have $\frac{n(n-1)}{2}$ different $\vect u_{pq}\in \mathbb R^k$, then there exist $n+1$ $(k-1)$-dimensional spaces $S_j^{k-1}= \vect v_{j}^\perp, j=1,\cdots,n+1$ with $0\leq \vect v_{p}^\top\vect v_j\leq \cos\Big(\frac{\pi}{4}\Big(\frac{2}{(n+1)(n+2)}\Big)^{\frac{1}{k-1}}\Big), p\neq j$, so that for all $\vect u_{pq}$'s we have 
\[
\btwonorm{\mathcal P_{S_j^{k-1}}(\vect u_{pq})} \geq \frac{||\vect{u}_{pq}||_2}{4}\Big(\frac{2}{(n+1)(n+2)}\Big)^{\frac{1}{k-1}}.
\]
\end{lem}
\begin{proof}
For $k\geq 2$, let $\Delta_k = \frac{\pi}{8}\Big(\frac{2}{(n+1)(n+2)}\Big)^{\frac{1}{k-1}}$ and $\theta_k = 2\Delta_k$. First, we consider the vectors (\ref{equ:vectorlist1}) in $\mathbb R^k$ with $\theta = \theta_k$ and $N=N_k = \lfloor \frac{\pi}{2\theta_k} \rfloor$. We denote $N(\vect{u}, \Delta) = \Big\{ \vect v \Big|\vect v \in \mathbb R^k, ||\vect v||_2=1, ||\mathcal P_{\vect v^\perp}(\vect{u})||_2 <||\vect{u}||_2 \sin\Delta \Big\}$. Thus, by Lemma \ref{lem:projectvectorangle2}, each $N(\vect{u}_{pq}, \Delta_k)$ contains at most one unit vector in (\ref{equ:vectorlist1}). Next recall that there are $N_k^{k-1}$ different vectors in (\ref{equ:vectorlist1}), where $N_k = \lfloor \frac{\pi}{2\theta_k} \rfloor \geq \frac{\pi}{2\theta_k} -1 $. Since 
\[
\theta_k =2 \Delta_k =\frac{\pi}{4}\Big(\frac{2}{(n+1)(n+2)}\Big)^\frac{1}{k-1}, 
\]
we have
\begin{align*}
N_k^{k-1}\geq \Big(\frac{\pi}{2\theta_k}-1\Big)^{k-1} = \Big(2 (\frac{(n+1)(n+2)}{2})^{\frac{1}{k-1}}-1\Big)^{k-1}
\geq \Big((\frac{(n+1)(n+2)}{2})^{\frac{1}{k-1}}\Big)^{k-1}= \frac{(n+1)(n+2)}{2}.
\end{align*}
Note that $\frac{(n+1)(n+2)}{2} - \frac{n(n-1)}{2} \geq  n+1$, we can at least find $(n+1)$ vectors $\vect v_j$'s in (\ref{equ:vectorlist1}) so that $||\mathcal P_{\vect v_j^\perp}(\vect{u}_{pq})||_2 \geq ||\vect{u}_{pq}||_2 \sin\Delta_k$ for all $\vect u_{pq}$'s. Thus we have 
\[
\btwonorm{\mathcal P_{\vect v_j^{\perp}}(\vect u_{pq})}\geq \frac{2||\vect{u}_{pq}||_2}{\pi}\frac{\pi}{8}\Big(\frac{2}{(n+1)(n+2)}\Big)^{\frac{1}{k-1}} = \frac{||\vect{u}_{pq}||_2}{4}\Big(\frac{2}{(n+1)(n+2)}\Big)^{\frac{1}{k-1}}
\]
for all $\vect u_{pq}$'s. On the other hand, by Lemma \ref{lem:projectvectorangle1} we have $0\leq \vect v_{p}^\top\vect v_j\leq \cos\Big(\frac{\pi}{4}\Big(\frac{2}{(n+1)(n+2)}\Big)^{\frac{1}{k-1}}\Big)$ for $p\neq j$.
\end{proof}

\begin{lem}\label{lem:spaceangleestimate6}
Define 
\[
d_{\min, k} = \frac{11.76\pi e 4^{k-1} ((n+2)(n+1)/2)^{\xi(k-1)}}{T\Omega }\Big(\frac{\sigma}{m_{\min}}\Big)^{\frac{1}{2n-1}}.
\]
For a vector $\vect u\in \mathbb R^k$, suppose we can find two $(k-1)$-dimensional spaces $S_{j}^{k-1}= \vect v_{j}^{\perp}, j=1,2$ with $0\leq \vect v_{1}^\top\vect v_2\leq \cos\Big(\frac{\pi}{4}\Big(\frac{2}{(n+1)(n+2)}\Big)^{\frac{1}{k-1}}\Big)$ so that
\[
\btwonorm{\mathcal P_{S_{j}^{k-1}}(\vect u)} < \frac{1}{2} d_{\min, k-1}, \quad j=1,2,
\]
and 
\[
\btwonorm{\mathcal P_{S_{j}^{k-1}}(\vect u)}\leq \frac{C(k-1, n)}{T\Omega}
\Big(\frac{\pi}{d_{\min,k-1}T\Omega}\Big)^{2n-2}\frac{\sigma}{m_{\min}}, \quad j=1,2,
\]
where 
\[C(k-1, n)=\big(4^{k-2}((n+2)(n+1)/2)^{\xi(k-2)}\big)^{(2n-1)}n2^{6n-3}e^{2n}\pi^{-\frac{1}{2}}.\]
Then we have 
\[
\btwonorm{\vect u} < \frac{1}{2} d_{\min, k}, 
\]
and
\[
\btwonorm{\vect u}\leq \frac{C(k, n)}{T\Omega}\Big(\frac{\pi}{d_{\min,k}T\Omega}\Big)^{2n-2}\frac{\sigma}{m_{\min}},
\]
where 
\[
C(k, n)=\big(4^{k-1}((n+2)(n+1)/2)^{\xi(k-1)}\big)^{(2n-1)}n2^{6n-3}e^{2n}\pi^{-\frac{1}{2}}.
\]
\end{lem}
\begin{proof}
Let $\theta = \frac{\pi}{4}\Big(\frac{2}{(n+1)(n+2)}\Big)^{\frac{1}{k-1}}$. By condition for $\vect v_1, \vect v_2$ in the lemma and Lemma \ref{lem:highdsupportproject2}, we have 
\begin{equation}
||\mathcal P_{\vect v_1^{\perp}}(\vect u)||_2^2+||\mathcal P_{\vect v_2^{\perp}}(\vect u)||_2^2  \geq (1-\cos\theta)||\vect u||_2^2.
\end{equation}
Thus 
\[
||\vect u||_2 \leq \frac{\sqrt{2}}{\sqrt{1-\cos \theta}}\frac{C(k-1, n)}{T\Omega}\Big(\frac{\pi}{d_{\min,k-1}T\Omega}\Big)^{2n-2}\frac{\sigma}{m_{\min}}.
\]
Using the inequality 
$1-\cos \theta \geq \frac{2}{\pi^2}\theta^2= \frac{1}{8}\big( \frac{2}{(n+2)(n+1)}\big)^{\frac{2}{k-1}}$, we further obtain
\begin{equation}\label{equ:proofvelocityequ3}
\btwonorm{ \vect u}\leq  \frac{4 ((n+2)(n+1)/2)^{\frac{1}{k-1}}C(k-1,n)}{T\Omega}\Big(\frac{\pi}{d_{\min,k-1}T\Omega}\Big)^{2n-2}\frac{\sigma}{m_{\min}}, \quad 1\leq j\leq n.
\end{equation}
We next claim that
\[
\btwonorm{\vect u}<\frac{d_{\min,k}}{2}.
\]
Indeed, by a direct calculation, we can verify that
\begin{align*}
&4 ((n+2)(n+1)/2)^{\frac{1}{k-1}}C(k-1,n)\Big(\frac{1}{11.76e4^{k-2}((n+2)(n+1)/2)^{\xi(k-2)}}\Big)^{2n-2}\\
<& \frac{1}{2} 11.76\pi e  4^{(k-1)}((n+2)(n+1)/2)^{\xi(k-1)}.
\end{align*}
On the other hand, we have
\[
(\frac{\pi}{d_{\min,k-1}T\Omega})^{2n-2}\frac{\sigma}{m_{\min}}\leq \Big(\frac{1}{11.76e4^{k-2}((n+2)(n+1)/2)^{\xi(k-2)}}\Big)^{2n-2}\Big(\frac{\sigma}{m_{\min}}\Big)^{\frac{1}{2n-1}}.
\]
Therefore, we have
\begin{align*}
&\frac{4((n+2)(n+1)/2)^{\frac{1}{k-1}}C(k-1,n)}{T\Omega}\Big(\frac{\pi}{d_{\min,k-1}T\Omega}\Big)^{2n-2}\frac{\sigma}{m_{\min}}\\
<&\frac{1}{2} \frac{11.76\pi e 4^{k-1}((n+2)(n+1)/2)^{\xi(k-1)}}{T\Omega }\Big(\frac{\sigma}{m_{\min}}\Big)^{\frac{1}{2n-1}}. 
\end{align*}
This is $\btwonorm{\vect u} < \frac{1}{2} d_{\min, k}$. Moreover, we have
\[
\btwonorm{\vect u}\leq  \frac{(4((n+2)(n+1)/2)^{\frac{1}{k-1}})^{2n-1}C(k-1,n)}{T\Omega}\Big(\frac{\pi}{d_{\min,k}T\Omega}\Big)^{2n-2}\frac{\sigma}{m_{\min}},
\]
which follows from (\ref{equ:proofvelocityequ3}) and the equation that $d_{\min,k-1}=\frac{d_{\min,k}}{4((n+2)(n+1)/2)^{\frac{1}{k-1}}}$. This completes the proof.
\end{proof}

\subsection{Theorems on source number  and location recovery for the one-dimensional static super-resolution problem}
The idea to prove Theorems \ref{thm:trackvelocitiesnumberlimit0} and \ref{thm:trackvelocitiessupplimit0} is to reduce the stability of a high-dimensional tracking problem to many one-dimensional super-resolution problems. Thus in this subsection we present results (Theorems \ref{thm:onednumberlimit1} and \ref{thm:onedsupplimit1}) for the  source number and location recovery in a one-dimensional static super-resolution problem that is tuned to subsequent proofs of stability of the velocity reconstruction. The results are slightly different from the ones in \cite{liu2021theorylse} and therefore, we present their detailed proofs here for the sake of completeness. 

We first introduce some lemmas and theorems from \cite{liu2021mathematicaloned}. We denote for integer $k\geq 1$, 
\begin{equation}\label{zetaxiformula1}
\begin{aligned}
\zeta(k)= \left\{
\begin{array}{cc}
(\frac{k-1}{2}!)^2,& \text{$k$ is odd,}\\
(\frac{k}{2})!(\frac{k-2}{2})!,& \text{$k$ is even,}
\end{array} 
\right. \ \beta(k)=\left\{
\begin{array}{cc}
\frac{1}{2},  & k=1,\\
\frac{(\frac{k-1}{2})!(\frac{k-3}{2})!}{4},& \text{$k$ is odd,\,\,$ k\geq 3$,}\\
\frac{(\frac{k-2}{2}!)^2}{4},& \text{$k$ is even}.
\end{array} 
\right.	
\end{aligned}
\end{equation}
We also define for positive integers $p, q$, and $z_1, \cdots, z_p, \hat z_1, \cdots, \hat z_q \in \mathbb C$, the following vector in $\mathbb{R}^p$
\begin{equation}\label{notation:eta}
\eta_{p,q}(z_1,\cdots,z_{p}, \hat z_1,\cdots,\hat z_q)=\left(\begin{array}{c}
|(z_1-\hat z_1)|\cdots|(z_1-\hat z_q)|\\
|(z_2-\hat z_1)|\cdots|(z_2-\hat z_q)|\\
\vdots\\
|(z_{p}-\hat z_1)|\cdots|(z_{p}-\hat  z_q)|
\end{array}\right).
\end{equation}
For a complex matrix $A$, we denote by $A^*$ its conjugate transpose. For integer $s$ and $z\in \mathbb C$, we define the complex Vandermonde-vector
\begin{equation}\label{phiformula}
\phi_s(z)=(1,z,\cdots,z^s)^\top.
\end{equation}

\begin{lem}\label{lem:stablemultiproduct1}
	For $-\frac{\pi}{2}\leq \theta_1<\theta_2<\cdots<\theta_k \leq \frac{\pi}{2}$ and $ \hat \theta_1, \hat \theta_2, \cdots, \hat \theta_k \in [-\frac{\pi}{2}, \frac{\pi}{2}]$, suppose
	\begin{align*}
	||\eta_{k,k}(e^{i \theta_1},\cdots,e^{i \theta_k}, e^{i \hat \theta_1},\cdots, e^{i \hat \theta_k})||_{\infty}< (\frac{2}{\pi})^{k}\epsilon, \text{ and } \theta_{\min} =\min_{q\neq j}|\theta_q-\theta_j|\geq  \Big(\frac{4\epsilon}{\lambda(k)}\Big)^{\frac{1}{k}},
	\end{align*}
	where $\eta_{k,k}$ is defined by (\ref{notation:eta}) and \begin{equation}\label{equ:lambda1}
	\lambda(k)=\left\{
	\begin{array}{ll}
	1,  & k=2,\\
	\beta(k-2),& k\geq 3.
	\end{array} 
	\right.	
	\end{equation}
	Then after reordering $\hat \theta_j$'s, we have
	\begin{equation}\label{equ:satblemultiproductcor1}
	|\hat \theta_j -\theta_j|< \frac{\theta_{\min}}{2}  \text{ and }
	|\hat \theta_j -\theta_j|\leq \frac{2^{k-1}\epsilon}{(k-2)!(\theta_{\min})^{k-1}}, \quad j=1,\cdots, k.
	\end{equation}
\end{lem}
\begin{proof}
See Corollary 9 in \cite{liu2021theorylse}. 
\end{proof}

\begin{thm}\label{spaceapprolowerbound1}
	Let $k\geq 1$. Assume that $\theta_j\in\big[\frac{-\pi}{2}, \frac{\pi}{2}\big], 1\leq j\leq k+1$ are $k+1$ distinct points, and $|a_j|\geq m_{\min},1\leq j\leq k+1$. Let $\theta_{\min}=\min_{p\neq j}|\theta_p-\theta_j|$. For $q\leq k$, let $\hat a(q)=(\hat a_1,\cdots, \hat a_q)^\top$, $a=(a_1,\cdots, a_{k+1})^\top$ and
	\[\hat A(q)= \big(\phi_{2k}(e^{i\hat \theta_1}),\cdots, \phi_{2k}(e^{i\hat \theta_q})\big), \quad A= \big(\phi_{2k}(e^{i\theta_1}),\cdots, \phi_{2k}(e^{i\theta_{k+1}})\big),\]
	where $\phi_{2k}(z)$ is defined as in (\ref{phiformula}). Then
	\begin{equation*}
	\min_{\hat a_p\in \mathbb C,\hat \theta_p\in \mathbb R,p=1,\cdots,q}\btwonorm{\hat A(q)\hat a(q)-Aa}\geq  \frac{\zeta(k+1)\beta(k)m_{\min}\theta_{\min}^{2k}}{\pi^{2k}}.
	\end{equation*}
\end{thm}
\begin{proof}
See Theorem 4 in \cite{liu2021theorylse}.
\end{proof}
\medskip
\begin{thm}\label{spaceapproxlowerbound2}
	Let $k\geq 2$. Assume that $\theta_1,\cdots,\theta_{k}\in \big[\frac{-\pi}{2}, \frac{\pi}{2}\big]$ are $k$ different points and $|a_j|\geq m_{\min},1\leq j\leq k$. Define $\theta_{\min}=\min_{p\neq j}|\theta_p-\theta_j|$. Let $k$ distinct  points $\hat \theta_1,\cdots,\hat \theta_k\in \big[\frac{-\pi}{2}, \frac{\pi}{2}\big]$ satisfy
	\[ ||\hat A\hat a-A a||_2< \sigma, \]
	where
	$\hat a=(\hat a_1,\cdots, \hat a_k)^\top$, $a=(a_1,\cdots, a_{k})^\top$ and
	\[\hat A= \big(\phi_{2k-1}(e^{i \hat \theta_1}),\cdots, \phi_{2k-1}(e^{i \hat \theta_k})\big),\  A= \big(\phi_{2k-1}(e^{i \theta_1}),\cdots, \phi_{2k-1}(e^{i \theta_{k}})\big).\]
	Then
	\[||\eta_{k,k}(e^{i \theta_1},\cdots,e^{i \theta_k},e^{i \hat \theta_1},\cdots,e^{i \hat \theta_k})||_{\infty}<\frac{2^{k}\pi^{k-1}}{\zeta(k)\theta_{\min}^{k-1}}\frac{\sigma}{m_{\min}}.\]
\end{thm}
\begin{proof}
See Theorem 5 in \cite{liu2021theorylse}.
\end{proof}

\begin{lem}\label{upperboundsupportcalculate1}
	Let $\zeta(n), 
	\beta(n)$ and $\lambda(n)$ be defined in (\ref{zetaxiformula1}) and (\ref{equ:lambda1}), respectively. For $n\geq 2$, we have
	\begin{equation}\label{equ:numberestimate1}
	\Big(\frac{2\sqrt{2n-1}}{\zeta(n)\beta(n-1)}\Big)^{\frac{1}{2n-2}}\leq \frac{4.4 e}{2n-1}, 
	\end{equation}
	\begin{equation}\label{equ:suppestimate1}
	\Big(\frac{8\sqrt{2n}}{\zeta(n)\lambda(n)}\Big)^{\frac{1}{2n-1}}\leq \frac{5.88 e}{2n},
	\end{equation}
	and 
	\begin{equation}\label{equ:suppestimate2}
	\frac{(2n)^{2n-\frac{3}{2}}}{\zeta(n)(n-2)!}\leq 2^{3n-3}e^{2n}\pi^{-\frac{3}{2}}. 
	\end{equation}
\end{lem}
\begin{proof}
See the appendix in \cite{liu2021theorylse}.
\end{proof}

We define a discrete measure as $\mu = \sum_{j=1}^n a_j \delta_{y_j}$ and the vector of its Fourier transform at $0, \Omega, 2\Omega, \cdots, T\Omega$ as 
\[
[\mu] = (\mathcal F[\mu](0), \mathcal F[\mu](\Omega), \cdots, \mathcal F[\mu](T\Omega))^\top,
\]
where $\mathcal F[\mu](x) = \sum_{j=1}^n a_j e^{i y_jx}$. We have the following theorems for the source number and location recovery in a one-dimensional static super-resolution problem.

\begin{thm}\label{thm:onednumberlimit1}
Let $n\geq 2$ and $T\geq 2n-2$. Suppose that the measurement is 
\[
\vect Y(t) = \sum_{j=1}^n a_j e^{iy_j \Omega t} + \vect W(t), \quad  t=0, \cdots, T,
\]	
with $|\vect W(t)|< \sigma$, $y_j\in [-\frac{(n-1)\pi}{T\Omega}, \frac{(n-1)\pi}{T\Omega}]$ and 
\begin{equation}\label{equ:onednumberlimitequ0}
d_{\min}:=\min_{p\neq q} \babs{y_p-y_q} \geq \frac{8.8e\pi}{T\Omega}\Big(\frac{\sigma}{m_{\min}}\Big)^{\frac{1}{2n-2}}.
\end{equation}
Then there is no $\hat \mu = \sum_{j=1}^k \hat a_j \delta_{\hat y_j}$ with $k<n$ so that 
\[
||[\hat \mu] - \vect Y||_{\infty} < \sigma. 
\]
\end{thm}
\begin{proof}
\textbf{Step 1.} We write 
\begin{equation}\label{Mdecomposition1}
T+1=(2n-1)r+q,
\end{equation}
where $r, q$ are integers with $r\geq 1$ and $0\leq q< 2n-1$. 
We denote by $\theta_j = 	y_jr\Omega, j=1,\cdots,n$. For $y_j\in [-\frac{(n-1)\pi}{T\Omega},\frac{(n-1)\pi}{T\Omega}]$, in view of (\ref{Mdecomposition1}), it is clear that
\begin{equation}\label{equ:distributionequ1}
\theta_j = y_jr\Omega\in \Big[\frac{-\pi}{2}, \frac{\pi}{2}\Big], \quad j=1,\cdots,n.
\end{equation}
For $\hat \mu =\sum_{j=1}^{k}\hat a_j \delta_{\hat y_j}$ with $k<n$, note that
\begin{align*}
[\hat \mu]-[\mu] = \left(\mathcal F [\hat \mu](0), \mathcal F [\hat \mu](\Omega), \cdots,\mathcal F [\hat \mu](T\Omega)\right)^\top-\left(\mathcal F [\mu](0), \mathcal F  [\mu](\Omega), \cdots,\mathcal F [\mu](T\Omega)\right)^\top.
\end{align*}
Using 
only the partial measurement at $qr\Omega, 0\leq q \leq 2n-2$, we have 
\begin{equation*}\label{equ:upperboundnumberlimithm0equ0}
\left(\mathcal F \hat \mu(0), \mathcal F \hat \mu(r\Omega), \cdots,\mathcal F \hat \mu((2n-2)r\Omega)\right)^\top-\left(\mathcal F \mu(0), \mathcal F \mu(r\Omega), \cdots,\mathcal F \mu((2n-2)r\Omega)\right)^\top =\hat B \hat a- B a , 
\end{equation*}
where $\hat a= (\hat a_1,\cdots, \hat a_k)^\top$, $a=(a_1, \cdots, a_n)^\top$ and  
\begin{equation}\label{matrixdecomposition1}
\begin{aligned}
&\hat B =\big(\phi_{2n-2}(e^{i \hat \theta_1}), \cdots,\phi_{2n-2}(e^{i\hat \theta_k} ) \big),\\
&B=\big(\phi_{2n-2}(e^{i \theta_1}), \cdots,\phi_{2n-2}(e^{i \theta_n})\big),
\end{aligned}
\end{equation}
with $\theta_j = y_j r\Omega, \hat \theta_j = \hat y_j r\Omega$. It is clear that 
\begin{equation} \label{equ:upperboundnumberlimitequ0}
\begin{aligned}
&\min_{\hat a\in \mathbb C^k, \hat y_j\in \mathbb R,j=1,\cdots,k}||[\hat \mu ]-[\mu ]||_{\infty} \geq	\min_{\hat a \in \mathbb C^k, \hat y_j\in \mathbb R,j=1,\cdots,k}||\hat B \hat a- B a||_\infty \\
\geq &\frac{1}{\sqrt{2n-1}}\min_{\alpha\in \mathbb C^{k}, \hat y_j\in \mathbb R,j=1,\cdots, k}||\hat B \alpha-B a||_2.
\end{aligned}
\end{equation}
In view of (\ref{equ:distributionequ1}), we can apply Theorem \ref{spaceapprolowerbound1} to get 
\begin{equation*}
\min_{\alpha\in \mathbb C^{k}, \hat y_j\in \mathbb R,j=1,\cdots,k}||\hat B \alpha-B a||_2\geq   \frac{m_{\min}\zeta(n)\beta(n-1)(\theta_{\min})^{2n-2}}{\pi^{2n-2}}, 
\end{equation*}
where $\theta_{\min}=\min_{j\neq p}|\theta_j-\theta_p|$. Combining the above estimate with (\ref{equ:upperboundnumberlimitequ0}), we get
\begin{align}\label{upperboundnumberlimitequ1}
&\min_{\hat a\in \mathbb C^k, \hat y_j\in \mathbb R,j=1,\cdots,k}||[\hat \mu ]-[\mu ]||_{\infty}\geq 
\frac{m_{\min}\zeta(n)\beta(n-1)(\theta_{\min})^{2n-2}}{\sqrt{2n-1}\pi^{2n-2}}. 	
\end{align}
\textbf{Step 2.} Recall that $d_{\min}= \min_{j\neq p}|y_j-y_p|$. Using the relation $\theta_j = 	y_jr\Omega$ and (\ref{Mdecomposition1}), we can show that 
\begin{equation*}\label{equ:angleinequ2}
\theta_{\min}=r\Omega d_{\min}\geq \frac{rT\Omega}{(2n-1)(r+1)}d_{\min}\geq \frac{T\Omega}{2(2n-1)}d_{\min}.\quad \Big(\text{$r\geq 1$}\Big) 
\end{equation*}
Then the separation condition (\ref{equ:onednumberlimitequ0}) implies
\begin{equation*}
\theta_{\min}\geq \frac{4.4 \pi e}{2n-1}\Big(\frac{\sigma}{m_{\min}}\Big)^{\frac{1}{2n-2}}\geq \Big(\frac{2\sqrt{2n-1}}{\zeta(n)\beta(n-1)}\frac{\sigma}{m_{\min}}\Big)^{\frac{1}{2n-2}},
\end{equation*}
where we have used (\ref{equ:numberestimate1}) for the last inequality above. Therefore (\ref{upperboundnumberlimitequ1}) implies that
\begin{align*}
&\min_{\hat a\in \mathbb C^k, \hat y_j\in \mathbb R,j=1,\cdots,k}||[\hat \mu ]-[\mu ]||_{\infty} \geq 2\sigma.
\end{align*}
It follows that
\begin{align*}
&||[\hat \mu]-\mathbf Y||_{\infty}=|| [\hat \mu]- [\mu]-\mathbf W||_{\infty}\\
\geq&||[\hat \mu]- [\mu] ||_{\infty}- ||\mathbf W||_{\infty}\geq|| [\hat \mu ]- [\mu] ||_{\infty}-\sigma \geq \sigma,
\end{align*}
which shows that $||[\hat \mu] - \vect Y||_{\infty} < \sigma$ is impossible.  This completes the proof.
\end{proof}

\begin{thm}\label{thm:onedsupplimit1}
Let $n\geq 2$ and $T\geq 2n-1$. Suppose that the measurement is 
\[
\vect Y(t) = \sum_{j=1}^n a_j e^{iy_j \Omega t} + \vect W(t), \quad  t=0, \cdots, T,
\]	
with $|\vect W(t)|< \sigma$, $y_j\in [-\frac{(n-1)\pi}{T\Omega}, \frac{(n-1)\pi}{T\Omega}]$ and 
\begin{equation}\label{supportlimithm0equ0}
d_{\min}:=\min_{p\neq q} \babs{y_p-y_q} \geq \frac{11.76e\pi}{T\Omega}\Big(\frac{\sigma}{m_{\min}}\Big)^{\frac{1}{2n-1}}.
\end{equation}
For a measure $\hat \mu = \sum_{j=1}^n \hat a_j \delta_{\hat y_j}$ with $\hat y_j\in [-\frac{(n-1)\pi}{T\Omega}, \frac{(n-1)\pi}{T\Omega}]$ and $||[\hat \mu]-\vect Y||_{\infty}<\sigma$, we can reorder $\hat y_j$'s so that 
\[
\babs{\hat y_j -y_j}<\frac{d_{\min}}{2},
\]
and 
\begin{align*}
\Big|\hat y_j-y_j\Big|\leq \frac{C(n)}{\Omega} (\frac{\pi}{T\Omega d_{\min}})^{2n-2} \frac{\sigma}{m_{\min}}, 
\end{align*}
where $C(n)=n2^{6n-3}e^{2n}\pi^{-\frac{1}{2}}$.
\end{thm}
\begin{proof}
\textbf{Step 1.} We first write
\begin{equation}\label{Mdecomposition2}
T+1=2nr+q, 
\end{equation}
where $r, q$ are integers with $r\geq 1$ and $0\leq q< 2n$. Since $y_j, \hat y_j\in [-\frac{(n-1)\pi}{T\Omega}, \frac{(n-1)\pi}{T\Omega}]$, we have 
\begin{equation}\label{equ:distributionequ2}
\theta_j :=y_jr\Omega\in \Big[\frac{-\pi}{2}, \frac{\pi}{2}\Big],\quad  \hat \theta_p :=\hat y_pr\Omega\in \Big[\frac{-\pi}{2}, \frac{\pi}{2}\Big], \ 1\leq j,p\leq n.
\end{equation}
Also, by (\ref{Mdecomposition2}),
\begin{equation}\label{equ:angleinequ4}
\babs{\theta_j-\theta_p}=r\Omega
\babs{y_j-y_p} \geq \frac{rT\Omega}{2n(r+1)}\geq  \frac{T\Omega}{4n}\babs{y_j-y_p}, \ \Big(\text{since $r\geq 1$}\Big)
\end{equation}
and 
\begin{equation}\label{equ:angleinequ3}
\theta_{\min}:=\min_{j\neq p}\babs{\theta_j-\theta_p}\geq \frac{T\Omega}{4n}d_{\min}.
\end{equation}
Similar to Step 1 in the proof of Theorem \ref{thm:onednumberlimit1}, we consider 
\begin{align*}
&\left(\mathcal F [\hat \mu](0), \mathcal F [\hat \mu](r\Omega), \cdots,\mathcal F [\hat \mu]((2n-1)r\Omega)\right)^\top\\
-&\left(\mathcal F [\mu](0), \mathcal F [\mu](r\Omega), \cdots,\mathcal F [\mu]((2n-1)r\Omega)\right)^\top =\hat B \hat a- B a, 
\end{align*}
where $\hat a= (\hat a_1,\cdots, \hat a_n)^\top$, $a=(a_1, \cdots, a_n)^\top$, and
\begin{equation}\label{matrixdecomposition1b}
\begin{aligned}
\hat B =\big(\phi_{2n-1}(e^{i \hat \theta_1}), \cdots,\phi_{2n-1}(e^{i\hat \theta_n} ) \big),\ B=\big(\phi_{2n-1}(e^{i \theta_1}), \cdots,\phi_{2n-1}(e^{i \theta_n})\big), 
\end{aligned}
\end{equation}
with $\theta_j = y_j r\Omega, \hat \theta_j = \hat y_j r\Omega$. It is clear that 
$$||\hat B \hat a- Ba||_{\infty}\leq  ||[\hat \mu]- [\mu] ||_\infty.$$
On the other hand, since $||[\hat \mu] -\mathbf Y||_{\infty}<\sigma$, we have $||[\hat \mu] -[\mu]||_\infty< 2\sigma$. It follows that $||\hat B \hat a- B a||_\infty < 2\sigma$,
whence we get 
\begin{equation}\label{upperboundsupportlimithm1equ2}
||\hat B \hat a- B a||_2 \leq \sqrt{2n}||\hat B \hat a- B a||_{\infty}<2\sqrt{2n}\sigma.
\end{equation}
In view of (\ref{equ:distributionequ2}), we can apply Theorem \ref{spaceapproxlowerbound2} to get
\begin{align}\label{upperboundsupportlimithm1equ3}
\Big|\Big|\eta_{n,n}(e^{i \theta_1},\cdots,e^{i \theta_n},e^{i \hat \theta_1},\cdots,e^{i \hat \theta_n})\Big|\Big|_{\infty}<\frac{\sqrt{2n}2^{n+1}\pi^{n-1}}{\zeta(n)(\theta_{\min})^{n-1}} \frac{\sigma}{m_{\min}}.
\end{align}

\textbf{Step 2.}
We apply Lemma \ref{lem:stablemultiproduct1} to estimate $|\hat \theta_j -\theta_j|$'s and $|\hat y_j - y_j|$'s. To do so, let  $\epsilon = \frac{2\sqrt{2n}\pi^{2n-1}}{\zeta(n)(\theta_{\min})^{n-1}} \frac{\sigma}{m_{\min}}$. It is clear that 
$||\eta_{n ,n}||_{\infty}<(\frac{2}{\pi})^n\epsilon$ and we only need to check the following condition
\begin{equation}\label{upperboundsupportlimithm1equ4}
\theta_{\min}\geq \Big(\frac{4\epsilon}{\lambda(n)}\Big)^{\frac{1}{n}}, \quad \mbox{or equivalently}\,\,\, (\theta_{\min})^n \geq \frac{4\epsilon}{\lambda(n)}.
\end{equation}
Indeed, by (\ref{equ:angleinequ3}) and the separation condition (\ref{supportlimithm0equ0}), 
\begin{equation}\label{upperboundsupportlimithm1equ-1}
\theta_{\min}\geq  \frac{11.76\pi e}{4n}\Big(\frac{\sigma}{m_{\min}}\Big)^{\frac{1}{2n-1}}\geq  \pi \Big(\frac{8\sqrt{2n}}{\lambda(n)\zeta(n)}\frac{\sigma}{m_{\min}}\Big)^{\frac{1}{2n-1}}, \end{equation}
where we have used (\ref{equ:suppestimate1}) in the last inequality.  
Then
$$
(\theta_{\min})^{2n-1}\geq \frac{\pi^{2n-1}8\sqrt{2n}}{\lambda(n)\zeta(n)}\frac{\sigma}{m_{\min}}, 
$$
whence we get (\ref{upperboundsupportlimithm1equ4}). Therefore, we can apply Lemma \ref{lem:stablemultiproduct1} to get that, after reordering $\hat \theta_j$'s,
\begin{equation} \label{equ:upperboundsupportlimithm1equ7}
\begin{aligned}
&\Big|\hat \theta_{j}-\theta_j\Big|< \frac{\theta_{\min}}{2},\\
&\Big|\hat \theta_{j}-\theta_j\Big|< \frac{\sqrt{2n}2^{n}\pi^{2n-1}}{\zeta(n)(n-2)!(\theta_{\min})^{2n-2}}\frac{\sigma}{m_{\min}} , 1\leq j\leq n.
\end{aligned}
\end{equation}
Finally, we estimate $\Big|\hat y_j - y_j\Big|$.  Since $\Big|\hat \theta_{j}-\theta_j\Big|< \frac{\theta_{\min}}{2}$, it is clear that 
$\Big|\hat y_j-y_j\Big|< \frac{d_{\min}}{2}$.
On the other hand, by (\ref{equ:angleinequ4}) 
\[
\Big|\hat y_j-y_j\Big|\leq \frac{4n}{T\Omega}\Big|\hat \theta_j -\theta_j\Big|. 
\]	
Using (\ref{equ:upperboundsupportlimithm1equ7}), (\ref{equ:angleinequ3}), and (\ref{equ:suppestimate2}), a direct calculation shows that 
\begin{align*}
\Big|\hat y_j-y_j\Big|\leq \frac{C(n)}{T\Omega} (\frac{\pi}{T\Omega d_{\min}})^{2n-2} \frac{\sigma}{m_{\min}}, 
\end{align*}
where $C(n)=n2^{6n-3}e^{2n}\pi^{-\frac{1}{2}}$. This completes the proof.
\end{proof}

\subsection{Proof of Theorem \ref{thm:trackvelocitiesnumberlimit0}}
\begin{proof} Note that for the $n$ different $\vect v_j$'s, we have at most $\frac{n(n-1)}{2}$ different $\vect u_{pq}= \tau \vect v_p -\tau \vect v_q, p<q$. By Lemma \ref{lem:spaceangleestimate5}, we can find a one-dimensional subspace $S$ so that 
\[
\btwonorm{\mathcal P_S(\vect u_{pq})} \geq \frac{||\vect{u}_{pq}||_2}{(\pi/2)^{d-1}}\Big(\frac{\pi}{n(n-1)}\Big)^{\xi(d-1)}\geq \frac{8.8e\pi}{T\Omega}\Big(\frac{\sigma}{m_{\min}}\Big)^{\frac{1}{2n-2}},
\]
where the last inequality is by the separation condition in the theorem. Thus we can find a unit vector $\vect v\in \mathbb R^{d}$, so that 
\begin{equation}\label{equ:prooftrackvelocitynumberequ1}
\babs{\vect v^\top\vect u_{pq}}\geq \frac{8.8e\pi}{T\Omega}\Big(\frac{\sigma}{m_{\min}}\Big)^{\frac{1}{2n-2}}.
\end{equation}
We consider the measurements at $\vect \omega = \Omega \vect v$ that  
\[
\vect Y_t(\vect v\Omega) = \sum_{j=1}^n a_je^{i(\vect y_j^\top+t\tau \vect v_j^\top)\vect v \Omega} + \vect W_t(\vect v \Omega), \quad t=0, \cdots, T. 
\]
Let $b_j = a_j e^{i\vect y_j^\top \vect v \Omega}$ and $y_j =\tau  \vect v_j^\top \vect v$ and $\vect W(t) = \vect W_t(\vect v \Omega)$, the measurement can be written as 
\[
\vect Y(t) = \sum_{j=1}^n b_j e^{iy_j \Omega t} + \vect W(t) \quad  t=0, \cdots, T,
\]	
with $|\vect W(t)|< \sigma$. Note also that (\ref{equ:prooftrackvelocitynumberequ1}) implies $\min_{p\neq q} |y_p-y_q| \geq \frac{8.8e\pi}{T\Omega}\Big(\frac{\sigma}{m_{\min}}\Big)^{\frac{1}{2n-2}}$ and $\tau \vect v_j\in B_{\frac{(n-1)\pi}{T\Omega}}^d(\vect 0)$ implies $y_j\in [-\frac{(n-1)\pi}{T\Omega}, \frac{(n-1)\pi}{T\Omega}]$. Thus by Theorem \ref{thm:onednumberlimit1}, there is no $k<n$ $\hat y_j$'s so that 
\[
|\sum_{j=1}^n \hat b_j e^{i\hat y_j \Omega t}- \vect Y(t)|<\sigma, \quad t =0, \cdots, T. 
\]
This implies there does not exist any $\sigma$-admissible parameter set of \,$\mathbf Y_t$'s with less than $n$ elements.
\end{proof}

\subsection{Proof of Theorem \ref{thm:trackvelocitiessupplimit0}}
\begin{proof}
\textbf{Step 1.} Note that we have at most $\frac{n(n-1)}{2}$ different $\vect u_{pq} = \tau \vect v_p -\tau \vect v_q, p<q$. The separation condition in this theorem means
\[
\min_{p<q}\btwonorm{\vect u_{pq}} \geq \frac{11.76e \pi 4^{d-1}\Big((n+2)(n+1)/2\Big)^{\xi(d-1)}}{T\Omega} (\frac{\sigma}{m_{\min}})^{\frac{1}{2n-1}}.
\]
For a subspace $S$ of $\mathbb R^d$ and $\vect v\in S$, we denote $\vect v^{\perp}(S)$ the orthogonal complement space of $\vect v$ in $S$. Let 
\begin{align}\label{equ:proofvelosuppequ1}
&\theta_{k}=\frac{\pi}{4}\Big(\frac{2}{(n+1)(n+2)}\Big)^{\frac{1}{k}}, \ k=1, \cdots, d-1,\nonumber \\
&d_{\min, k} = \frac{11.76\pi e 4^{k-1} ((n+2)(n+1)/2)^{\xi(k-1)}}{T\Omega }\Big(\frac{\sigma}{m_{\min}}\Big)^{\frac{1}{2n-1}}, \ k=1, \cdots, d-1.
\end{align}
By Lemma \ref{lem:spaceangleestimate7}, in $\mathbb R^d$ we can construct $n+1$ $(d-1)$-dimensional spaces $S_{j_{d-1}}^{d-1}= \vect v_{j_{d-1}}^\perp(\mathbb R^d), j_{d-1}=1,\cdots,n+1$ with $0\leq \vect v_{p_{d-1}}^\top\vect v_{j_{d-1}}\leq \cos(\theta_{d-1}), p_{d-1}\neq j_{d-1}$, so that for all $\vect u_{pq}$'s we have 
\[
\btwonorm{\mathcal P_{S_{j_{d-1}}^{d-1}}(\vect u_{pq})} \geq \frac{||\vect{u}_{pq}||_2}{4}\Big(\frac{2}{(n+1)(n+2)}\Big)^{\frac{1}{d-1}}\geq \frac{11.76e \pi 4^{d-2}\Big((n+2)(n+1)/2\Big)^{\xi(d-2)}}{T\Omega} (\frac{\sigma}{m_{\min}})^{\frac{1}{2n-1}}.
\]
Since each $S_{j_{d-1}}^{d-1}$ is a $(d-1)$-dimensional space, applying Lemma \ref{lem:spaceangleestimate7} again for $S_{j_{d-1}}^{d-1}$ and $\mathcal P_{S_{j_{d-1}}^{d-1}}(\vect u_{pq})$'s, for each $S_{j_{d-1}}^{d-1}$ we can construct $n+1$ $(d-2)$-dimensional subspaces $S_{j_{d-2},j_{d-1}}^{d-2} = \vect v_{j_{d-2},j_{d-1}}^\perp(S_{j_{d-1}}^{d-1}),j_{d-2}=1, \cdots, n+1$ with $0\leq \vect v_{p_{d-2}, j_{d-1}}^\top\vect v_{j_{d-2}, j_{d-1}}\leq \cos(\theta_{d-2}), p_{d-2}\neq j_{d-2}$, so that for all $\vect u_{pq}$'s we have 
\begin{align*}
\btwonorm{\mathcal P_{S_{j_{d-2},j_{d-1}}^{d-2}}(\vect u_{pq})} = &\btwonorm{\mathcal P_{S_{j_{d-2},j_{d-1}}^{d-2}}\Big(\mathcal P_{S_{j_{d-1}}^{d-1}}(\vect u_{pq})\Big)} \geq  \frac{||\vect{u}_{pq}||_2}{4^2}\Big(\frac{2}{(n+1)(n+2)}\Big)^{\frac{1}{d-1}+\frac{1}{d-2}}\\
\geq & \frac{11.76e \pi 4^{d-3}\Big((n+2)(n+1)/2\Big)^{\xi(d-3)}}{T\Omega} (\frac{\sigma}{m_{\min}})^{\frac{1}{2n-1}},
\end{align*}
where the first equality is because $S_{j_{d-2},j_{d-1}}^{d-2}$ is a subspace of $S_{j_{d-1}}^{d-1}$. We can continue the process and construct $n+1$ subspaces of $S_{j_{d-2},j_{d-1}}^{d-2}, S_{j_{d-3}, j_{d-2},j_{d-1}}^{d-3}, \cdots, S_{j_{2}, j_{3},\cdots, j_{d-1}}^{2}$, respectively. For each two dimensional space $S_{j_{2}, j_{3}, \cdots, j_{d-1}}^2$, we can construct $n+1$ one-dimensional subspaces $S_{j_1, j_{2}, j_{3}, \cdots, j_{d-1}}^1= \vect v_{j_1, j_{2}, j_{3}, \cdots, j_{d-1}}^{\perp}(S_{j_{2}, j_{3}, \cdots, j_{d-1}}^2), j_1=1, \cdots, n+1$ with $0\leq \vect v_{p_{1},j_2, \cdots, j_{d-1}}^\top\vect v_{j_1, j_2, \cdots, j_{d-1}}\leq \cos(\theta_{1}), p_{1}\neq j_{1}$, so that for all $\vect u_{pq}$'s we have 
\[
\btwonorm{\mathcal P_{S_{j_1, j_{2}, j_{3}, \cdots, j_{d-1}}^{1}}(\vect u_{pq})} \geq \frac{||\vect{u}_{pq}||_2}{4^{d-1}}\Big(\frac{2}{(n+1)(n+2)}\Big)^{\xi(d-1)}\geq \frac{11.76e\pi}{T\Omega}\Big(\frac{\sigma}{m_{\min}}\Big)^{\frac{1}{2n-1}}. 
\]
Thus for each $\{j_2, \cdots, j_{d-1}\}$, we can find $n+1$ unit vectors $\vect q_j$'s with $0\leq |\vect q_{p}^\top\vect q_{j}|\leq \cos(\theta_{1}), p\neq j$, so that 
\begin{equation}\label{equ:Proofvelovalueequ1}
\mathcal P_{S_{j_1, j_{2}, j_{3}, \cdots, j_{d-1}}^{1}}(\vect u_{pq}) = \vect q_j^\top \vect u_{pq}, \text{\ and \ }
\babs{\vect q_j^\top \vect u_{pq}}\geq \frac{11.76e\pi}{T\Omega}\Big(\frac{\sigma}{m_{\min}}\Big)^{\frac{1}{2n-1}}.
\end{equation}
We consider these $\vect q_j$'s in Step 2. 

\textbf{Step 2.}
Without loss of generality, we first consider $\vect q_1$. We consider the measurements at $\vect \omega = \Omega \vect q_1$ that  
\[
\vect Y_t(\vect q_1 \Omega) = \sum_{j=1}^n a_je^{i(\vect y_j^\top+t\tau \vect v_j^\top)\vect q_1 \Omega} + \vect W_t(\vect q_1 \Omega), \quad t=0, \cdots, T. 
\]
Let $b_j = a_j e^{i\vect y_j^\top \vect q_1 \Omega}$, $y_j = \tau \vect v_j^\top \vect q_1$ and $\vect W(t) = \vect W_t(\vect q_1 \Omega)$, the measurement can be written as 
\[
\vect Y(t) = \sum_{j=1}^n b_j e^{iy_j \Omega t} + \vect W(t), \quad  t=0, \cdots, T,
\]	
with $|\vect W(t)|< \sigma$. By (\ref{equ:Proofvelovalueequ1}) we have 
\[
\min_{p\neq q} \babs{y_p-y_q} \geq \frac{11.76e\pi}{T\Omega}\Big(\frac{\sigma}{m_{\min}}\Big)^{\frac{1}{2n-1}}.
\]
On the other hand, the measurement constraint  
\[
\babs{\sum_{j=1}^n \hat a_je^{i(\vect {\hat y}_j^\top + t\tau \vect {\hat v}_j^\top)\vect q_1\Omega}-\vect Y_t(\vect q_1\Omega)}<\sigma, 
\]
can be written as 
\[
\babs{\sum_{j=1}^n \hat b_j e^{i \hat y_j t}-\vect Y(t)}<\sigma, 
\]
where $\hat b_j = \hat a_j e^{i\vect {\hat y}_j^\top \vect q_1 \Omega}$ and $\hat y_j = \tau \vect {\hat v}_j^\top \vect q_1$. Note that $|\tau \vect v_j^\top\vect q_1| \leq \frac{(n-1)\pi}{T\Omega}$ and $|\tau \vect {\hat v}_j^\top\vect q_1| \leq \frac{(n-1)\pi}{T\Omega}$ since $\tau \vect {\hat v}_j, \tau \vect v_j \in B_{\frac{(n-1)\pi}{T\Omega}}^d(\vect 0)$. By Theorem \ref{thm:onedsupplimit1}, we have after reordering $\hat y_j$'s, 
\[
\Big|\hat y_j - y_j\Big|< \frac{1}{2}d_{\min,1},
\]
where $d_{\min, 1}$ is defined as in (\ref{equ:proofvelosuppequ1}), and 
\[
\Big|\hat y_j - y_j\Big|\leq  \frac{C(1,n)}{T\Omega}\Big(\frac{\pi}{d_{\min,1}T\Omega}\Big)^{2n-2}\frac{\sigma}{m_{\min}},
\]
where $C(1,n)=n2^{6n-3}e^{2n}\pi^{-\frac{1}{2}}$. Because we have $n+1$ $\vect q_j$'s, we have $n+1$ permutations $\tau_j$'s of $\{1, \cdots, n\}$ so that 
\[
\babs{\tau \vect {\hat v}_{\tau_{j}(p)}^\top\vect q_j- \tau \vect v_p^\top\vect q_j}< \frac{1}{2}d_{\min,1}, \quad j =1, \cdots, n+1,
\]
and 
\[
\babs{\tau \vect {\hat v}_{\tau_{j}(p)}^\top\vect q_j- \tau \vect v_p^\top\vect q_j}\leq \frac{C(1,n)}{T\Omega}\Big(\frac{\pi}{d_{\min,1}T\Omega}\Big)^{2n-2}\frac{\sigma}{m_{\min}}, \quad j =1, \cdots, n+1.
\]
Since we have $n+1$ $\vect q_j$'s but only $n$ $\vect {\hat v}_p$'s, by the pigeonhole principle, for each $\vect v_p$, there exist $\vect {\hat v}_{\tau_{j_1}(p)}=\vect {\hat v}_{\tau_{j_2}(p)} = \vect {\hat v}_{p'}$ so that 
\[
\babs{\tau \vect {\hat v}_{p'}^\top\vect q_{j_t}- \tau \vect v_p^\top\vect q_{j_t}}< \frac{1}{2}d_{\min,1}, \ t=1,2,
\]
and 
\[
\babs{\tau \vect {\hat v}_{p'}^\top\vect q_{j_t}- \tau \vect v_p^\top\vect q_{j_t}}\leq  \frac{C(1,n)}{T\Omega}\Big(\frac{\pi}{d_{\min,1}T\Omega}\Big)^{2n-2}\frac{\sigma}{m_{\min}}, \ t=1,2. 
\]
Thus we have two $S_{{j_1}_t,j_2, \cdots, j_{d-1}}^1, t=1,2$, so that 
\[
\babs{\mathcal P_{S_{{j_1}_t,j_2, \cdots, j_{d-1}}^1}(\tau \vect {\hat v}_{p'}- \tau \vect v_p)}< \frac{1}{2}d_{\min,1}, \ t=1,2,
\]
and 
\[
\babs{\mathcal P_{S_{{j_1}_t,j_2, \cdots, j_{d-1}}^1}(\tau \vect {\hat v}_{p'}- \tau \vect v_p)}\leq  \frac{C(1,n)}{T\Omega}\Big(\frac{\pi}{d_{\min,1}T\Omega}\Big)^{2n-2}\frac{\sigma}{m_{\min}}, \ t=1,2. 
\]
From the results obtained in Step 1, $S_{{j_1}_t,j_2, \cdots, j_{d-1}}^1=\vect v_{{j_1}_t, j_{2}, j_{3}, \cdots, j_{d-1}}^{\perp}(S_{j_{2}, j_{3}, \cdots, j_{d-1}}^2), t=1,2,$ are both one-dimensional subspaces of $S_{j_2, \cdots, j_{d-1}}^2$ and $0\leq \vect v_{{j_1}_1,j_2, \cdots, j_{d-1}}^\top\vect v_{{j_1}_2, j_2, \cdots, j_{d-1}}\leq \cos(\theta_{1})$. By Lemma \ref{lem:spaceangleestimate6}, we have 
\[
\btwonorm{\mathcal P_{S_{j_2, \cdots, j_{d-1}}^2}(\tau \vect {\hat v}_{p'} - \tau \vect v_p) } < \frac{1}{2} d_{\min, 2}, \quad  \btwonorm{\mathcal P_{S_{j_2, \cdots, j_{d-1}}^2}(\tau \vect {\hat v}_{p'} - \tau \vect v_p) }\leq \frac{C(2, n)}{T\Omega}\Big(\frac{\pi}{d_{\min,2}T\Omega}\Big)^{2n-2}\frac{\sigma}{m_{\min}},
\]
where 
\[
C(2, n)=\big(4((n+2)(n+1)/2)^{\xi(1)}\big)^{(2n-1)}n2^{6n-3}e^{2n}\pi^{-\frac{1}{2}}.
\]
Since we do not specify the $\{j_2, \cdots, j_{d-1}\}$, for fixed $\{j_3, \cdots, j_{d-1}\}$, the above results hold for all $\{j_2, j_3, \cdots, j_{d-1}\}$, $j_2=1, \cdots, n+1$ with that the $p'$ is related to $j_2$. Again from Step 1, for fixed $\{j_3, \cdots, j_{d-1}\}$ we have $S_{j_{2}, j_{3}, \cdots, j_{d-1}}^2= \vect v_{j_{2}, j_{3}, \cdots, j_{d-1}}^{\perp}(S_{j_{3}, \cdots, j_{d-1}}^3), j_2=1, \cdots, n+1$ with $0\leq \vect v_{p_2, j_3, \cdots, j_{d-1}}^\top\vect v_{j_2, j_3, \cdots, j_{d-1}}\leq \cos(\theta_{2}), p_{2}\neq j_{2}$. Similar to the above arguments, since we have $n+1$ $j_2$'s while only $n$ $\vect{\hat v}_q$'s, by the pigeonhole principle, for each $\vect v_p$ we can find $\vect {\hat v}_{p'}$ so that   
\begin{align*}
&\btwonorm{\mathcal P_{S_{{j_2}_t, \cdots, j_{d-1}}^2}(\tau \vect {\hat v}_{p'} - \tau \vect v_p) } < \frac{1}{2} d_{\min, 2}, \\ 
&\btwonorm{\mathcal P_{S_{{j_2}_t, \cdots, j_{d-1}}^2}(\tau \vect {\hat v}_{p'} - \tau \vect v_p) }\leq \frac{C(2, n)}{T\Omega}\Big(\frac{\pi}{d_{\min,2}T\Omega}\Big)^{2n-2}\frac{\sigma}{m_{\min}}, \ t=1,2,
\end{align*}
where 
\[
C(2, n)=\big(4((n+2)(n+1)/2)^{\xi(1)}\big)^{(2n-1)}n2^{6n-3}e^{2n}\pi^{-\frac{1}{2}}.
\]
By Lemma \ref{lem:spaceangleestimate6}, we have 
\[
\btwonorm{\mathcal P_{S_{j_3, \cdots, j_{d-1}}^3}(\tau \vect {\hat v}_{p'} - \tau  \vect v_p) } < \frac{1}{2} d_{\min, 3}, \quad  \btwonorm{\mathcal P_{S_{j_3, \cdots, j_{d-1}}^3}(\tau \vect {\hat v}_{p'} - \tau \vect v_p) }\leq \frac{C(3, n)}{T\Omega}\Big(\frac{\pi}{d_{\min,3}T\Omega}\Big)^{2n-2}\frac{\sigma}{m_{\min}},
\]
where 
\[
C(3, n)=\big(4^2((n+2)(n+1)/2)^{\xi(2)}\big)^{(2n-1)}n2^{6n-3}e^{2n}\pi^{-\frac{1}{2}}.
\]
Thus, by continuing the process, there exists $\vect {\hat v}_{p'}$ so that
\[
\btwonorm{\tau \vect {\hat v}_{p'} - \tau \vect v_p} < \frac{1}{2} d_{\min, d}, \quad  \btwonorm{\tau \vect {\hat v}_{p'} - \tau \vect v_p}\leq \frac{C(d, n)}{T\Omega}\Big(\frac{\pi}{d_{\min,d}T\Omega}\Big)^{2n-2}\frac{\sigma}{m_{\min}},
\]
where 
\[
C(d, n)=\big(4^{d-1}((n+2)(n+1)/2)^{\xi(d-1)}\big)^{(2n-1)}n2^{6n-3}e^{2n}\pi^{-\frac{1}{2}}.
\]
Since $\min_{p\neq q}||\vect v_p-\vect v_q||_2\geq d_{\min, d}$, for each $\vect v_p$ there exists one and only one $\vect {\hat v}_{p'}$ satisfying the above condition. Thus after reordering $\vect {\hat v}_{p}$'s we have
\[
\btwonorm{\tau \vect {\hat v}_{p} - \tau \vect v_p} < \frac{1}{2} d_{\min, d}, \  \btwonorm{\tau \vect {\hat v}_{p} - \tau \vect v_p}\leq \frac{C(d, n)}{T\Omega}\Big(\frac{\pi}{d_{\min,d}T\Omega}\Big)^{2n-2}\frac{\sigma}{m_{\min}}, \ 1\leq p\leq n.
\]
This completes the proof. 
\end{proof}

\bibliographystyle{plain}
\bibliography{references_final}	

\begin{thebibliography}{10}

\bibitem{akaike1998information}
Hirotogu Akaike.
\newblock Information theory and an extension of the maximum likelihood
  principle.
\newblock In {\em Selected papers of hirotugu akaike}, pages 199--213.
  Springer, 1998.

\bibitem{akaike1974new}
Hirotugu Akaike.
\newblock A new look at the statistical model identification.
\newblock In {\em Selected Papers of Hirotugu Akaike}, pages 215--222.
  Springer, 1974.

\bibitem{akinshin2015accuracy}
Andrey Akinshin, Dmitry Batenkov, and Yosef Yomdin.
\newblock Accuracy of spike-train fourier reconstruction for colliding nodes.
\newblock In {\em 2015 International Conference on Sampling Theory and
  Applications (SampTA)}, pages 617--621. IEEE, 2015.

\bibitem{alberti2019dynamic}
Giovanni~S Alberti, Habib Ammari, Francisco Romero, and Timoth{\'e}e Wintz.
\newblock Dynamic spike superresolution and applications to ultrafast
  ultrasound imaging.
\newblock {\em SIAM Journal on Imaging Sciences}, 12(3):1501--1527, 2019.

\bibitem{azais2015spike}
Jean-Marc Azais, Yohann De~Castro, and Fabrice Gamboa.
\newblock Spike detection from inaccurate samplings.
\newblock {\em Applied and Computational Harmonic Analysis}, 38(2):177--195,
  2015.

\bibitem{batenkov2020conditioning}
Dmitry Batenkov, Laurent Demanet, Gil Goldman, and Yosef Yomdin.
\newblock Conditioning of partial nonuniform fourier matrices with clustered
  nodes.
\newblock {\em SIAM Journal on Matrix Analysis and Applications},
  41(1):199--220, 2020.

\bibitem{batenkov2019super}
Dmitry Batenkov, Gil Goldman, and Yosef Yomdin.
\newblock {Super-resolution of near-colliding point sources}.
\newblock {\em Information and Inference: A Journal of the IMA}, 05 2020.
\newblock iaaa005.

\bibitem{PALM}
E.~Betzig, G.H. Patterson, R.~Sougrat, O.W. Lindwasser, S.~Olenych, J.S.
  Bonifacino, M.W. Davidson, J.~Lippincott-Schwartz, and H.F. Hess.
\newblock Imaging intracellular fluorescent proteins at nanometer resolution.
\newblock {\em Science}, 313:1642--1645, 2006.

\bibitem{ref8}
E.~Betzig, G.H. Patterson, R.~Sougrat, O.W. Lindwasser, S.~Olenych, J.S.
  Bonifacino, M.W. Davidson, J.~Lippincott-Schwartz, and H.F. Hess.
\newblock Imaging intracellular fluorescent proteins at nanometer resolution.
\newblock {\em Science}, 313:1642--1645, 2006.

\bibitem{candes2014towards}
Emmanuel~J. Cand{\`e}s and Carlos Fernandez-Granda.
\newblock Towards a mathematical theory of super-resolution.
\newblock {\em Communications on Pure and Applied Mathematics}, 67(6):906--956,
  2014.

\bibitem{chen2020algorithmic}
Sitan Chen and Ankur Moitra.
\newblock Algorithmic foundations for the diffraction limit.
\newblock {\em arXiv preprint arXiv:2004.07659}, 2020.

\bibitem{chen1991detection}
Weiguo Chen, Kon~Max Wong, and James~P Reilly.
\newblock Detection of the number of signals: A predicted eigen-threshold
  approach.
\newblock {\em IEEE Transactions on Signal Processing}, 39(5):1088--1098, 1991.

\bibitem{chen1992direction}
Y-H Chen and C-H Chen.
\newblock Direction-of-arrival and frequency estimations for narrowband sources
  using two single rotation invariance algorithms with the marked subspace.
\newblock In {\em IEE Proceedings F (Radar and Signal Processing)}, volume 139,
  pages 297--300. IET, 1992.

\bibitem{couture2011microbubble}
Olivier Couture, Benoit Besson, Gabriel Montaldo, Mathias Fink, and Mickael
  Tanter.
\newblock Microbubble ultrasound super-localization imaging (musli).
\newblock In {\em 2011 IEEE International Ultrasonics Symposium}, pages
  1285--1287. IEEE, 2011.

\bibitem{couture2018ultrasound}
Olivier Couture, Vincent Hingot, Baptiste Heiles, Pauline Muleki-Seya, and
  Mickael Tanter.
\newblock Ultrasound localization microscopy and super-resolution: A state of
  the art.
\newblock {\em IEEE transactions on ultrasonics, ferroelectrics, and frequency
  control}, 65(8):1304--1320, 2018.

\bibitem{dabiri2020particle}
Dana Dabiri and Charles Pecora.
\newblock {\em Particle tracking velocimetry}.
\newblock IOP Publishing Bristol, 2020.

\bibitem{demanet2015recoverability}
Laurent Demanet and Nam Nguyen.
\newblock The recoverability limit for superresolution via sparsity.
\newblock {\em arXiv preprint arXiv:1502.01385}, 2015.

\bibitem{demene2021transcranial}
Charlie Demen{\'e}, Justine Robin, Alexandre Dizeux, Baptiste Heiles, Mathieu
  Pernot, Mickael Tanter, and Fabienne Perren.
\newblock Transcranial ultrafast ultrasound localization microscopy of brain
  vasculature in patients.
\newblock {\em Nature biomedical engineering}, 5(3):219--228, 2021.

\bibitem{denoyelle2017support}
Quentin Denoyelle, Vincent Duval, and Gabriel Peyr{\'e}.
\newblock Support recovery for sparse super-resolution of positive measures.
\newblock {\em Journal of Fourier Analysis and Applications}, 23(5):1153--1194,
  2017.

\bibitem{desailly2013sono}
Yann Desailly, Olivier Couture, Mathias Fink, and Mickael Tanter.
\newblock Sono-activated ultrasound localization microscopy.
\newblock {\em Applied Physics Letters}, 103(17):174107, 2013.

\bibitem{donoho1992superresolution}
David~L. Donoho.
\newblock Superresolution via sparsity constraints.
\newblock {\em SIAM journal on mathematical analysis}, 23(5):1309--1331, 1992.

\bibitem{duval2015exact}
Vincent Duval and Gabriel Peyr{\'e}.
\newblock Exact support recovery for sparse spikes deconvolution.
\newblock {\em Foundations of Computational Mathematics}, 15(5):1315--1355,
  2015.

\bibitem{errico2015ultrafast}
Claudia Errico, Juliette Pierre, Sophie Pezet, Yann Desailly, Zsolt Lenkei,
  Olivier Couture, and Mickael Tanter.
\newblock Ultrafast ultrasound localization microscopy for deep
  super-resolution vascular imaging.
\newblock {\em Nature}, 527(7579):499--502, 2015.

\bibitem{han2013improved}
Keyong Han and Arye Nehorai.
\newblock Improved source number detection and direction estimation with nested
  arrays and ulas using jackknifing.
\newblock {\em IEEE Transactions on Signal Processing}, 61(23):6118--6128,
  2013.

\bibitem{he2010detecting}
Zhaoshui He, Andrzej Cichocki, Shengli Xie, and Kyuwan Choi.
\newblock Detecting the number of clusters in n-way probabilistic clustering.
\newblock {\em IEEE Transactions on Pattern Analysis and Machine Intelligence},
  32(11):2006--2021, 2010.

\bibitem{ref4}
S.W. Hell and J.~Wichmann.
\newblock Breaking the diffraction resolution limit by stimulated emission:
  stimulated-emission-depletion fluorescence microscopy.
\newblock {\em Opt. Lett.}, 19:780--782, 1994.

\bibitem{ref7}
S.T. Hess, T.P.K. Girirajan, and M.D. Mason.
\newblock Ultra-high resolution imaging by fluorescence photoactivation
  localization microscopy.
\newblock {\em Biophys. J.}, 91:4258--4272, 2006.

\bibitem{johnson1991operational}
Richard~L Johnson and Gina~E Miner.
\newblock {An operational system implementation of the ESPIRIT DF algorithm}.
\newblock {\em IEEE transactions on aerospace and electronic systems},
  27(1):159--166, 1991.

\bibitem{kwik2003membrane}
Jeanne Kwik, Sarah Boyle, David Fooksman, Leonid Margolis, Michael~P Sheetz,
  and Michael Edidin.
\newblock Membrane cholesterol, lateral mobility, and the phosphatidylinositol
  4, 5-bisphosphate-dependent organization of cell actin.
\newblock {\em Proceedings of the National Academy of Sciences},
  100(24):13964--13969, 2003.

\bibitem{lawley1956tests}
DN~Lawley.
\newblock Tests of significance for the latent roots of covariance and
  correlation matrices.
\newblock {\em biometrika}, 43(1/2):128--136, 1956.

\bibitem{li2021stable}
Weilin Li and Wenjing Liao.
\newblock Stable super-resolution limit and smallest singular value of
  restricted fourier matrices.
\newblock {\em Applied and Computational Harmonic Analysis}, 51:118--156, 2021.

\bibitem{liao2016music}
Wenjing Liao and Albert~C. Fannjiang.
\newblock {MUSIC} for single-snapshot spectral estimation: Stability and
  super-resolution.
\newblock {\em Applied and Computational Harmonic Analysis}, 40(1):33--67,
  2016.

\bibitem{liu2022mathematical}
Ping Liu, Sanghyeon Yu, Ola Sabet, Lucas Pelkmans, and Habib Ammari.
\newblock Mathematical foundation of sparsity-based multi-illumination
  super-resolution.
\newblock {\em arXiv preprint arXiv:2202.11189}, 2022.

\bibitem{liu2021mathematicalhighd}
Ping Liu and Hai Zhang.
\newblock A mathematical theory of computational resolution limit in
  multi-dimensional spaces.
\newblock {\em Inverse Problems}, 37(10):104001, 2021.

\bibitem{liu2021theorylse}
Ping Liu and Hai Zhang.
\newblock A theory of computational resolution limit for line spectral
  estimation.
\newblock {\em IEEE Transactions on Information Theory}, 67(7):4812--4827,
  2021.

\bibitem{liu2021mathematicaloned}
Ping Liu and Hai Zhang.
\newblock A mathematical theory of computational resolution limit in one
  dimension.
\newblock {\em Applied and Computational Harmonic Analysis}, 56:402--446, 2022.

\bibitem{liu2022measurement}
Ping Liu and Hai Zhang.
\newblock A measurement decoupling based fast algorithm for super-resolving
  point sources with multi-cluster structure.
\newblock {\em arXiv preprint arXiv:2204.00469}, 2022.

\bibitem{meijering2009tracking}
Erik Meijering, Oleh Dzyubachyk, Ihor Smal, and Wiggert~A van Cappellen.
\newblock Tracking in cell and developmental biology.
\newblock In {\em Seminars in cell \& developmental biology}, volume~20, pages
  894--902. Elsevier, 2009.

\bibitem{meijering2008time}
Erik Meijering, Ihor Smal, Oleh Dzyubachyk, and Jean-Christophe Olivo-Marin.
\newblock Time-lapse imaging.
\newblock {\em Microscope Image Processing}, 401:440, 2008.

\bibitem{moitra2015super}
Ankur Moitra.
\newblock Super-resolution, extremal functions and the condition number of
  vandermonde matrices.
\newblock In {\em Proceedings of the Forty-seventh Annual ACM Symposium on
  Theory of Computing}, STOC '15, pages 821--830, New York, NY, USA, 2015. ACM.

\bibitem{morgenshtern2020super}
Veniamin~I Morgenshtern.
\newblock Super-resolution of positive sources on an arbitrarily fine grid.
\newblock {\em arXiv preprint arXiv:2005.06756}, 2020.

\bibitem{morgenshtern2016super}
Veniamin~I. Morgenshtern and Emmanuel~J. Candes.
\newblock Super-resolution of positive sources: The discrete setup.
\newblock {\em SIAM Journal on Imaging Sciences}, 9(1):412--444, 2016.

\bibitem{ram2008high}
Sripad Ram, Prashant Prabhat, Jerry Chao, E~Sally Ward, and Raimund~J Ober.
\newblock High accuracy 3d quantum dot tracking with multifocal plane
  microscopy for the study of fast intracellular dynamics in live cells.
\newblock {\em Biophysical journal}, 95(12):6025--6043, 2008.

\bibitem{rissanen1978modeling}
Jorma Rissanen.
\newblock Modeling by shortest data description.
\newblock {\em Automatica}, 14(5):465--471, 1978.

\bibitem{ref9}
M.J. Rust, M.~Bates, and X.~Zhuang.
\newblock Sub-diffraction-limit imaging by stochastic optical reconstruction
  microscopy (storm).
\newblock {\em Nat. Methods}, 3:793--796, 2006.

\bibitem{schiffmann2006open}
David~A Schiffmann, Dina Dikovskaya, Paul~L Appleton, Ian~P Newton, Douglas~A
  Creager, Chris Allan, Inke~S N{\"a}thke, and Ilya~G Goldberg.
\newblock Open microscopy environment and findspots: integrating image
  informatics with quantitative multidimensional image analysis.
\newblock {\em Biotechniques}, 41(2):199--208, 2006.

\bibitem{schmidt1986multiple}
Ralph Schmidt.
\newblock Multiple emitter location and signal parameter estimation.
\newblock {\em IEEE transactions on antennas and propagation}, 34(3):276--280,
  1986.

\bibitem{schwarz1978estimating}
Gideon Schwarz et~al.
\newblock Estimating the dimension of a model.
\newblock {\em The annals of statistics}, 6(2):461--464, 1978.

\bibitem{stoica1989music}
Petre Stoica and Arye Nehorai.
\newblock {MUSIC}, maximum likelihood, and {Cramer-Rao} bound.
\newblock {\em IEEE Transactions on Acoustics, speech, and signal processing},
  37(5):720--741, 1989.

\bibitem{tang2014near}
Gongguo Tang, Badri~Narayan Bhaskar, and Benjamin Recht.
\newblock Near minimax line spectral estimation.
\newblock {\em IEEE Transactions on Information Theory}, 61(1):499--512, 2014.

\bibitem{STED}
V.~VOLKER~Westphalsilvio, O.~Rizzolimarcel, A.~Lauterbachdirk,
  J.~Kaminereinhard, and S.W. Hell.
\newblock Video-rate far-field optical nanoscopy dissects synaptic vesicle
  movementvideo-rate far-field optical nanoscopy dissects synaptic vesicle
  movement.
\newblock {\em Science}, 320:246--249, 2008.

\bibitem{wax1985detection}
Mati Wax and Thomas Kailath.
\newblock Detection of signals by information theoretic criteria.
\newblock {\em IEEE Transactions on acoustics, speech, and signal processing},
  33(2):387--392, 1985.

\bibitem{wax1989detection}
Mati Wax and Ilan Ziskind.
\newblock Detection of the number of coherent signals by the mdl principle.
\newblock {\em IEEE Transactions on Acoustics, Speech, and Signal Processing},
  37(8):1190--1196, 1989.

\bibitem{yilmazer2006matrix}
Nuri Yilmazer, Raul Fernandez-Recio, and Tapan~K Sarkar.
\newblock Matrix pencil method for simultaneously estimating azimuth and
  elevation angles of arrival along with the frequency of the incoming signals.
\newblock {\em Digital Signal Processing}, 16(6):796--816, 2006.

\bibitem{zenisek2000transport}
D~Zenisek, JA~Steyer, and W~Almers.
\newblock Transport, capture and exocytosis of single synaptic vesicles at
  active zones.
\newblock {\em Nature}, 406(6798):849--854, 2000.

\bibitem{zoltowski1989sensor}
Michael~D Zoltowski and Demosthenis Stavrinides.
\newblock Sensor array signal processing via a procrustes rotations based
  eigenanalysis of the esprit data pencil.
\newblock {\em IEEE transactions on acoustics, speech, and signal processing},
  37(6):832--861, 1989.

\end{thebibliography}

\end{document}